\documentclass[12pt]{article}
\usepackage{amsmath,amsthm,amssymb,mathtools}
\usepackage{graphicx,epstopdf,framed,color}
\usepackage{euscript,paralist,wrapfig,}
\usepackage[colorlinks=true, linkcolor=blue, urlcolor=blue, citecolor=blue]{hyperref}
\usepackage[T1]{fontenc}
\usepackage[utf8]{inputenc}
\usepackage{tikz}

\theoremstyle{remark}
\newtheorem{theorem}{Theorem}[section]

\newtheorem{lemma}[theorem]{Lemma}
\newtheorem{assumption}[theorem]{Assumption}

\newtheorem{proposition}[theorem]{Proposition}

\newtheorem{remark}[theorem]{Remark}

\newtheorem{problem}[theorem]{Problem}

\numberwithin{equation}{section}

\newfont{\roma}{cmr10 scaled 1200}
\renewcommand{\cline}{{\mathbb C}}

\newcommand{\nline}  {{\mathbb N}}
\newcommand{\rline}  {{\mathbb R}}

\newcommand{\tline}  {{\mathbb T}}

\newcommand{\dd}   {{\rm d}\hbox{\hskip 0.5pt}}

\newcommand{\Lscr} {{\cal L}}

\newcommand{\Pscr} {{\cal P}}

\newcommand{\mm}    {{\hbox{\hskip 0.5pt}}}

\newcommand{\bluff} {{\hbox{\raise 15pt \hbox{\mm}}}}
\newcommand{\sbluff}{{\hbox{\raise  7pt \hbox{\mm}}}}

\newcommand{\g}      {{\gamma}}
\newcommand{\FORALL} {{\hbox{$\hskip 11mm \forall \;$}}}

\renewcommand{\Re}   {{\rm Re\,}}
\renewcommand{\Im}   {{\rm Im\,}}

\newcommand{\bbm}[1]{\left[\begin{matrix} #1 \end{matrix}\right]}
\newcommand{\sbm}[1]{\left[\begin{smallmatrix} #1
   \end{smallmatrix}\right]}

\begin{document}
\renewcommand{\thefootnote}{\fnsymbol{footnote}}
\renewcommand{\thefootnote}{\fnsymbol{footnote}}
\newcommand{\footremember}[2]{%
   \footnote{#2}
    \newcounter{#1}
    \setcounter{#1}{\value{footnote}}%
}
\newcommand{\footrecall}[1]{%
    \footnotemark[\value{#1}]%
}
\makeatletter
\def\blfootnote{\gdef\@thefnmark{}\@footnotetext}
\makeatother

\begin{center}
{\Large \bf LQR based stabilization of an 1D heat equation with advection and memory effects}\\[2ex]
Bhargav Pavan Kumar Sistla, Vivek Natarajan \vspace{2mm}
\blfootnote{B.P.K. Sistla is supported by the Prime Minister’s Research Fellowship via the grant RSPMRF0262.}
\blfootnote{B. P. K. Sistla (bhargav@sc.iitb.ac.in) and V. Natarajan (vivek.natarajan@iitb.ac.in) are with the Centre for Systems and Control, Indian Institute of Technology Bombay, Mumbai, India, 400076, Ph:+912225765385.}
\end{center}

\begin{center}
  {\bf Abstract}\vspace{2mm}

\parbox{5.5in}{{\noindent
We derive an one-dimensional model for heat transfer in a moving fluid that incorporates Fourier conduction, an exponentially decaying memory term, and advection under thermally insulated boundary conditions. Our objective is to numerically construct a bounded state feedback law such that, for every initial state, the closed-loop solution decays exponentially to zero with decay rate at least $\omega>0$, that is, to solve the $\omega$-stabilization problem for this model. We give an explicit description of the eigenvalues of the state operator $A$ associated with our model. A part of the eigenvalue set is shown to converge to a finite negative accumulation point. This finite accumulation point determines the intrinsic upper bound on decay rates achievable by bounded feedback. Since $A$ lacks compact resolvent, additional analysis is required to show that the spectrum is the closure of its eigenvalues and that all eigenvalues have finite algebraic multiplicity. These results are then used to verify the stabilizability properties needed for the $\omega$-stabilization problem. For every $\omega$ below the accumulation bound, we show that the $\omega$-stabilization problem has a solution provided the control operator $B$ satisfies a non-orthogonality condition. To compute stabilizing gains, we formulate an LQR problem for the system and solve it using finite-dimensional approximations. For each $n\in\mathbb{N}$, we construct finite-dimensional approximations $A_n$ and $B_n$ of $A$ and $B$ and solve the associated finite-dimensional algebraic Riccati equation to obtain a feedback gain $K_n$. We show that, for all sufficiently large $n$, the gain $K_n$ can be chosen so that all eigenvalues of $A_n+B_nK_n$ satisfy $\operatorname{Re}\lambda<-\omega$. We establish the uniform, with respect to $n$, stabilizability of the pair $(A_n+\omega I,B_n)$, which allows us to prove that, for all sufficiently large $n$, these gains solve the $\omega$-stabilization problem for the original system. We validate our theoretical results numerically using an example.
}}
\end{center}

\noindent
{\bf Keywords}. Sesquilinear form, Galerkin approximation, Hautus test, linear quadratic regulator problem, Riccati equation, uniform stabilizability \vspace{-3mm}


\section{Introduction}\label{section: Introduction}\vspace{-1mm}
\setcounter{equation}{0} 

\ \ \ Transport of thermal energy in a moving fluid is commonly modeled by a balance law in which diffusion describes conductive transport and advection describes bulk motion. In materials with memory, the heat flux need not depend only on the present temperature gradient. After linearization, Fourier's law may be replaced by a constitutive relation involving the present temperature gradient together with a convolution over its past history, see \cite{Nu:1971}. In this paper, we derive an one-dimensional heat-transfer model on $(0,1)$ that combines advective transport, instantaneous Fourier conduction and an exponentially fading memory kernel under thermally insulated boundary conditions. The temperature is actuated by a scalar input through a prescribed shape function. By introducing a memory variable, the model is rewritten as a coupled PDE--ODE system in which the memory component is driven by the current temperature and decays exponentially. This formulation leads to an abstract evolution equation on $X=L^2(0,1)\times H^1(0,1)$ with state operator $A$ and control operator $B$. The resulting system differs substantially from the classical heat equation. In particular, the associated state operator may have noncompact resolvent and a finite spectral accumulation point, making spectral analysis, stabilizability verification and LQR-based feedback synthesis more difficult. The paper follows the broad $\omega$-stabilization strategy of \cite{SiAkMiNa:2025}, but the presence of advection and thermally insulated boundary conditions leads to a more involved spectral problem and requires separate analysis.

Heat equations with memory exhibit subtle controllability phenomena. 
Approximate controllability was established in \cite{BaIa:2000} under suitable assumptions on completely monotone memory kernels, while stronger controllability objectives may fail. 
In particular, \cite{IvPa:09} shows that one-dimensional heat equations with memory are not controllable to rest for broad classes of kernels and control mechanisms. 
Further negative results were obtained in \cite{GuIm:13,HaPa:12}, where even boundary controls may fail to steer heat equations with memory to zero at a prescribed final time. 
For interior actuation, \cite{ZhGa:14} shows that approximate controllability may hold while null controllability fails, highlighting the limitations introduced by memory. 
In \cite{SiZhZu:17}, memory-type null controllability is studied using controls with moving support, and it is shown that, except in trivial cases, the moving control region must cover the whole spatial domain during the control time interval. 
These results motivate a shift from exact controllability objectives toward stabilizability.

Several stabilization methods have been developed for heat equations with memory and related parabolic coupled systems, including direct stabilization arguments and approximation-based feedback constructions. 
For linear heat equations with memory, stability and exponential stabilization were studied in \cite{LiZhGu:18}. 
Riccati-based feedback laws and implementable compensator constructions for parabolic PDE--ODE coupled models, based on finite-dimensional approximations, were developed in \cite{BrKu:14,BrKu:17}. 
Closer to the present setting, \cite{AkMi:2022} studies an $\omega$-stabilization problem for a viscous Burgers equation with memory using linearization, spectral analysis, and finite-dimensional Riccati design. 
Feedback stabilization of a parabolic coupled system, together with finite-dimensional approximation and numerical validation, was developed in \cite{AkMiNaRa:2024}. 
The Riccati-based $\omega$-stabilization approach for a heat equation with memory was developed in \cite{SiAkMiNa:2025}. 
These works motivate the approximation-based Riccati strategy used in this paper.

A practical route to controller design for PDE models is the early-lumping approach. 
The infinite-dimensional system is first replaced by finite-dimensional approximations, and controllers are then designed on the reduced models. 
Such approximation ideas are standard in parameter estimation and/or LQR design for distributed-parameter systems \cite{BaIt:1988,BaKu:1984,ChSuNa:2025}. 
They have since appeared in several infinite-dimensional control problems, including thermoelastic optimal control, computational LQR, and backstepping-based boundary stabilization of parabolic PDEs \cite{GiRoTa:1992,GrMo:1996,BaKr:2002,BoBaKr:2003}. 
Uniform stability along approximation schemes, convergence of Riccati-based controllers, and approximation-driven design questions such as actuator placement have also been investigated in \cite{BuCh:2022,KaMo:2013,RaTaTu:2007}. 
Approximation-based constructions have also been used recently in motion planning for parabolic equations using flatness and finite-difference schemes \cite{ChNa:2025}. 
More recently, gain approximation methods for abstract parabolic systems and spectral projection-based constructions have been developed in \cite{BaRa:2024,BaRa:2024b}. 

Our objective is to achieve a prescribed exponential decay rate by bounded state feedback. More precisely, for a given $0<\omega<\omega_0$, we seek a bounded feedback law such that every closed-loop solution decays to zero at a rate strictly larger than $\omega$. This problem was treated in our recent work \cite{SiAkMiNa:2025} for a heat equation with memory, where the state operator has a negative spectral accumulation point that imposes an intrinsic upper bound on the decay rates achievable by bounded feedback. It was shown there that every admissible decay rate below this bound can be realized by an LQR-based controller obtained through finite-dimensional approximations. The present paper extends that framework to the advection--diffusion model with fading memory considered here. This extension is not routine. In contrast with \cite{SiAkMiNa:2025}, the eigenvalues are now described by the roots of a mode-dependent cubic polynomial, and the presence of advection makes the spectral analysis substantially more involved.

The main contribution of this paper is to establish the operator-theoretic and approximation-theoretic framework needed for $\omega$-stabilization of this advection--diffusion model with fading memory. We explicitly characterize the point spectrum of $A$, showing that it consists of the two special eigenvalues $0$ and $-\kappa$ together with three families of roots of a mode-dependent cubic polynomial. One branch diverges to $-\infty$, while the other two branches converge to the unique accumulation point $-\omega_0=-(\kappa+\gamma/\eta)$. Apart from the eigenvalue $0$, the remaining eigenvalues lie in the open left half-plane. Since the resolvent of $A$ is not compact, it is not automatic that the full spectrum is determined by the point spectrum or that eigenvalues have finite algebraic multiplicity. We therefore prove separately that $\sigma(A)=\overline{\sigma_p(A)}$ and that every eigenvalue has finite algebraic multiplicity, see Theorems~\ref{thm: complete spectrum} and~\ref{thm:finite-multi}. These results allow us to apply the Hautus test and obtain verifiable $\omega$-stabilizability conditions for the shifted pair $(A+\omega I,B)$. Finally, we construct finite-dimensional approximations, solve the corresponding algebraic Riccati equations, and adapt the uniform stabilizability arguments developed in \cite{SiAkMiNa:2025} for heat equations with memory. This allows us to prove that the resulting feedback gains solve the original $\omega$-stabilization problem for all sufficiently large approximation orders.

The remainder of the paper is organized as follows. In Section~\ref{sec:Math-model-problem}, we derive the PDE--ODE model, define the state and control operators, and state the $\omega$-stabilization problem. Section~\ref{sec:properties-A-B} studies the spectral properties of the state operator, proves analytic semigroup generation, and presents verifiable $\omega$-stabilizability conditions. In Section~\ref{sec4}, we formulate the LQR design, introduce the finite-dimensional approximations, and prove convergence of the Riccati-based feedbacks. A numerical example is presented in Section~\ref{sec:numerical-example}. Section~\ref{sec:conclu} contains concluding remarks.

\noindent{\bf Notations and definitions}. For $\omega\in\rline$, define $\cline_\omega^{+}=\{\lambda\in\cline:\Re\lambda>\omega\}$ and $\cline_\omega^{-}=\{\lambda\in\cline:\Re\lambda<\omega\}$, and let $\overline{\cline_\omega^{+}}$ and $\overline{\cline_\omega^{-}}$ denote their closures in $\cline$. The symbols $\rline$, $\cline$ and $\nline$ denote the sets of real numbers, complex numbers and positive integers, respectively. Unless stated otherwise, all function spaces are complex valued. For $\lambda\in\cline$, $\Re\lambda$ and $\Im\lambda$ denote its real and imaginary parts. We use the Hilbert spaces $L^2(0,1)$, $H^1(0,1)$ and $H^2(0,1)$ with their usual meanings. The inner product on $L^2(0,1)$ is $\langle u,v\rangle_{L^2(0,1)}=\int_0^1 u(\xi)\overline{v(\xi)}\,d\xi$, and the inner product on $H^1(0,1)$ is $\langle u,v\rangle_{H^1(0,1)}=\int_0^1 u(\xi)\overline{v(\xi)}\,d\xi+\int_0^1 u_\xi(\xi)\overline{v_\xi(\xi)}\,d\xi$, with the associated norms denoted by $\|\cdot\|_{L^2(0,1)}$ and $\|\cdot\|_{H^1(0,1)}$, respectively. For Hilbert spaces $X$ and $Y$, the product space $X\times Y$ is equipped with the inner product $\left\langle\sbm{x_1\\ y_1},\sbm{x_2\\ y_2}\right\rangle_{X\times Y}=\langle x_1,x_2\rangle_X+\langle y_1,y_2\rangle_Y$. We write $\Lscr(X,Y)$ for the space of bounded linear operators from $X$ to $Y$, equipped with the norm $\|\cdot\|_{\Lscr(X,Y)}$, and write $\Lscr(X)$ when $X=Y$. For a linear operator $A$, we denote its kernel by $\ker(A)$, its range by $\operatorname{Ran}(A)$, its spectrum by $\sigma(A)$, its point spectrum by $\sigma_p(A)$, its resolvent set by $\rho(A)$, and its adjoint by $A^*$. If $A$ is Fredholm, then $\operatorname{ind}(A)$ denotes its Fredholm index. The identity operator on the relevant space is denoted by $I$.
\section{System model and problem statement}\label{sec:Math-model-problem}
\setcounter{equation}{0} 

\ \ \ We consider a 1D heat-transfer model for a fluid with advective transport, Fourier conduction, and conduction with memory, under thermally insulated boundary conditions. Let $\xi\in(0,1)$ denote the spatial variable. The temperature $y$ is assumed to satisfy the balance law
\begin{equation}\label{eq:energy-balance}
\rho c_p\, y_t(\xi,t) + \varphi_{\xi}(\xi,t) = \tilde{b}(\xi)\,u(t)
\FORALL \xi\in(0,1),\ \forall t>0 .
\end{equation}
Here $\varphi$ denotes the total heat flux, $\tilde{b}$ the input shape function and $u$ the control input. The constants $\rho$ and $c_p$ denote the fluid density and the specific heat capacity at constant pressure, respectively.

We model the heat flux $\varphi(\xi,t)$, for $\xi\in(0,1)$ and $t>0$, by
\begin{equation}\label{eq:total-flux}
\varphi(\xi,t)
= \rho c_p \nu y(\xi,t)
- k y_{\xi}(\xi,t)
- \gamma \rho c_p \int_{0}^{t} e^{-\kappa (t-s)} y_{\xi}(\xi,s)\,\mathrm{d}s.
\end{equation}
The term $\rho c_p \nu y(\xi,t)$ represents advective transport with bulk velocity $\nu>0$. The term $-k y_{\xi}(\xi,t)$ represents instantaneous Fourier conduction with conductivity $k>0$. The integral term represents memory in the heat flux, so that the present heat flux depends on past temperature gradients. The parameter $\kappa>0$ determines how fast the influence of past gradients decays, while $\gamma>0$ determines the strength of the memory term.

Set $\eta=\frac{k}{\rho c_p}$ and $b=\frac{1}{\rho c_p}\tilde{b}$. Substituting $\varphi$ considered in \eqref{eq:total-flux} into \eqref{eq:energy-balance} and dividing the resulting equation by $\rho c_p$, we get
\begin{equation}\label{eq:divergence-form}
y_t(\xi,t)
+\Big(
\nu y(\xi,t)
- \eta  y_{\xi}(\xi,t)
- \int_{0}^{t} \gamma e^{-\kappa(t-s)} y_{\xi}(\xi,s) \dd s
\Big)_{\xi}
= b(\xi)u(t)
\end{equation}
for all $\xi\in(0,1)$ and $t>0$. By defining $z(\xi,t)= \int_{0}^{t}\gamma e^{-\kappa(t-s)}y(\xi,s)\dd s$ we rewrite \eqref{eq:divergence-form} as the following coupled PDE–ODE system
\begin{align}
y_t(\xi,t) &= \big( \eta y_{\xi}(\xi,t) - \nu y(\xi,t) + z_{\xi}(\xi,t) \big)_{\xi} + b(\xi)u(t)
&& \forall\, \xi\in(0,1),\ \forall t>0, \label{eq:PDE-ODE-1} \\[1 mm]
z_t(\xi,t) &= -\kappa z(\xi,t) + \gamma y(\xi,t)
&& \forall\, \xi\in(0,1),\ \forall t>0. \label{eq:PDE-ODE-2}
\end{align}
Substituting the expression for $z$ into \eqref{eq:total-flux} and imposing thermal insulation at the boundaries, i.e. $\varphi(0,t)=\varphi(1,t)=0$, we obtain the boundary conditions for the coupled system \eqref{eq:PDE-ODE-1}-\eqref{eq:PDE-ODE-2} to be 
\begin{equation}\label{eq:boundary-PDE-ODE}
\nu y(0,t) = \eta y_{\xi}(0,t) + z_{\xi}(0,t), \qquad
\nu y(1,t) = \eta y_{\xi}(1,t) + z_{\xi}(1,t)
\FORALL t>0.
\end{equation} 
Consider the Hilbert spaces $X=L^2(0,1)\times H^1(0,1)$ and $V=H^1(0,1)\times H^1(0,1)$. We write the coupled system \eqref{eq:PDE-ODE-1}--\eqref{eq:boundary-PDE-ODE} as an abstract evolution equation on the state space $X$ as
\begin{equation}\label{eq:abstract-evolution}
\bbm{\dot{y}(t)\\ \dot{z}(t)}=A\bbm{y(t)\\z(t)}+Bu(t) \FORALL t>0.
\end{equation}
Here the state operator $A\colon D(A)\subset X\to X$ has domain
\begin{equation}\label{eq:Domain-A}
\begin{aligned}
D(A)=\left\{\bbm{y\\ z}\in V \,\middle|\,
\begin{array}{l}
\eta y+z \in H^{2}(0,1),\\[1mm]
(\eta y+z)_{\xi}(0)=\nu\,y(0),\\
(\eta y+z)_{\xi}(1)=\nu\,y(1)
\end{array}
\right\},
\end{aligned}
\end{equation}
with
\begin{equation}\label{eq:operator-A}
A\bbm{y\\z}=\bbm{(\eta y+z)_{\xi\xi}-\nu y_{\xi}\\-\kappa z+\gamma y} \FORALL \bbm{y\\z}\in D(A).
\end{equation}
We take the control input $u$ to be scalar valued, that is, $u(t)\in\cline$. We assume that $\tilde{b}\in L^2(0,1)$, and hence $b=\frac{1}{\rho c_p}\tilde{b}\in L^2(0,1)$. The control operator $B\in\Lscr(\cline,X)$ in \eqref{eq:abstract-evolution} is defined by
\begin{equation}\label{eq:operator-B}
Bu =\bbm{b u \\ 0} \FORALL u \in \cline.
\end{equation}
The state operator $A$ generates an analytic semigroup $\mathbb{T}$ on $X$, see Theorem \ref{thm:analytic-semigroup} in Section \ref{sec:properties-A-B}. Hence, for each input $u\in L^2((0,\infty);\mathbb{C})$ and initial state $\sbm{y^0\\ z^0}\in X$, the mild solution $\sbm{y\\ z}\in C([0,\infty);X)$ of \eqref{eq:abstract-evolution} is given by
$$
\bbm{y(t)\\ z(t)}
=\mathbb{T}(t)\bbm{y^0\\ z^0}
+\int_0^t \mathbb{T}(t-s)\,B\,u(s)\,ds
\FORALL t\ge 0.
$$
This mild solution is also the unique weak solution of \eqref{eq:PDE-ODE-1}--\eqref{eq:boundary-PDE-ODE} for the same input $u$ and initial state $\sbm{y^0\\ z^0}$, see \cite[Proposition 4.2.5 and Remark 4.2.6]{obs_book}. We take the first component $y$ of the mild solution with initial state $\sbm{y^0\\ 0}$ as the solution of \eqref{eq:divergence-form}, with insulated boundary conditions $\varphi(0,t)=\varphi(1,t)=0$, for the same input $u$ and initial state $y^0$.

Define $\omega_0=\kappa+\frac{\gamma}{\eta}$. We address the following $\omega$-stabilization problem in this paper. 
\begin{framed}
\begin{problem}\label{prob:w-stab}
Given $0 < \omega < \omega_0$, determine a stabilizing feedback gain 
$K_\omega \in \Lscr(X, \cline)$ such that the strongly continuous semigroup 
$\tline^{cl}$ on $X$, generated by the operator $A + B K_\omega$, satisfies
\begin{equation}\label{eq:problem}
  \|\tline^{cl}(t)\|_{\Lscr(X)} \leq M e^{-(\omega+\epsilon)t}
\FORALL t>0,  
\end{equation}
for some constants $M > 0$ and $\epsilon > 0$.
\end{problem}
\end{framed}

To solve Problem \ref{prob:w-stab}, we formulate an LQR problem for \eqref{eq:abstract-evolution}. For each $n\in\mathbb{N}$, we construct finite-dimensional approximations $A_n$ and $B_n$ of $A$ and $B$, and solve the corresponding finite-dimensional LQR problem to obtain $K_n$ such that every eigenvalue of $A_n+B_nK_n$ satisfies $\Re(\lambda)<-\omega$. We then show that, for all sufficiently large $n$, the choice $K_\omega=K_n$ solves Problem \ref{prob:w-stab}, see Theorem \ref{thm:main-result}. In practice, $K_n$ is obtained from a finite-dimensional algebraic Riccati equation. This approximation-based LQR strategy is related to \cite{SiAkMiNa:2025}. However, the advection term makes the spectral analysis substantially more involved, and the required spectral properties are established separately in Section \ref{sec:properties-A-B}.

In Problem \ref{prob:w-stab}, $\mathbb{T}^{cl}$ denotes the semigroup of the closed-loop system under the state feedback law $u=K_\omega\sbm{y\\ z}$ in \eqref{eq:abstract-evolution}. If \eqref{eq:problem} holds, then the same feedback applied to \eqref{eq:PDE-ODE-1}--\eqref{eq:boundary-PDE-ODE} yields exponential decay of $y$. The law $u=K_\omega\sbm{y\\ z}$ is static for the coupled PDE--ODE system, since $\sbm{y\\ z}$ is the state. When the heat equation with memory is written only in terms of $y$, however, the same law becomes dynamic, since $z$ must be generated through $\dot{z}(t)=-\kappa z(t)+\gamma y(t)$. Problem \ref{prob:w-stab} admits no solution when $\omega\geq\omega_0$, see Remark~\ref{rem:omega-ge-omega0}.

\section{Properties of the state operator $A$ and control
operator $B$}\label{sec:properties-A-B}
\setcounter{equation}{0} 
\ \ \ In Section \ref{sec:spectral-A}, we analyze the spectrum of the state operator $A$ defined in \eqref{eq:Domain-A}--\eqref{eq:operator-A}. We first characterize the point spectrum of $A$, see Theorem \ref{thm:point-spectrum}, and identify the unique accumulation point $-\omega_0$, where $\omega_0=\kappa+\gamma/\eta$. We then show that the full spectrum of $A$ is the closure of its eigenvalues, see Theorem \ref{thm: complete spectrum}, and prove that every eigenvalue of $A$ has finite algebraic multiplicity, see Theorem \ref{thm:finite-multi}. These results will be used later in the verification of the PBH test, since for each $0<\omega<\omega_0$ the relevant part of the shifted spectrum consists of finitely many isolated eigenvalues of finite algebraic multiplicity.

In Section \ref{sec:analytic-semigroup}, we associate a sesquilinear form with $A$ and show that $A$ generates an analytic semigroup on $X=L^2(0,1)\times H^1(0,1)$, see Theorem \ref{thm:analytic-semigroup}.

In Section \ref{sec:stab-A-B}, we use these spectral and semigroup properties to verify the PBH test and obtain a sufficient condition for $\omega$-stabilizability of the pair $(A,B)$ for every $0<\omega<\omega_0$, see Theorem \ref{thm:omega-stab}. In particular, this gives a condition under which Problem \ref{prob:w-stab} admits a solution.
\subsection{\textbf{Spectral description of $A$}}\label{sec:spectral-A}
\ \ \ In the following theorem we compute eigenvalues of $A$ and show that they have a unique accumulation point.
\begin{framed}
\begin{theorem}\label{thm:point-spectrum}
The set of eigenvalues of the state operator $A$ is
$$
\sigma_p(A)=\{0,\,-\kappa\}\;\cup\;\bigcup_{j\in\mathbb{N}}\{\lambda_{1,j},\lambda_{2,j},\lambda_{3,j}\}.
$$
For each $j\in\mathbb{N}$, the complex numbers $\lambda_{1,j}$, $\lambda_{2,j}$, and $\lambda_{3,j}$ are the roots of the cubic polynomial
\begin{equation}\label{eq:cubic-eq}
\lambda^{3}+a_j\lambda^{2}+b_j\lambda+c_j=0,
\end{equation}
where
\begin{align}
a_j&=\frac{\nu^{2}+4\gamma+8\eta\kappa+4\pi^{2}j^{2}\eta^{2}}{4\eta}, \label{eq:a-n}\\
b_j&=\frac{2\kappa\nu^{2}+4\gamma\kappa+4\eta\kappa^{2}+8\pi^{2}j^{2}\eta(\gamma+\eta\kappa)}{4\eta},\label{eq:b-n}\\
c_j&=\frac{\kappa^{2}\nu^{2}+4\pi^{2}j^{2}(\gamma+\eta\kappa)^{2}}{4\eta}.\label{eq:c-n}
\end{align}
For each $j\in\mathbb{N}$, we label the three roots of \eqref{eq:cubic-eq} as $\lambda_{1,j}$, $\lambda_{2,j}$ and $\lambda_{3,j}$ so that
$\Re\lambda_{1,j}\le \Re\lambda_{2,j}\le \Re\lambda_{3,j}$.
When $\Re\lambda_{1,j}=\Re\lambda_{2,j}$, we adopt the convention that $\Im\lambda_{1,j}\le \Im\lambda_{2,j}$.
When $\Re\lambda_{2,j}=\Re\lambda_{3,j}$, we use the same convention to ensure that $\Im\lambda_{2,j}\le \Im\lambda_{3,j}$.
In the case of repeated roots, the corresponding equalities among $\lambda_{1,j}$, $\lambda_{2,j}$ and $\lambda_{3,j}$ are understood.

Define $\omega_0=\kappa+\gamma/\eta$. With the above labeling the cubic roots satisfy 
\begin{equation}\label{eq:eigen-limits}
\lim_{j\to\infty}\lambda_{1,j}=-\infty,
\qquad
\lim_{j\to\infty}\lambda_{2,j}=-\omega_0,
\qquad
\lim_{j\to\infty}\lambda_{3,j}=-\omega_0.
\end{equation}
Furthermore, $\Re\lambda_{k,j}<0$ for all $k\in\{1,2,3\}$ and $j\in\mathbb{N}$, and the eigenspace $\ker(\lambda I-A)$ is one-dimensional for every $\lambda\in\sigma_p(A)$.
\end{theorem}
\end{framed}

\begin{proof}
To determine eigenvalues of $A$, we consider the equation 
\begin{equation}\label{eq:eigen-equation}
\left(A\bbm{ y \\ z}\right)(\xi)=\bbm{(\eta y+z)_{\xi \xi}(\xi)-\nu y_\xi(\xi)\\ \gamma y(\xi)-\kappa z(\xi)}=\lambda\bbm{y(\xi) \\ z(\xi)} \qquad \text{for a.e. } \xi\in(0,1),
\end{equation}
subject to the boundary conditions prescribed by $D(A)$
\begin{equation}\label{eq:eigen-boundary}
(\eta y+z)_\xi(0)=\nu y(0), \qquad (\eta y+z)_\xi(1)=\nu y(1).
\end{equation}
Set $w=\eta y+z$. Combining this change of variables with the second component of \eqref{eq:eigen-equation}, namely $\gamma y-\kappa z=\lambda z$, we re-write \eqref{eq:eigen-equation}-\eqref{eq:eigen-boundary} as the following ODE in $w$:
\begin{equation}\label{eq:eigen-main-ODE}
\bigl(\eta(\lambda+\kappa)+\gamma\bigr) w_{\xi\xi}(\xi)
-\nu(\lambda+\kappa) w_\xi(\xi)
-\lambda(\lambda+\kappa) w(\xi)=0
\qquad \text{for a.e. } \xi\in(0,1),
\end{equation}
with boundary conditions
\begin{align}
\bigl(\eta(\lambda+\kappa)+\gamma\bigr) w_\xi(0)
&=\nu(\lambda+\kappa) w(0),\label{eq:eigen-main-boundary-1} \\
\bigl(\eta(\lambda+\kappa)+\gamma\bigr) w_\xi(1)
&=\nu(\lambda+\kappa) w(1).\label{eq:eigen-main-boundary-2}
\end{align}
Consider $\lambda=-(\kappa+\gamma/\eta)$. Then $\eta(\lambda+\kappa)+\gamma=0$ and $\lambda+\kappa\neq 0$, so \eqref{eq:eigen-main-ODE} becomes $\nu w_\xi(\xi)-(\kappa+\gamma/\eta)w(\xi)=0$ for all $\xi\in(0,1)$. From \eqref{eq:eigen-main-boundary-1} we get $w(0)=0$, and the ODE then forces $w(\xi)=0$ on $(0,1)$. Hence $w=\eta y+z=0$, so $z=-\eta y$. Substituting into \eqref{eq:eigen-equation} gives $y_\xi(\xi)=(1/\nu)(\kappa+\gamma/\eta)y(\xi)$ for a.e. $\xi\in(0,1)$. Moreover, $w_\xi(0)=0$ and \eqref{eq:eigen-boundary} implies $y(0)=0$, which yields $y(\xi)=0$ and $z(\xi)=0$ for all $\xi\in(0,1)$. Therefore $\lambda=-(\kappa+\gamma/\eta)$ is not an eigenvalue of $A$. In the remainder of the proof we consider $\lambda\neq-(\kappa+\gamma/\eta)$. 

Define
\begin{equation}\label{eq:coefficients}
\alpha_{\lambda}
=-\frac{\nu(\lambda+\kappa)}{\eta(\lambda+\kappa)+\gamma},
\qquad
\beta_{\lambda}
=-\frac{\lambda(\lambda+\kappa)}{\eta(\lambda+\kappa)+\gamma}.
\end{equation}
We divide \eqref{eq:eigen-main-ODE} by $\eta(\lambda+\kappa)+\gamma$ and rewrite the obtained equation using \eqref{eq:coefficients} as 
\begin{equation}\label{eq:first-order-system}
\frac{d}{d \xi}\bbm{w(\xi)\\ w_{\xi}(\xi)}
=
\bbm{0 & 1\\ -\beta_{\lambda} & -\alpha_{\lambda}}
\bbm{w(\xi)\\ w_\xi(\xi)} \qquad \text{for a.e. } \xi\in(0,1).
\end{equation}
The solution of \eqref{eq:first-order-system} can be written explicitly as
\begin{equation}\label{eq:solution}
\bbm{w(\xi)\\ w_\xi(\xi)}
=e^{\sbm{0& 1 \\ -\beta_\lambda & -\alpha_\lambda} \xi} 
\bbm{w(0)\\ w_\xi(0)}
\FORALL \xi\in(0,1).
\end{equation}
Using the definition of $\alpha_\lambda$ in \eqref{eq:coefficients}, the boundary conditions
\eqref{eq:eigen-main-boundary-1}–\eqref{eq:eigen-main-boundary-2}
reduce to $w_\xi(0)=-\alpha_\lambda w(0)$ and $w_\xi(1)=-\alpha_\lambda w(1)$.
Evaluating \eqref{eq:solution} at $\xi=1$ and enforcing the boundary condition at $\xi=1$
yields
$$
\bbm{\alpha_\lambda & 1}e^{\sbm{0 & 1 \\ -\beta_\lambda & -\alpha_\lambda}}\bbm{1\\ -\alpha_\lambda}w(0)=0.
$$
Assume that $w(0)=0$. Then \eqref{eq:eigen-main-boundary-1} and $\eta(\lambda+\kappa)+\gamma\neq 0$ gives $w_\xi(0)=0$, and \eqref{eq:solution} implies $w(\xi)=0$ for all $\xi\in(0,1)$. Hence $w=\eta y+z=0$, so $z=-\eta y$. Substituting into the second component of \eqref{eq:eigen-equation}, $\gamma y-\kappa z=\lambda z$, yields $(\eta(\lambda+\kappa)+\gamma)y=0$, and thus $y=z=0$. Therefore $w(0)=0$ leads only to the trivial solution, so if $\lambda$ is an eigenvalue then necessarily $w(0)\neq 0$ and $\lambda$ must satisfy
\begin{equation}\label{eq:main-eigenvalue-equation-two}
\bbm{\alpha_\lambda & 1}\,
e^{\sbm{0 & 1 \\ -\beta_\lambda & -\alpha_\lambda}}\,
\bbm{1 \\ -\alpha_\lambda}=0.
\end{equation}
To simplify \eqref{eq:main-eigenvalue-equation-two}, we compute the matrix
exponential of
$ \sbm{0 & 1 \\ -\beta_\lambda & -\alpha_\lambda}$. This is achieved by diagonalizing this matrix, for which we first determine its eigenvalues,
\begin{equation}\label{eq:eigenvalues-matrix}
\mu^{+}_\lambda
=\frac{-\alpha_\lambda+\sqrt{\alpha_\lambda^{2}-4\beta_\lambda}}{2}, \qquad \mu^{-}_\lambda
=\frac{-\alpha_\lambda-\sqrt{\alpha_\lambda^{2}-4\beta_\lambda}}{2}.
\end{equation}
We begin with the case $\mu_\lambda^{+}\neq \mu_\lambda^{-}$ (the repeated-root case $\mu_\lambda^{+}=\mu_\lambda^{-}$ is treated later). In this case,
$$
\bbm{0 & 1 \\ -\beta_\lambda & -\alpha_\lambda}
=
\bbm{1 & 1 \\ \mu_\lambda^{+} & \mu_\lambda^{-}}
\bbm{\mu_\lambda^{+} & 0 \\ 0 & \mu_\lambda^{-}}
\bbm{1 & 1 \\ \mu_\lambda^{+} & \mu_\lambda^{-}}^{-1}.
$$
Using the diagonalization above, we evaluate the matrix exponential explicitly. Substituting the computed matrix exponential into \eqref{eq:main-eigenvalue-equation-two} gives
\begin{equation}\label{eq:case-1}
(\alpha_\lambda+\mu_\lambda^{+})(\alpha_\lambda+\mu_\lambda^{-})(e^{\mu^{+}_\lambda}-e^{\mu^{-}_\lambda})=0.
\end{equation}
Hence, if $\lambda$ is an eigenvalue of $A$ and $\mu_\lambda^{+}\neq \mu_\lambda^{-}$, then at least one of the three factors in \eqref{eq:case-1} must be zero. We consider these cases below.
\begin{itemize}
\item[Case 1.] Suppose that $\alpha_\lambda+\mu_\lambda^{+}=0$. Substituting the formula for $\mu_\lambda^{+}$ in \eqref{eq:eigenvalues-matrix},
this condition reduces to $\beta_\lambda=0$.
By the definition of $\beta_\lambda$ in \eqref{eq:coefficients},
the condition holds if and only if $\lambda=0$ or $\lambda=-\kappa$.
The choice $\lambda=-\kappa$ yields $\alpha_\lambda=\beta_\lambda=0$,
which contradicts the standing assumption $\mu_\lambda^{+}\neq\mu_\lambda^{-}$.
Consequently, $\lambda=0$ is the only admissible eigenvalue arising from this case.

For $\lambda=0$, the boundary value problem \eqref{eq:eigen-main-ODE}--\eqref{eq:eigen-main-boundary-2}
admits a nontrivial solution $w=w^0$, where the normalized solution is
$$
w^0(\xi)=e^{\frac{\nu\kappa}{\eta\kappa+\gamma}\xi} \FORALL \xi\in(0,1).
$$
Since $w=\eta y+z$ and the eigenvalue relation $\gamma y-\kappa z=\lambda z$ reduces to $y=\frac{\kappa}{\gamma}z$ for $\lambda=0$, it follows that $z=\frac{\gamma}{\eta\kappa+\gamma}w^0$ and $y=\frac{\kappa}{\eta\kappa+\gamma}w^0$. Hence, the eigenpair of $A$ corresponding to $\lambda=0$ is given by
\begin{equation}\label{eq:eigen-pair-0}
    \left(0,\bbm{\frac{\kappa}{\eta\kappa +\gamma}w^0\\[1mm] \frac{\gamma}{\eta \kappa +\gamma}w^0} \right).
\end{equation}
\item[Case 2.] If $\alpha_\lambda+\mu^{-}_\lambda=0$, substituting the formula for $\mu_\lambda^{-}$ shows that this condition also reduces to
$\beta_\lambda=0$. Hence, this case leads to the same admissible eigenpair \eqref{eq:eigen-pair-0} as in Case 1.
\item[Case 3.] Suppose that $e^{\mu_\lambda^{+}}=e^{\mu_\lambda^{-}}$. Since $\mu_\lambda^{+}\neq \mu_\lambda^{-}$ in the present case, this is equivalent to
$\mu_\lambda^{+}-\mu_\lambda^{-}=2\pi i j$ for some $j\in\mathbb{Z}\setminus\{0\}$.
Using \eqref{eq:eigenvalues-matrix}, we then have $\sqrt{\alpha_\lambda^{2}-4\beta_\lambda}=2\pi i j$.
Squaring this relation and substituting \eqref{eq:coefficients} yields the cubic eigenvalue condition \eqref{eq:cubic-eq}.

Fix $j\in\mathbb{N}$ and $k\in\{1,2,3\}$. For $\lambda=\lambda_{k,j}$, the characteristic roots $\mu_\lambda^{\pm}$ in \eqref{eq:eigenvalues-matrix} yield the general solution $w(\xi)=c_1 e^{\mu_\lambda^{+}\xi}+c_2 e^{\mu_\lambda^{-}\xi}$.
Imposing the boundary condition $w_\xi(0)=-\alpha_\lambda w(0)$ determines the ratio $c_2/c_1$.
Using the explicit formulas for $\mu_\lambda^{\pm}$ in \eqref{eq:eigenvalues-matrix} together with the identity
$\sqrt{\alpha_\lambda^{2}-4\beta_\lambda}=2\pi i j$, we obtain a normalized nontrivial solution $w=w^{k,j}$ of \eqref{eq:eigen-main-ODE}--\eqref{eq:eigen-main-boundary-2} for $\lambda=\lambda_{k,j}$ as
$$
w^{k,j}(\xi)=e^{-\frac{\alpha_{\lambda_{k,j}}}{2}\xi}
\left(\cos(j\pi \xi)-\frac{\alpha_{\lambda_{k,j}}}{2j\pi}\sin(j\pi \xi)\right)
\FORALL \xi\in(0,1).
$$
Using $w=\eta y+z$ and $\gamma y=(\lambda+\kappa)z$ from \eqref{eq:eigen-equation}, we obtain the eigenpair in this case to be
$$
\left(
\lambda_{k,j},\;
\bbm{
\dfrac{\lambda_{k,j}+\kappa}{\eta(\lambda_{k,j}+\kappa)+\gamma}\,w^{k,j}\\[2mm]
\dfrac{\gamma}{\eta(\lambda_{k,j}+\kappa)+\gamma}\,w^{k,j}
}
\right),
$$
for $k=1,2,3$ and $j\in\mathbb{N}$.

Since $\eta,\gamma,\kappa,\nu>0$, the coefficients $a_j,b_j,c_j$ in \eqref{eq:a-n}-\eqref{eq:c-n} are strictly positive for every $j\in\mathbb{N}$. Moreover, a direct calculation using \eqref{eq:a-n}-\eqref{eq:c-n} gives
$$
a_jb_j-c_j=\frac{\bigl(4\pi^2\eta^2 j^2+4\eta\kappa+2\gamma+\nu^2\bigr)\bigl((4\pi^2\eta j^2+4\kappa)(\gamma+\eta\kappa)+\kappa\nu^2\bigr)}{8\eta^2}>0,
$$
so by the Routh-Hurwitz criterion, all three roots of \eqref{eq:cubic-eq}, namely $\lambda_{1,j},\lambda_{2,j},\lambda_{3,j}$, have strictly negative real part for each $j\in\mathbb{N}$.

We now show that, as $j\to\infty$, exactly one of the three eigenvalue sequences $\{\lambda_{1,j},\lambda_{2,j},\lambda_{3,j}\}$ diverges to $-\infty$, while the other two converge to the limit point $-\omega_0$. Define
\begin{equation}\label{eq:polynomial-limit}
    \Pscr_j(\lambda)
    =\frac{1}{\pi^2 j^2}
    \bigl(\lambda^3+a_j\lambda^2+b_j\lambda+c_j\bigr).
\end{equation}
Dividing the coefficients in \eqref{eq:a-n}--\eqref{eq:c-n} by $\pi^2 j^2$ and taking the limit $j\to\infty$, we obtain
\begin{equation}\label{eq:coefficients-limits}
\lim_{j\to \infty}\frac{a_j}{\pi^2 j^2}= \eta,\quad
\lim_{j\to \infty}\frac{b_j}{\pi^2 j^2}= 2(\gamma+\eta\kappa),\quad
\lim_{j\to \infty}\frac{c_j}{\pi^2 j^2}= \frac{(\gamma+\eta\kappa)^2}{\eta}.
\end{equation}
Fix $r>0$. Taking the limit $j\to\infty$ in \eqref{eq:polynomial-limit} and using \eqref{eq:coefficients-limits}, we obtain
\begin{equation}\label{eq:limit-polynomial}
\lim_{j\to\infty}\Pscr_j(\lambda)=
\eta\Bigl(\lambda+\kappa+\frac{\gamma}{\eta}\Bigr)^2,
\end{equation}
uniformly for all $\lambda\in\mathbb{C}$ such that $|\lambda+\omega_0|\le r$. The limit function in \eqref{eq:limit-polynomial} is holomorphic and not identically zero, with a zero of order two at $\lambda=-\omega_0=-(\kappa+\frac{\gamma}{\eta})$. In particular, the limit function has no zeros on the circle $\{\lambda\in\mathbb{C}\colon |\lambda+\omega_0|=\rho\}$ for any $\rho>0$. Let $0<\rho<r$ be fixed. By Hurwitz's theorem \cite[Chapter~VII, Section~2, Theorem~2.5]{Conway:1978}, it follows that, for all sufficiently large $j$, the polynomial $\Pscr_j$ has exactly two zeros in the set $\{\lambda\in\mathbb{C}\colon |\lambda+\omega_0|<\rho\}$, counted with multiplicity, and these zeros converge to $-\omega_0$ as $j\to\infty$.

Since $\Pscr_j$ and the cubic polynomial \eqref{eq:cubic-eq} have the same zeros, and $\lambda_{1,j}$, $\lambda_{2,j}$ and $\lambda_{3,j}$ are the three roots of \eqref{eq:cubic-eq}, they satisfy the relation
$$
\lambda_{1,j}+\lambda_{2,j}+\lambda_{3,j}=-a_j.
$$
Moreover, $\lim_{j\to\infty} a_j=+\infty$. By Hurwitz’s theorem, for all sufficiently large $j$ the polynomial has exactly two roots in $\{\lambda\in\mathbb{C}\colon |\lambda+\omega_0|<\rho\}$, and these two roots remain bounded, so their sum is bounded. Since $\lambda_{1,j}+\lambda_{2,j}+\lambda_{3,j}=-a_j\to-\infty$, it follows that the third root cannot lie in this disc and must diverge to $-\infty$. Consequently, as $j\to\infty$, exactly two of the three sequences $\lambda_{1,j},\lambda_{2,j},\lambda_{3,j}$ converge to $-\omega_0$ and the remaining one diverges to $-\infty$. In view of the ordering convention for the roots of \eqref{eq:cubic-eq} stated below \eqref{eq:c-n}, this identifies the diverging root as $\lambda_{1,j}$ and the two bounded roots as $\lambda_{2,j}$ and $\lambda_{3,j}$, and hence yields the limits in \eqref{eq:eigen-limits}.
\end{itemize}
Next, we consider the case when $\mu^{+}_\lambda$ and $\mu^{-}_\lambda$
defined in \eqref{eq:eigenvalues-matrix} are equal.
This occurs if and only if $\alpha_\lambda^{2}-4\beta_\lambda=0$.
In this case, the matrix
\(
\sbm{0 & 1 \\ -\beta_\lambda & -\alpha_\lambda}
\)
admits the Jordan decomposition
\begin{equation}\label{eq:jordan-decomposition}
\bbm{0 & 1 \\ -\beta_\lambda & -\alpha_\lambda}
=
\bbm{1 & 0 \\ -\frac{\alpha_\lambda}{2} & 1}
\bbm{-\frac{\alpha_\lambda}{2} & 1 \\ 0 & -\frac{\alpha_\lambda}{2}}
\bbm{1 & 0 \\ -\frac{\alpha_\lambda}{2} & 1}^{-1}.
\end{equation}
Using the Jordan decomposition \eqref{eq:jordan-decomposition}, we evaluate the matrix exponential of
$\sbm{0 & 1 \\ -\beta_\lambda & -\alpha_\lambda}$ and substitute it into \eqref{eq:main-eigenvalue-equation-two}. This simplifies to
$$
\alpha_\lambda^{2}e^{-\alpha_\lambda/2}=0,
$$
and therefore $\alpha_\lambda=0$. By the definition of $\alpha_\lambda$ in \eqref{eq:coefficients}, this is equivalent to $\lambda=-\kappa$ (indeed, for $\lambda=-\kappa$ one has $\alpha_\lambda^{2}-4\beta_\lambda=0$). Substituting $\lambda=-\kappa$ into \eqref{eq:eigen-main-ODE}--\eqref{eq:eigen-main-boundary-2} yields the boundary value problem $w_{\xi\xi}(\xi)=0$ for a.e. $\xi\in(0,1)$ with $w_\xi(0)=w_\xi(1)=0$, whose nontrivial solutions are constant functions $w(\xi)= c$ for all $\xi\in (0,1)$, $c\neq 0$. Finally, since $w=\eta y+z$ and the second equation in \eqref{eq:eigen-equation} at $\lambda=-\kappa$ gives $\gamma y=0$, we obtain $y=0$ and $z= c$. Hence the eigenpair associated with $\lambda=-\kappa$ is
$$
\left(-\kappa,\bbm{0\\1}\right).
$$

Fix $\lambda\in \sigma_p(A)$. The corresponding function $w=\eta y+z$ satisfies \eqref{eq:eigen-main-ODE} together with the boundary condition \eqref{eq:eigen-main-boundary-1} at $\xi=0$, which fixes $w_\xi(0)$ uniquely in terms of $w(0)$. Since \eqref{eq:eigen-main-ODE} is a second-order linear ODE, prescribing the initial data $(w(0),w_\xi(0))$ uniquely determines a solution. Therefore, once $w(0)$ is chosen, the value of $w_\xi(0)$ is forced by \eqref{eq:eigen-main-boundary-1}, and the resulting solution $w$ is uniquely determined. In particular, the space of solutions of \eqref{eq:eigen-main-ODE} satisfying \eqref{eq:eigen-main-boundary-1} is one-dimensional, parametrized by the single scalar $w(0)$.

Moreover, if $w(0)=0$, then \eqref{eq:eigen-main-boundary-1} gives $w_\xi(0)=0$, and uniqueness for \eqref{eq:eigen-main-ODE} implies $w(\xi)=0$ for all $\xi\in(0,1)$, using \eqref{eq:eigen-equation} and $w=\eta y+z$ (together with $\eta(\lambda+\kappa)+\gamma\neq 0$ for $\lambda\in\sigma_p(A)$), this forces $y=z=0$. Hence every nontrivial eigenfunction is unique up to a multiplicative constant. Consequently, $\dim\ker(\lambda I-A)=1$ for every $\lambda\in\sigma_p(A)$. This completes the proof of the theorem.
\end{proof}
We next establish a well-posedness result for a second-order boundary value problem with complex coefficients and complex-valued unknown. This result will play a central role in the spectral analysis of $A$, and will be used to prove that the spectrum of $A$ is the closure of its point spectrum, see Theorem~\ref{thm: complete spectrum}.
\begin{framed}
\begin{proposition}\label{prop:robin-bvp-wellposedness}
Let $\alpha,\beta\in\mathbb{C}$. For $f\in L^2(0,1)$, consider the differential equation
\begin{equation}\label{eq:robin-bvp}
w_{\xi\xi}(\xi)+\alpha w_\xi(\xi)+\beta w(\xi)=f(\xi) 
\qquad \text{for a.e. } \xi\in(0,1),
\end{equation}
with the boundary conditions
\begin{equation}\label{eq:robin-bc}
w_\xi(0)+\alpha w(0)=0,
\qquad
w_\xi(1)+\alpha w(1)=0.
\end{equation}
Assume that, when $f=0$ in \eqref{eq:robin-bvp}, the only solution $w\in H^2(0,1)$ satisfying \eqref{eq:robin-bc} is $w=0$. Then, for every $f\in L^2(0,1)$, there exists a unique solution $w_f\in H^2(0,1)$ satisfying \eqref{eq:robin-bvp} and \eqref{eq:robin-bc}. Moreover, there exists a constant $C>0$ such that
\begin{equation}\label{eq:robin-bvp-estimate}
\|w_f\|_{H^2(0,1)}\le C\|f\|_{L^2(0,1)}.
\end{equation}
\end{proposition}
\end{framed} 

\begin{proof}
Consider the Volterra-type transformation
\begin{equation}\label{eq:volterra-transform}
v(\xi)=w(\xi)+\alpha\int_0^\xi w(s)\,ds \qquad \forall \xi\in(0,1).
\end{equation}
The inverse transformation is given by
\begin{equation}\label{eq:volterra-transform-inverse}
w(\xi)=v(\xi)-\alpha\int_0^\xi e^{-\alpha(\xi-s)}v(s)\,ds \qquad \forall \xi\in(0,1).
\end{equation}
Then \eqref{eq:robin-bvp} and \eqref{eq:robin-bc} hold if and only if
\begin{equation}\label{eq:transformed-bvp}
v_{\xi\xi}(\xi)+\beta\left(v(\xi)-\alpha\int_0^\xi e^{-\alpha(\xi-s)}v(s)\,ds\right)=f(\xi)
\qquad \text{for a.e. } \xi\in(0,1),
\end{equation}
and
\begin{equation}\label{eq:transformed-bvp-boundary}
v_\xi(0)=0,
\qquad
v_\xi(1)=0.
\end{equation}
Indeed, if $w\in H^2(0,1)$ and $v$ is defined by \eqref{eq:volterra-transform}, then $v_\xi(\xi)=w_\xi(\xi)+\alpha w(\xi)$ for all $\xi\in(0,1)$ and $v_{\xi\xi}(\xi)=w_{\xi\xi}(\xi)+\alpha w_\xi(\xi)$ for a.e. $\xi\in(0,1)$, so \eqref{eq:robin-bc} is equivalent to \eqref{eq:transformed-bvp-boundary}, and \eqref{eq:robin-bvp} yields \eqref{eq:transformed-bvp}. Conversely, if $v\in H^2(0,1)$ and $w$ is defined by \eqref{eq:volterra-transform-inverse}, then $w\in H^2(0,1)$, and the same identities in reverse show that \eqref{eq:transformed-bvp}--\eqref{eq:transformed-bvp-boundary} implies \eqref{eq:robin-bvp}--\eqref{eq:robin-bc}. Hence the two problems are equivalent. Since both transformations are bounded on $H^2(0,1)$, the $H^2(0,1)$ norms of $w$ and $v$ are equivalent. In particular, if $f=0$, then the transformed problem has only the trivial solution.

Define $L:D(L)\subset L^2(0,1)\to L^2(0,1)$ by
$$
D(L)=\left\{v\in H^2(0,1)\ \middle|\ v_\xi(0)=0,\ v_\xi(1)=0\right\},
$$
and
$$
(Lv)(\xi)=v_{\xi\xi}(\xi)+\beta v(\xi)-\alpha\beta\int_0^\xi e^{-\alpha(\xi-s)}v(s)\,ds 
\qquad \forall v\in D(L), \quad \text{for a.e. } \xi\in(0,1).
$$
Then \eqref{eq:transformed-bvp} and \eqref{eq:transformed-bvp-boundary} can be written as $Lv=f$. Moreover, $L$ is injective, since $Lv=0$ implies that the associated $w$, given by \eqref{eq:volterra-transform-inverse}, solves \eqref{eq:robin-bvp} and \eqref{eq:robin-bc} with $f=0$, hence $w=0$ by hypothesis and so $v=0$.

Consider the operator $T:D(T)\subset L^2(0,1)\to L^2(0,1)$ defined by $D(T)=D(L)$ and
$$
(Tv)(\xi)=v_{\xi\xi}(\xi)+v(\xi) 
\qquad \forall v\in D(T), \quad \text{for a.e. } \xi\in(0,1).
$$
It can be checked directly that $T:D(T)\to L^2(0,1)$ is bijective and closed. Moreover,
$$
T^{-1}\in\mathcal L\bigl(L^2(0,1),D(T)\bigr)\subset \mathcal L\bigl(L^2(0,1),H^2(0,1)\bigr).
$$
Let
$$
(Kv)(\xi)=\int_0^\xi e^{-\alpha(\xi-s)}v(s)\,ds
\qquad \forall v\in L^2(0,1), \quad \text{for a.e. } \xi\in(0,1).
$$
Then $K\in\mathcal L(L^2(0,1),L^2(0,1))$. Now define $G\in\mathcal L(L^2(0,1),L^2(0,1))$ by
$$
(Gv)(\xi)=(\beta-1)v(\xi)-\alpha\beta\int_0^\xi e^{-\alpha(\xi-s)}v(s)\,ds
\qquad \forall v\in L^2(0,1), \quad \text{for a.e. } \xi\in(0,1).
$$
With this definition, we have $L=T+G$ on $D(T)=D(L)$. Since we have $T^{-1}\in\mathcal L(L^2(0,1),H^2(0,1))$ and the embedding $H^2(0,1)\hookrightarrow L^2(0,1)$ is compact, $T^{-1}$ is compact as an operator from $L^2(0,1)$ into $L^2(0,1)$. Since $G\in\mathcal L(L^2(0,1),L^2(0,1))$, the operator $GT^{-1}$ is compact. Thus $G$ is $T$-compact in the sense of \cite[Chapter~IV, Section~1]{Kato}. Since $T$ is closed, \cite[Chapter~IV, Section~1, Theorem~1.11]{Kato} implies that $L$ is closed. Moreover, $T$ is Fredholm with $\operatorname{ind}(T)=0$, and \cite[Chapter~IV, Section~5, Theorem~5.26]{Kato} yields that
\begin{equation}\label{eq:index-L-zero}
\operatorname{ind}(L)=\operatorname{ind}(T)=0.
\end{equation}
Since $L$ is injective and semi-Fredholm, \eqref{eq:index-L-zero} gives $\operatorname{codim}(\operatorname{Ran}(L))=0$ and $\operatorname{Ran}(L)$ is closed. Hence $\operatorname{Ran}(L)=L^2(0,1)$, and therefore $L$ is bijective.

Now endow $D(L)$ with the graph norm. Since $L$ is closed, $D(L)$ is Banach and the bounded inverse theorem gives $L^{-1}\in\mathcal L(L^2(0,1),D(L))$. Since the graph norm of $L$ on $D(L)$ is equivalent to the $H^2(0,1)$ norm, it follows that $L^{-1}\in\mathcal L(L^2(0,1),H^2(0,1))$. In particular, if $Lv_f=f$, then $\|v_f\|_{H^2(0,1)}\le C\|f\|_{L^2(0,1)}$. Finally, let $w_f$ be defined from $v_f$ by \eqref{eq:volterra-transform-inverse}. Then $w_f\in H^2(0,1)$ satisfies \eqref{eq:robin-bvp}--\eqref{eq:robin-bc}, and the equivalence of the $H^2(0,1)$ norms of $v_f$ and $w_f$ yields \eqref{eq:robin-bvp-estimate}. Uniqueness follows from the hypothesis. This completes the proof.
\end{proof}
\begin{remark}
The boundary value problems arising in the spectral analysis of $A$ are not
regular self-adjoint Sturm--Liouville problems. After the change of variables
$w=\eta y+z$, the advection term produces a first-order term and places a
complex parameter in the boundary conditions, see \eqref{eq:eigen-main-ODE}--\eqref{eq:eigen-main-boundary-2}, so the coefficients of the differential equation and of the boundary conditions are
in general complex. The spectral theory of regular self-adjoint Sturm--Liouville problems is therefore not directly applicable. In Proposition~\ref{prop:robin-bvp-wellposedness} we prove that uniqueness for the homogeneous problem implies unique solvability of the inhomogeneous problem together with an $H^2(0,1)$ estimate, and this result is used in the proof of Theorem~\ref{thm: complete spectrum}.
\end{remark}

Having computed the point spectrum of $A$ in Theorem~\ref{thm:point-spectrum}, we now combine it with Proposition~\ref{prop:robin-bvp-wellposedness} to show that $\sigma(A)$ is precisely the closure of its point spectrum.

\begin{framed}
\begin{theorem}\label{thm: complete spectrum}
The spectrum of $A$ is the closure of its set of eigenvalues, i.e.
$$
\sigma(A)=\overline{\sigma_p(A)}.
$$
\end{theorem}
\end{framed}

\begin{proof}
Fix $\mu\in\mathbb{C}\setminus\overline{\sigma_p(A)}$ and let $\sbm{f\\ g}\in X$ be arbitrary. Consider
\begin{equation}\label{eq:resolvent-equation}
(\mu I-A)\bbm{y\\ z}=\bbm{f\\ g}.
\end{equation}
Since $\mu\notin\sigma_p(A)$, the operator $\mu I-A$ is injective. We show that \eqref{eq:resolvent-equation} has a unique solution $\sbm{y\\ z}\in D(A)$ for every $\sbm{f\\ g}\in X$, and that this solution depends continuously on $\sbm{f\\ g}$. Hence $\mu\in\rho(A)$. Since $\mu$ was arbitrary, $\mathbb{C}\setminus\overline{\sigma_p(A)}\subset\rho(A)$, so $\sigma(A)\subset\overline{\sigma_p(A)}$. The reverse inclusion follows from $\sigma_p(A)\subset\sigma(A)$ and the closedness of $\sigma(A)$. Therefore $\sigma(A)=\overline{\sigma_p(A)}$.

Define $w=\eta y+z$. Using the definition of the operator $A$ in \eqref{eq:Domain-A}--\eqref{eq:operator-A}, we eliminate $y$ and $z$ from \eqref{eq:resolvent-equation} in favour of $w$ and obtain
\begin{equation}\label{eq:w-eq}
-w_{\xi\xi}(\xi)-\alpha_\mu w_\xi(\xi)-\beta_\mu w(\xi)=r_\mu(\xi)\qquad \text{for a.e. } \xi\in(0,1),
\end{equation}
with the boundary conditions
\begin{equation}\label{eq:w-bc}
w_\xi(0)+\alpha_\mu w(0)=-\frac{\nu g(0)}{\gamma +\eta(\mu+\kappa)},
\qquad
w_\xi(1)+\alpha_\mu w(1)=-\frac{\nu g(1)}{\gamma +\eta(\mu+\kappa)}.
\end{equation}
Here the coefficients $\alpha_\mu$ and $\beta_\mu$, together with the function $r_\mu$, are
\begin{equation}\label{eq:mu-coe}
\alpha_\mu=-\frac{\nu(\mu+\kappa)}{\gamma +\eta(\mu+\kappa)},
\qquad
\beta_\mu=-\frac{\mu(\mu+\kappa)}{\gamma +\eta(\mu+\kappa)},
\qquad
r_\mu=f+\frac{\mu g+\nu g_\xi}{\gamma +\eta(\mu+\kappa)}.
\end{equation}
Define
\begin{equation}\label{eq:v-def}
v(\xi)
=
w(\xi)
+
\frac{\nu}{\gamma_\mu}
\left[
\frac{(\alpha_\mu+1)g(0)-g(1)}{\alpha_\mu^{2}}
+
\frac{g(1)-g(0)}{\alpha_\mu}\,\xi
\right]
\FORALL \xi\in(0,1).
\end{equation}
Here
$\gamma_\mu=\gamma+\eta(\mu+\kappa).$
Since $\mu\notin\overline{\sigma_p(A)}$ and $-\omega_0=-\kappa-\frac{\gamma}{\eta}\in\overline{\sigma_p(A)}$, we have $\gamma_\mu\neq 0$. Also, $\mu\notin\overline{\sigma_p(A)}$ implies $\mu\neq-\kappa$, and hence $\alpha_\mu$ is well defined and nonzero. Therefore $v$ is well defined. Substituting \eqref{eq:v-def} into \eqref{eq:w-eq} and \eqref{eq:w-bc} yields
\begin{equation}\label{eq:v-eq}
v_{\xi\xi}(\xi)+\alpha_\mu v_\xi(\xi)+\beta_\mu v(\xi)=-s_\mu(\xi)\qquad \text{for a.e. } \xi\in(0,1),
\end{equation}
with boundary conditions
\begin{equation}\label{eq:v-eq-boundary}
v_\xi(0)+\alpha_\mu v(0)=0,
\qquad
v_\xi(1)+\alpha_\mu v(1)=0.
\end{equation}
The function $s_\mu$ appearing on the right-hand side of \eqref{eq:v-eq} is given by
\begin{equation}\label{eq:s-mu}
s_\mu(\xi)=r_\mu(\xi)
+
\frac{\nu}{\gamma_\mu}\!\left(
\left[\,1-\frac{\beta_\mu(\alpha_\mu+1)}{\alpha_\mu^2}+\frac{\beta_\mu}{\alpha_\mu}\xi\,\right]g(0)
+
\left[\,-1+\frac{\beta_\mu}{\alpha_\mu^2}-\frac{\beta_\mu}{\alpha_\mu}\xi\,\right]g(1)
\right)
\end{equation}
for a.e. $\xi\in(0,1)$.

We claim that the only solution $v\in H^2(0,1)$ of
$$
v_{\xi\xi}(\xi)+\alpha_\mu v_\xi(\xi)+\beta_\mu v(\xi)=0 \qquad \text{for a.e. } \xi\in(0,1),
$$
together with \eqref{eq:v-eq-boundary}, is $v=0$. Indeed, if $v$ solves it, then $w:=v$ satisfies \eqref{eq:w-eq}--\eqref{eq:w-bc} with $f=g=0$, and the pair $\sbm{y\\ z}$ defined by \eqref{eq:yz-def} with $g=0$ lies in $D(A)$ and satisfies $(\mu I-A)\sbm{y\\ z}=\sbm{0\\ 0}$, by the case $f=g=0$ of the computation below. Since $\mu\notin\sigma_p(A)$, this gives $v=w=0$. Therefore Proposition~\ref{prop:robin-bvp-wellposedness} applies with $\alpha=\alpha_\mu$, $\beta=\beta_\mu$, and $f=-s_\mu$, and there exists a unique solution $v\in H^2(0,1)$ satisfying \eqref{eq:v-eq} and \eqref{eq:v-eq-boundary}. Moreover,
\begin{equation}\label{eq:estimate main}
\|v\|_{H^2(0,1)}
\le C_\mu \|s_\mu\|_{L^2(0,1)}.
\end{equation}
From \eqref{eq:v-def} we obtain that $w$ equals $v$ minus the affine function appearing in \eqref{eq:v-def}, so $w\in H^2(0,1)$ and $w$ solves \eqref{eq:w-eq}--\eqref{eq:w-bc}. Define
\begin{equation}\label{eq:yz-def}
y=\frac{(\mu+\kappa)w-g}{\gamma_\mu},
\qquad
z=w-\eta y.
\end{equation}
Since $w\in H^2(0,1)$ and $g\in H^1(0,1)$, we have $y,z\in H^1(0,1)$ and $\eta y+z=w\in H^2(0,1)$. From \eqref{eq:w-bc} and the definition of $\alpha_\mu$, for $j=0,1$ we have
$$
w_\xi(j)=-\alpha_\mu w(j)-\frac{\nu g(j)}{\gamma_\mu}
=\frac{\nu(\mu+\kappa)}{\gamma_\mu}w(j)-\frac{\nu g(j)}{\gamma_\mu}
=\nu\,\frac{(\mu+\kappa)w(j)-g(j)}{\gamma_\mu}=\nu y(j),
$$
so $\sbm{y\\ z}\in D(A)$. For the second component of $(\mu I-A)\sbm{y\\ z}$, using $z=w-\eta y$ we get
$$
(\mu+\kappa)z-\gamma y=(\mu+\kappa)(w-\eta y)-\gamma y=(\mu+\kappa)w-\gamma_\mu y=g.
$$
For the first component, using $y_\xi=\frac{(\mu+\kappa)w_\xi-g_\xi}{\gamma_\mu}$ together with \eqref{eq:w-eq} and the definitions of $\alpha_\mu$, $\beta_\mu$ and $r_\mu$, we obtain
\begin{align*}
\gamma_\mu\bigl(\mu y-w_{\xi\xi}+\nu y_\xi\bigr)
=&-\gamma_\mu w_{\xi\xi}+\nu(\mu+\kappa)w_\xi+\mu(\mu+\kappa)w-\mu g-\nu g_\xi\\
=&\gamma_\mu r_\mu-\mu g-\nu g_\xi\\
=&\gamma_\mu f.
\end{align*}
Dividing by $\gamma_\mu\neq 0$ and recalling that $w=\eta y+z$, we obtain $\mu y-\bigl((\eta y+z)_{\xi\xi}-\nu y_\xi\bigr)=f$, which gives us
$$
(\mu I-A)\bbm{y\\ z}=\bbm{f\\ g}.
$$
The affine correction in \eqref{eq:v-def} depends linearly on $g(0)$ and $g(1)$, so its $H^2(0,1)$ norm is bounded by a constant times $|g(0)|+|g(1)|$, and the trace estimate gives $|g(0)|+|g(1)|\le C\|g\|_{H^1(0,1)}$. Therefore,
$$
\|w\|_{H^2(0,1)}\le \|v\|_{H^2(0,1)}+C\|g\|_{H^1(0,1)}.
$$
Also, \eqref{eq:estimate main} and \eqref{eq:s-mu} imply that
$$
\|v\|_{H^2(0,1)}
\le C_\mu\|s_\mu\|_{L^2(0,1)}
\le C_\mu\bigl(\|f\|_{L^2(0,1)}+\|g\|_{H^1(0,1)}\bigr),
$$
and combining these estimates, we obtain
$$
\|w\|_{H^2(0,1)}\le C_\mu\Bigl(\|f\|_{L^2(0,1)}+\|g\|_{H^1(0,1)}\Bigr).
$$
Using \eqref{eq:yz-def}, we then get
$$
\|y\|_{L^2(0,1)}+\|z\|_{H^1(0,1)}
\le C_\mu\Bigl(\|w\|_{H^2(0,1)}+\|g\|_{H^1(0,1)}\Bigr)
\le C_\mu\left(\|f\|_{L^2(0,1)}+\|g\|_{H^1(0,1)}\right).
$$
Hence
$$
\left\|\bbm{y\\ z}\right\|_X
\le C_\mu\left\|\bbm{f\\ g}\right\|_X.
$$
Therefore $\mu I-A$ is bijective, and the above estimate shows that $\mu\in\rho(A)$. Hence $\sigma(A)\subset \overline{\sigma_p(A)}$. The reverse inclusion follows from $\sigma_p(A)\subset\sigma(A)$ and closedness of $\sigma(A)$. Therefore $\sigma(A)=\overline{\sigma_p(A)}$. This completes the proof of the theorem.
\end{proof}

Recall that the algebraic multiplicity of an isolated eigenvalue $\lambda$ of a linear operator $L$ is defined as the dimension of its generalized eigenspace, i.e.,
\[
m_a(\lambda,L)
=
\dim \bigcup_{m\in\mathbb N}
\ker\bigl((\lambda I-L)^m\bigr).
\]
Having identified the spectrum of $A$, we now prove that every eigenvalue of $A$ has finite algebraic multiplicity.
\begin{framed}
\begin{theorem}\label{thm:finite-multi}
Every eigenvalue $\lambda\in\sigma_p(A)$ has finite algebraic multiplicity.
\end{theorem}
\end{framed}

\begin{proof}
Let $\omega_0=\kappa+\frac{\gamma}{\eta}$ and $w=\eta y+z$. Define the transformation $T\colon X\to X$ by
$$
T\bbm{y\\ z}=\bbm{\eta y+z\\ z} \FORALL \bbm{y\\ z}\in X.
$$
Note that $T$ is invertible and defines a bounded isomorphism on both $X$ and $V$. Define the transformed operator $A_{tr}=TAT^{-1}\colon D(A_{tr})\subset X\to X$. More precisely, the domain of $A_{tr}$ is
\begin{equation}\label{eq:Domain-AT}
\begin{aligned}
D(A_{tr})=\left\{\bbm{w\\ z}\in V \,\middle|\,
\begin{array}{l}
w\in H^2(0,1),\\[1mm]
w_\xi(0)=\dfrac{\nu}{\eta}\bigl(w(0)-z(0)\bigr),\\
w_\xi(1)=\dfrac{\nu}{\eta}\bigl(w(1)-z(1)\bigr)
\end{array}
\right\}
\end{aligned}
\end{equation}
with
\[
A_{tr}\bbm{w\\ z}
=
\bbm{
\eta w_{\xi\xi}-\nu w_\xi+\nu z_\xi+\dfrac{\gamma}{\eta}w-\omega_0 z\\[1mm]
\dfrac{\gamma}{\eta}w-\omega_0 z
} \FORALL \bbm{w\\ z}\in D(A_{tr}).
\]
We next consider the realization of $A_{tr}$ on the space $V=H^1(0,1)\times H^1(0,1)$. Define $A_{re}\colon D(A_{re})\subset V\to V$ by
\begin{equation}\label{eq:Domain-AV}
\begin{aligned}
D(A_{re})=\left\{\bbm{w\\ z}\in V \,\middle|\,
\begin{array}{l}
w\in H^2(0,1),\\[1mm]
\eta w_{\xi\xi}+\nu z_\xi\in H^1(0,1),\\[1mm]
w_\xi(0)=\dfrac{\nu}{\eta}\bigl(w(0)-z(0)\bigr),\\
w_\xi(1)=\dfrac{\nu}{\eta}\bigl(w(1)-z(1)\bigr)
\end{array}
\right\}
\end{aligned}
\end{equation}
with
\begin{equation}\label{eq:operator-AV}
A_{re}\bbm{w\\ z}
=
\bbm{
\eta w_{\xi\xi}-\nu w_\xi+\nu z_\xi+\dfrac{\gamma}{\eta}w-\omega_0 z\\[1mm]
\dfrac{\gamma}{\eta}w-\omega_0 z
} \FORALL \bbm{w\\ z}\in D(A_{re}).
\end{equation}
The operator $A_{re}$ is defined by the same differential expression as $A_{tr}$, but on the space $V$, with the additional regularity condition $\eta w_{\xi\xi}+\nu z_\xi\in H^1(0,1)$ ensuring that $A_{re}$ maps $D(A_{re})$ into $V$.

The proof proceeds in three steps. First, we show that $A$ and $A_{tr}$ have the same spectrum and point spectrum, while $A_{re}$ has the same point spectrum and satisfies $\sigma(A_{re})\subset\sigma(A)$. Using Theorems~\ref{thm:point-spectrum} and~\ref{thm: complete spectrum}, we then show that every eigenvalue of $A_{re}$ is isolated and has geometric multiplicity one. Second, we show that the generalized eigenspaces of $A_{re}$ and $A_{tr}$ coincide and use the similarity $A_{tr}=TAT^{-1}$ to conclude that algebraic multiplicity of an eigenvalue for $A_{re}$ transfers to $A$. Third, for each $\lambda\in\sigma_p(A_{re})$, we write $\lambda I-A_{re}$ as a relatively compact perturbation of a boundedly invertible operator. By \cite[Chapter~IV, Section~5, Theorem~5.26]{Kato}, $\lambda I-A_{re}$ is Fredholm, and hence has closed range. Therefore, by \cite[Theorem~4.1]{Baxley:1972}, every eigenvalue of $A_{re}$ has finite algebraic multiplicity, and we conclude that every eigenvalue of $A$ has finite algebraic multiplicity.

The similarity $A_{tr}=TAT^{-1}$ shows that $A$ and $A_{tr}$ have the same
spectrum and point spectrum,
\begin{equation}\label{eq:spectra-equal}
\sigma_p(A_{tr})=\sigma_p(A),\qquad \sigma(A_{tr})=\sigma(A),
\end{equation}
and that $\ker(\lambda I-A_{tr})=T\ker(\lambda I-A)$ for all $\lambda\in\mathbb{C}$. By \eqref{eq:Domain-AT} and \eqref{eq:Domain-AV}, we have
\begin{equation}\label{eq:Are-part}
D(A_{re})=\bigl\{\,x\in D(A_{tr}) : A_{tr}x\in V\,\bigr\},
\qquad A_{re}x=A_{tr}x \FORALL x\in D(A_{re}) .
\end{equation}
Indeed, for $\sbm{w\\ z}\in D(A_{tr})$ the second component of $A_{tr}\sbm{w\\ z}$ always
lies in $H^1(0,1)$, while the first lies in $H^1(0,1)$ if and only if
$\eta w_{\xi\xi}+\nu z_\xi\in H^1(0,1)$.

Since $A_{re}$ is the restriction of $A_{tr}$ to $D(A_{re})\subset D(A_{tr})$, we have
$\ker(\lambda I-A_{re})\subset\ker(\lambda I-A_{tr})$ for every $\lambda\in\mathbb{C}$.
Conversely, if $\sbm{w\\ z}\in\ker(\lambda I-A_{tr})$, then
$A_{tr}\sbm{w\\ z}=\lambda\sbm{w\\ z}\in V$ because $D(A_{tr})\subset V$, so
$\sbm{w\\ z}\in D(A_{re})$ by \eqref{eq:Are-part}. Hence
$\ker(\lambda I-A_{re})=\ker(\lambda I-A_{tr})$ for all $\lambda\in\mathbb{C}$, and in
particular
\[
\sigma_p(A)=\sigma_p(A_{tr})=\sigma_p(A_{re}).
\]

We now show that $A_{re}$ is closed. By \eqref{eq:spectra-equal} we have
$\rho(A_{tr})=\rho(A)$, and $\rho(A)\neq\varnothing$ by
Theorems~\ref{thm:point-spectrum} and~\ref{thm: complete spectrum}. Since an
operator with nonempty resolvent set is closed, $A_{tr}$ is closed. Let $\sbm{w_n\\ z_n}\in D(A_{re})$ with
$$
\lim_{n\to\infty}
\left\|
\bbm{w_n\\ z_n}-\bbm{w\\ z}
\right\|_{V}
=0,
\qquad
\lim_{n\to\infty}
\left\|
A_{re}\bbm{w_n\\ z_n}-\bbm{p\\ q}
\right\|_{V}
=0.
$$
As $V$ embeds continuously in $X$, both
limits also hold in $X$, and the closedness of $A_{tr}$ gives $\sbm{w\\ z}\in D(A_{tr})$
with $A_{tr}\sbm{w\\ z}=\sbm{p\\ q}\in V$. By \eqref{eq:Are-part} we get $\sbm{w\\ z}\in D(A_{re})$
and $A_{re}\sbm{w\\ z}=\sbm{p\\ q}$. Hence $A_{re}$ is closed.

Let $\mu\in\rho(A_{tr})$ and $\sbm{f\\ g}\in V$, and set
$\sbm{w\\ z}=(\mu I-A_{tr})^{-1}\sbm{f\\ g}\in D(A_{tr})$. Then
$A_{tr}\sbm{w\\ z}=\mu\sbm{w\\ z}-\sbm{f\\ g}\in V$, so $\sbm{w\\ z}\in D(A_{re})$ by
\eqref{eq:Are-part}. Moreover, $\mu I-A_{re}$ is injective, since
$(\mu I-A_{re})\sbm{w\\ z}=0$ implies $(\mu I-A_{tr})\sbm{w\\ z}=0$ and hence
$\sbm{w\\ z}=0$. Thus $\mu I-A_{re}\colon D(A_{re})\to V$ is bijective, and since $A_{re}$ is closed,
the bounded inverse theorem gives $(\mu I-A_{re})^{-1}\in\Lscr(V)$. Therefore
$\rho(A_{tr})\subset\rho(A_{re})$, i.e.
\[
\sigma(A_{re})\subset\sigma(A_{tr})=\sigma(A)=\overline{\sigma_p(A)}.
\]
By Theorem~\ref{thm:point-spectrum} the unique accumulation point of
$\sigma_p(A)$ is $-\omega_0$, which is not an eigenvalue. Hence every
$\lambda\in\sigma_p(A_{re})=\sigma_p(A)$ is an isolated point of
$\overline{\sigma_p(A)}$, and therefore of $\sigma(A_{re})$.
Moreover, $\ker(\lambda I-A_{re})=\ker(\lambda I-A_{tr})=T\ker(\lambda I-A)$ and the
invertibility of $T$ give $\dim\ker(\lambda I-A_{re})=\dim\ker(\lambda I-A)=1$ by
Theorem~\ref{thm:point-spectrum}, so each $\lambda\in\sigma_p(A_{re})$ has
geometric multiplicity one.

We now compare the algebraic multiplicities of $A_{re}$ and $A$. Fix
$\lambda\in\sigma_p(A_{re})=\sigma_p(A)$. We prove
\begin{equation}\label{eq:algebraic-multiplicity-A-Are}
m_a(\lambda,A)=m_a(\lambda,A_{re}).
\end{equation}
We claim that
\begin{equation}\label{eq:generalized-eigenspaces-re-tr}
\ker\bigl((\lambda I-A_{tr})^m\bigr)=\ker\bigl((\lambda I-A_{re})^m\bigr)\qquad \forall m\in\mathbb{N}.
\end{equation}
Since $A_{re}$ is a restriction of $A_{tr}$, the inclusion
$\ker\bigl((\lambda I-A_{re})^m\bigr)\subseteq\ker\bigl((\lambda I-A_{tr})^m\bigr)$ is clear, and we prove the reverse
by induction on $m$. The case $m=1$ is the eigenspace equality
$\ker(\lambda I-A_{re})=\ker(\lambda I-A_{tr})$ shown above. Assume the equality for some
$m\geq1$ and let $p\in\ker\bigl((\lambda I-A_{tr})^{m+1}\bigr)$. Set $q=(\lambda I-A_{tr})p$. Since
$(\lambda I-A_{tr})^m q=(\lambda I-A_{tr})^{m+1}p=0$, we have
$q\in\ker\bigl((\lambda I-A_{tr})^m\bigr)=\ker\bigl((\lambda I-A_{re})^m\bigr)\subset V$. As $p\in D(A_{tr})\subset V$,
it follows that $A_{tr}p=\lambda p-q\in V$, so \eqref{eq:Are-part} gives $p\in D(A_{re})$ with
$(\lambda I-A_{re})p=(\lambda I-A_{tr})p=q$. Hence
$(\lambda I-A_{re})^{m+1}p=(\lambda I-A_{re})^m q=0$, that is, $p\in\ker\bigl((\lambda I-A_{re})^{m+1}\bigr)$.
This proves \eqref{eq:generalized-eigenspaces-re-tr}.

Taking the union over $m$ in \eqref{eq:generalized-eigenspaces-re-tr}, the generalized
eigenspaces of $A_{re}$ and $A_{tr}$ at $\lambda$ coincide, so
$m_a(\lambda,A_{tr})=m_a(\lambda,A_{re})$. Moreover, $A_{tr}=TAT^{-1}$ gives
$\ker\bigl((\lambda I-A_{tr})^m\bigr)=T\ker\bigl((\lambda I-A)^m\bigr)$ for every $m$, and since $T$ is an
isomorphism, $m_a(\lambda,A)=m_a(\lambda,A_{tr})$. Combining these two equalities yields
\eqref{eq:algebraic-multiplicity-A-Are}. Hence, if $\lambda$ has finite algebraic multiplicity for $A_{re}$,
then $\lambda$ has finite algebraic multiplicity for $A$.

It remains to show that every $\lambda\in\sigma_p(A_{re})$ has finite algebraic
multiplicity for $A_{re}$. Choose $\mu_0\in[\mu_1,\infty)$ as in
Lemma~\ref{lem:SL-solvability}, and fix
$\lambda\in\sigma_p(A_{re})=\sigma_p(A)$. Since
$-\omega_0\notin\sigma_p(A)$, we have $\lambda+\omega_0\neq0$. In what follows,
$C_{\lambda,\mu_0}$ denotes a positive constant depending only on $\lambda$,
$\mu_0$, and the model parameters, which may change from one occurrence to the
next. Define $A_{dl}\colon D(A_{re})\subset V\to V$ and $P_\lambda\in\Lscr(V)$ by
\[
A_{dl}\bbm{w\\ z}
=
\bbm{
\eta w_{\xi\xi}-\nu w_\xi+\nu z_\xi
+\left(\frac{\gamma}{\eta}+\lambda-\mu_0\right)w-\omega_0 z\\[1mm]
-\omega_0 z
}
\qquad\forall\,\bbm{w\\ z}\in D(A_{re}),
\]
\[
P_\lambda\bbm{w\\ z}
=
\bbm{
(\mu_0-\lambda)w\\[1mm]
\frac{\gamma}{\eta}w
}
\qquad\forall\,\bbm{w\\ z}\in V.
\]
Then $A_{re}=A_{dl}+P_\lambda$ on $D(A_{re})$, so
\[
\lambda I-A_{re}=(\lambda I-A_{dl})-P_\lambda.
\]
Given $\sbm{f\\ g}\in V$, we solve
$(\lambda I-A_{dl})\sbm{w\\ z}=\sbm{f\\ g}$. The second component gives
\begin{equation}\label{eq:z-shifted}
z=\frac{g}{\lambda+\omega_0}\in H^1(0,1).
\end{equation}
Substituting \eqref{eq:z-shifted} into the first
component gives, for a.e. $\xi\in(0,1)$,
\begin{equation}\label{eq:w-eq-shifted}
-\eta w_{\xi\xi}(\xi)+\nu w_\xi(\xi)
+\left(\mu_0-\frac{\gamma}{\eta}\right)w(\xi)
=
f(\xi)+\frac{\nu}{\lambda+\omega_0}g_\xi(\xi)
-\frac{\omega_0}{\lambda+\omega_0}g(\xi),
\end{equation}
together with the boundary conditions
\begin{equation}\label{eq:w-bc-shifted}
w_\xi(0)-\frac{\nu}{\eta}w(0)
=
-\frac{\nu}{\eta(\lambda+\omega_0)}g(0),
\qquad
w_\xi(1)-\frac{\nu}{\eta}w(1)
=
-\frac{\nu}{\eta(\lambda+\omega_0)}g(1).
\end{equation}
Define the affine function
\begin{equation}\label{eq:q-lambda}
q_\lambda(\xi)=c_\lambda\xi+d_\lambda,\qquad
c_\lambda=\frac{g(1)-g(0)}{\lambda+\omega_0},\qquad
d_\lambda=\frac{1}{\lambda+\omega_0}
\left[
\frac{\eta}{\nu}g(1)
+
\left(1-\frac{\eta}{\nu}\right)g(0)
\right].
\end{equation}
A direct calculation shows that $q_\lambda$ satisfies the two boundary
relations in \eqref{eq:w-bc-shifted} in place of $w$. Hence, setting
$w=v+q_\lambda$, the inhomogeneous boundary terms cancel and the problem
\eqref{eq:w-eq-shifted}--\eqref{eq:w-bc-shifted} becomes
\begin{equation}\label{eq:v-shifted-eq}
-\eta v_{\xi\xi}(\xi)+\nu v_\xi(\xi)
+\left(\mu_0-\frac{\gamma}{\eta}\right)v(\xi)
=
\widetilde h_\lambda(\xi)
\qquad \text{for a.e. } \xi\in(0,1),
\end{equation}
with the homogeneous boundary conditions
\begin{equation}\label{eq:v-shifted-bc}
v_\xi(0)-\frac{\nu}{\eta}v(0)=0,
\qquad
v_\xi(1)-\frac{\nu}{\eta}v(1)=0,
\end{equation}
where $\widetilde h_\lambda\in L^2(0,1)$ is given by
\begin{equation}\label{eq:h-tilde}
\widetilde h_\lambda(\xi)
=
f(\xi)+\frac{\nu}{\lambda+\omega_0}g_\xi(\xi)
-\frac{\omega_0}{\lambda+\omega_0}g(\xi)
-\nu c_\lambda
-\left(\mu_0-\frac{\gamma}{\eta}\right)(c_\lambda\xi+d_\lambda)
\end{equation}
for a.e. $\xi\in (0,1)$. 

We now show that \eqref{eq:v-shifted-eq}--\eqref{eq:v-shifted-bc} has a unique
solution $v\in H^2(0,1)$ and estimate it in terms of $\sbm{f\\ g}$. The change
of variable $v=e^{\nu\xi/(2\eta)}u$ transforms
\eqref{eq:v-shifted-eq}--\eqref{eq:v-shifted-bc} into
\eqref{eq:u-ode}--\eqref{eq:u-bc} with $h=e^{-\nu\xi/(2\eta)}\widetilde h_\lambda$,
which by Lemma~\ref{lem:SL-solvability} has a unique solution $u\in H^2(0,1)$
with $\|u\|_{H^2(0,1)}\le C_{\mu_0}\|e^{-\nu\xi/(2\eta)}\widetilde h_\lambda\|_{L^2(0,1)}$.
Since $e^{\pm\nu\xi/(2\eta)}$ and their derivatives are bounded on $[0,1]$,
the function $v=e^{\nu\xi/(2\eta)}u\in H^2(0,1)$ is the unique solution of
\eqref{eq:v-shifted-eq}--\eqref{eq:v-shifted-bc} and
\begin{equation}\label{eq:v-est}
\|v\|_{H^2(0,1)}\le C_{\lambda,\mu_0}\|\widetilde h_\lambda\|_{L^2(0,1)}.
\end{equation}
Next, since $g\in H^1(0,1)\hookrightarrow C[0,1]$, we have
$|g(0)|+|g(1)|\le C\|g\|_{H^1(0,1)}$. As $c_\lambda$ and $d_\lambda$ in
\eqref{eq:q-lambda} depend linearly on $g(0)$ and $g(1)$, and every term in
\eqref{eq:h-tilde} is bounded in $L^2(0,1)$ by $\|f\|_{L^2(0,1)}$,
$\|g\|_{H^1(0,1)}$, or $|g(0)|+|g(1)|$, it follows that
\begin{equation}\label{eq:qh-est}
\|q_\lambda\|_{H^2(0,1)}+\|\widetilde h_\lambda\|_{L^2(0,1)}
\le C_{\lambda,\mu_0}\bigl(\|f\|_{L^2(0,1)}+\|g\|_{H^1(0,1)}\bigr).
\end{equation}
Finally, combining $w=v+q_\lambda$ with \eqref{eq:v-est} and \eqref{eq:qh-est}
gives $w\in H^2(0,1)$ with
$\|w\|_{H^2(0,1)}\le C_{\lambda,\mu_0}(\|f\|_{L^2(0,1)}+\|g\|_{H^1(0,1)})$,
while \eqref{eq:z-shifted} gives
$\|z\|_{H^1(0,1)}\le C_{\lambda,\mu_0}\|g\|_{H^1(0,1)}$. Since
$\|f\|_{L^2(0,1)}\le\|f\|_{H^1(0,1)}$, we conclude that
\begin{equation}\label{eq:resolvent-est-shifted}
\left\|\bbm{w\\ z}\right\|_{H^2(0,1)\times H^1(0,1)}
\le
C_{\lambda,\mu_0}\left\|\bbm{f\\ g}\right\|_V
\qquad\forall\,\bbm{f\\ g}\in V.
\end{equation}

We now verify that the pair $\sbm{w\\ z}$, with $w=v+q_\lambda$ and $z$ given by
\eqref{eq:z-shifted}, lies in $D(A_{re})$. Since $v\in H^2(0,1)$ and $q_\lambda$
is affine, we have $w\in H^2(0,1)$, and since $g\in H^1(0,1)$, we have
$z\in H^1(0,1)$. Using $z(j)=g(j)/(\lambda+\omega_0)$ for $j=0,1$, the boundary
conditions \eqref{eq:w-bc-shifted} take the form
$w_\xi(j)=\frac{\nu}{\eta}\bigl(w(j)-z(j)\bigr)$. Finally, substituting
\eqref{eq:z-shifted} into \eqref{eq:w-eq-shifted} and rearranging gives
\[
\eta w_{\xi\xi}+\nu z_\xi=\mu_0 w+\nu w_\xi-\frac{\gamma}{\eta}w+\omega_0 z-f .
\]
The right-hand side lies in $H^1(0,1)$ because $f,g\in H^1(0,1)$,
$w\in H^2(0,1)$, and $z\in H^1(0,1)$. Hence
$\eta w_{\xi\xi}+\nu z_\xi\in H^1(0,1)$, and all the conditions in
\eqref{eq:Domain-AV} hold, so $\sbm{w\\ z}\in D(A_{re})$.

Thus, for every $\sbm{f\\ g}\in V$, the equation
$(\lambda I-A_{dl})\sbm{w\\ z}=\sbm{f\\ g}$ has a solution
$\sbm{w\\ z}\in D(A_{re})$, and the solution is unique because $z$ is determined
by \eqref{eq:z-shifted} and $v$ is the unique solution of
\eqref{eq:v-shifted-eq}--\eqref{eq:v-shifted-bc}. Therefore
$\lambda I-A_{dl}\colon D(A_{re})\subset V\to V$ is bijective. Moreover,
$A_{dl}\sbm{w\\ z}=\lambda\sbm{w\\ z}-\sbm{f\\ g}$, so \eqref{eq:resolvent-est-shifted}
yields
$$
\left\|\bbm{w\\ z}\right\|_V+\left\|A_{dl}\bbm{w\\ z}\right\|_V
\le C_{\lambda,\mu_0}\left\|\bbm{f\\ g}\right\|_V,
$$
that is, $(\lambda I-A_{dl})^{-1}\in\Lscr\bigl(V,D(A_{re})\bigr)$ when
$D(A_{re})$ carries the graph norm of $A_{dl}$.

Since $A_{re}=A_{dl}+P_\lambda$ with $P_\lambda\in\Lscr(V)$, the operators
$A_{dl}$ and $A_{re}$ have the same domain $D(A_{re})$ and equivalent graph
norms on it. Hence $(\lambda I-A_{dl})^{-1}\in\Lscr\bigl(V,D(A_{re})\bigr)$ also
when $D(A_{re})$ is endowed with the graph norm of $A_{re}$.

We next show that $-P_\lambda$ is $(\lambda I-A_{dl})$-compact. Let $\sbm{u_n\\ v_n}\in D(A_{re})$ be such that $\sbm{u_n\\ v_n}$ and $(\lambda I-A_{dl})\sbm{u_n\\ v_n}$ are bounded in $V$. By
\eqref{eq:resolvent-est-shifted}, $\sbm{u_n\\ v_n}$ is then bounded in $H^2(0,1)\times H^1(0,1)$, so $(u_n)_{n\in\nline}$ is bounded in $H^2(0,1)$. Since the embedding $H^2(0,1)\hookrightarrow H^1(0,1)$ is compact, a subsequence of $(u_n)_{n\in\nline}$ converges in $H^1(0,1)$.
As
\[
(-P_\lambda)\bbm{u_n\\ v_n}
=
\bbm{
(\lambda-\mu_0)u_n\\[1mm]
-\dfrac{\gamma}{\eta}u_n
},
\]
the corresponding subsequence of $\bigl((-P_\lambda)\sbm{u_n\\ v_n}\bigr)_{n\in\nline}$ converges in $V$.
Thus $-P_\lambda$ is $(\lambda I-A_{dl})$-compact \cite[Chapter~IV, Section~1]{Kato}.

Since $A_{re}$ is closed and $P_\lambda\in\Lscr(V)$, the identity $A_{dl}=A_{re}-P_\lambda$ on $D(A_{re})$ implies that $A_{dl}$, and hence $\lambda I-A_{dl}$, is closed. We have already shown that $\lambda I-A_{dl}\colon D(A_{re})\subset V\to V$ is bijective. Therefore, $\lambda I-A_{dl}$ is boundedly invertible and, in particular, Fredholm with index zero. Since $-P_\lambda$ is $(\lambda I-A_{dl})$-compact, by
\cite[Chapter~IV, Section~5, Theorem~5.26]{Kato},
\[
\lambda I-A_{re}
=
(\lambda I-A_{dl})+(-P_\lambda)
\]
is semi-Fredholm and
\[
\operatorname{ind}(\lambda I-A_{re})
=
\operatorname{ind}(\lambda I-A_{dl})
=
0.
\]
Since $\lambda I-A_{re}$ is semi-Fredholm with finite index zero, it is Fredholm. In particular, $\operatorname{Ran}(\lambda I-A_{re})$ is closed in $V$. 

Since $A_{re}$ is closed and densely defined in $V$ by Lemma~\ref{lem:density},
and since $\lambda$ is isolated in $\sigma(A_{re})$ with
$\dim\ker(\lambda I-A_{re})=1$, the Fredholm property of $\lambda I-A_{re}$
gives the remaining closed-range condition required in
\cite[Theorem~4.1]{Baxley:1972}. Hence $\lambda$ has finite algebraic
multiplicity for $A_{re}$. By \eqref{eq:algebraic-multiplicity-A-Are}, the same
holds for $A$. Since $\lambda\in\sigma_p(A)$ was arbitrary, the proof is
complete.
\end{proof}

\begin{remark}\label{rem:omega-ge-omega0}
Suppose that $\omega\geq\omega_0$, and note that $B\in\Lscr(\cline,X)$ is of finite rank,
hence compact. If $\omega>\omega_0$, then Theorem~\ref{thm:point-spectrum} shows that
$A+\omega I$ has infinitely many eigenvalues in $\cline_0^{+}$, so by
\cite[Theorem~5.2.6]{CuZw:1995} there is no $K_\omega\in\Lscr(X,\cline)$ for which
$A+\omega I+BK_\omega$ is exponentially stable.

If $\omega=\omega_0$, then by Theorems~\ref{thm:point-spectrum} and~\ref{thm: complete spectrum} the point $0$ belongs to $\sigma(A+\omega_0 I)$ but not to $\sigma_p(A+\omega_0 I)$, and lies on the imaginary axis, so the previous argument does not apply. Suppose, for
contradiction, that $A+\omega_0 I+BK_{\omega_0}$ is exponentially stable for some
$K_{\omega_0}\in\Lscr(X,\cline)$. Then $(A+\omega_0 I,B)$ is exponentially stabilizable,
and \cite[Theorem~5.2.3]{CuZw:1995} gives a constant $\delta<0$ such that
$\sigma(A+\omega_0 I)\cap\cline_\delta^{+}$ consists only of eigenvalues. Since
$0\in\cline_\delta^{+}$, the point $0$ would then be an eigenvalue of $A+\omega_0 I$, a
contradiction. Therefore Problem~\ref{prob:w-stab} has no solution for $\omega\geq\omega_0$.
\end{remark}


\subsection{\textbf{$A$ generates an analytic semigroup}}\label{sec:analytic-semigroup}
Recall $X=L^2(0,1)\times H^1(0,1)$ and $V=H^1(0,1)\times H^1(0,1)$. Consider the sesquilinear form $a:V\times V\to \mathbb{C}$ defined as follows
\begin{align}\label{eq:sesquilinear form}
a\left(\bbm{y\\z},\bbm{p\\q}\right)
=\langle \eta y_\xi+z_\xi-\nu y,\;p_\xi \rangle_{L^2(0,1)}
-\langle \gamma y-\kappa z,\;q\rangle_{H^1(0,1)}
\end{align}
for all $\bbm{y & z}^{\top},\bbm{p &q}^{\top}\in V$. It is easy to check that $a$ is continuous
\begin{align}\label{eq:continous estimate}
|a(v_1,v_2)|\leq \left(1+\eta+\gamma+\nu+\kappa\right)\|v_1\|_{V}\|v_2\|_{V}\FORALL v_1,v_2\in V.
\end{align}
A straightforward calculation yields G\aa rding's inequality
\begin{equation}\label{eq:gardings-equality}
\Re\, a(v,v)
+\left(\frac{\nu^2}{\eta}+\frac{\gamma^2}{2\kappa}+\frac{(1-\gamma)^2}{\eta}+\frac{\eta+\kappa}{2}\right)\|v\|_X^2 \geq
\min\left\{\frac{\eta}{4},\frac{\kappa}{2}\right\}\|v\|_V^2
\end{equation}
for all $v\in V$. The operator $\hat{A}$ associated with the sesquilinear form $a$ is defined as follows:
$$
D(\hat{A})=\left\{u \in V \mid a(u, v)=\langle w, v\rangle_X \text { for some } w \in X \text { and all } v \in V\right\}
$$
and for $u \in D(\hat{A})$ and $w$ such that $a(u, v)=\langle w, v\rangle_X$ for all $v \in V$ we let $\hat{A} u=-w$. Because the sesquilinear form $a$ is continuous and satisfies G\aa rding's inequality, it follows that $\hat{A}$ generates an analytic semigroup on $X$. In the next theorem, we show that $\hat{A}$ is the same as $A$ and that it generates an analytic semigroup on $X$.
\begin{framed}
\begin{theorem}\label{thm:analytic-semigroup}
The state operator $A$ in \eqref{eq:Domain-A}-\eqref{eq:operator-A} is the operator associated with the sesquilinear form $a$ in \eqref{eq:sesquilinear form} and it generates an analytic semigroup $\mathbb{T}$ on $X$.
\end{theorem}
\end{framed}

\begin{proof}
Let $\bbm{y & z}^{\top} \in D(A)$. Then for all $\bbm{p &  q}^{\top} \in V$ we have
\begin{align*}
\left\langle -A\bbm{y\\ z},\bbm{p\\ q}\right\rangle_X
&= -\big\langle (\eta y+z)_{\xi\xi} - \nu y_\xi,\,p\big\rangle_{L^2(0,1)}
   - \big\langle -\kappa z + \gamma y,\,q\big\rangle_{H^1(0,1)} \\
&= \big\langle \eta y_\xi + z_\xi - \nu y,\,p_\xi\big\rangle_{L^2(0,1)}
   - \big\langle \gamma y - \kappa z,\,q\big\rangle_{H^1(0,1)} \\
&= a\left(\bbm{y\\ z},\bbm{p\\ q}\right).
\end{align*}
In the above calculations we used integration by parts and the boundary conditions $(\eta y+z)_\xi(0)=\nu y(0)$ and $(\eta y+z)_\xi(1)=\nu y(1)$. Hence the boundary term vanishes. It now follows from the definition of $\hat{A}$ that $D(A)\subset D(\hat{A})$ and $\hat{A}u=Au$ for all $u\in D(A)$.

Fix $\lambda>0$ large enough such that $\lambda\in\rho(\hat{A})\cap\rho(A)$. Then,
$$
(\lambda I-\hat{A}) D(\hat{A})
= X
= (\lambda I-A) D(A)
= (\lambda I-\hat{A}) D(A).
$$
Applying $(\lambda I-\hat{A})^{-1} \in \mathcal{L}(X)$ to the first and last terms above we get $D(A)=D(\hat{A})$.
We have already shown that $A u=\hat{A} u$ for all $u \in D(A)$. Therefore $A$ is the same as the operator $\hat{A}$
associated with the sesquilinear form $a$ and it is the generator of the analytic semigroup $\mathbb{T}$ on $X$.
\end{proof}


\subsection{\textbf{$\omega$-stabilizability of $(A,B)$ for all $\omega<\omega_0$}}\label{sec:stab-A-B}
\ \ \ Recall the state operator $A$ defined in \eqref{eq:Domain-A}-\eqref{eq:operator-A}, the constants $\eta,\gamma,\kappa,\nu>0$, and the control operator $B$ in \eqref{eq:operator-B}. Set $\omega_0=\kappa+\frac{\gamma}{\eta}$. We call the pair $(A,B)$ $\omega$-stabilizable if there exists $K\in\mathcal{L}(X,\mathbb{C})$ such that the closed-loop semigroup $\mathbb{T}^{cl}$ generated by $A+BK$ satisfies \eqref{eq:problem} for some $M,\epsilon>0$. To apply the Hautus test, we first determine the adjoint operator $A^*$. We then prove in Theorem~\ref{thm:omega-stab} that $(A,B)$ is $\omega$-stabilizable for every $0<\omega<\omega_0$, provided $B$ satisfies a suitable non-orthogonality condition.

Consider the sesquilinear form $a^* : V \times V \to \mathbb{C}$
\begin{equation}\label{eq:adjoint-form}
a^*\!\left(
\begin{bmatrix} y \\ z \end{bmatrix},
\begin{bmatrix} p \\ q \end{bmatrix}
\right)
=
\left\langle y_\xi,\; \eta p_\xi + q_\xi - \nu p \right\rangle_{L^2(0,1)}
-
\left\langle z,\; \gamma p - \kappa q \right\rangle_{H^1(0,1)}
\quad \forall
\begin{bmatrix} y \\ z \end{bmatrix},
\begin{bmatrix} p \\ q \end{bmatrix} \in V.
\end{equation}
Since $a^*(u,v)=\overline{a(v,u)}$ for all $u,v\in V$, the continuity estimate \eqref{eq:continous estimate} and the G\aa rding inequality \eqref{eq:gardings-equality} derived for $a$ apply for $a^*$ as well. 
\begin{equation}\label{eq:adjoint-continuity}
\bigl| a^*(v_1, v_2) \bigr|
\le
(1 + \eta + \gamma + \nu + \kappa)\, \|v_1\|_V \, \|v_2\|_V
\FORALL v_1, v_2 \in V,
\end{equation}
\begin{equation}\label{eq:gardings-adjoint}
\Re a^{*}(v,v)
+\left(\frac{\nu^2}{\eta}+\frac{\gamma^2}{2\kappa}+\frac{(1-\gamma)^2}{\eta}+\frac{\eta+\kappa}{2}\right)\|v\|_X^2 \geq 
\min\left\{\frac{\eta}{4},\frac{\kappa}{2}\right\}\|v\|_V^2
\end{equation}
for all $v\in V$.
\begin{framed}
\begin{lemma}\label{lem:adjoint-A}
Define the operator $A^*:D(A^*)\subset X\to X$ by
\begin{equation}\label{eq:Domain-Astar}
\begin{aligned}
D(A^*)=\left\{\bbm{y\\ z}\in V \,\middle|\,
\begin{array}{l}
\eta y-\gamma z \in H^{2}(0,1),\\[1mm]
(\eta y-\gamma z)_\xi(0)=0,\\
(\eta y-\gamma z)_\xi(1)=0
\end{array}
\right\},
\end{aligned}
\end{equation}
and
\begin{equation}\label{eq:A-star-def}
A^*\bbm{y\\ z}
=
\bbm{
(\eta y - \gamma z)_{\xi\xi} + \nu y_\xi + \gamma z\\
\psi_{y,z}
}
\FORALL  \bbm{y\\ z} \in D(A^*).
\end{equation}
Here $\psi_{y,z}\in H^1(0,1)$ is the unique solution of
\begin{equation}\label{eq:adjoint-weak-comp}
\left\langle \psi_{y,z}, \phi \right\rangle_{H^1(0,1)}
=
-\kappa \langle z, \phi \rangle_{L^2(0,1)}
-
\left\langle y_\xi + \kappa z_\xi, \phi_\xi \right\rangle_{L^2(0,1)} \quad \forall \phi \in H^1(0,1).
\end{equation}
The operator $A^*$ is the adjoint of the state operator $A$ defined in \eqref{eq:Domain-A}-\eqref{eq:operator-A}.
\end{lemma}
\end{framed}

\begin{proof}
The operator $A^*$ is well defined since it maps $D(A^*)$ into $X$. Indeed, if $\sbm{y\\ z}\in D(A^*)$ then $\eta y-\gamma z\in H^2(0,1)$ so $(\eta y-\gamma z)_{\xi\xi}\in L^2(0,1)$ and $y,z\in H^1(0,1)$ implies $y_\xi,z\in L^2(0,1)$ hence $(\eta y-\gamma z)_{\xi\xi}+\nu y_\xi+\gamma z\in L^2(0,1)$. Also, for fixed $y,z\in H^1(0,1)$ the right-hand side of \eqref{eq:adjoint-weak-comp} defines a bounded linear functional on $H^1(0,1)$ so by the Riesz representation theorem there exists a unique $\psi_{y,z}\in H^1(0,1)$ satisfying \eqref{eq:adjoint-weak-comp}.
The operator $\tilde{A}$ associated with the sesquilinear form $a^*$ is defined as follows:
$$
D(\tilde{A})
=
\left\{
u \in V \ \middle| \
a^*(u, v)=\langle w, v\rangle_X 
\text{ for some } w \in X \text{ and all } v \in V
\right\},
$$
for all $u \in D(\tilde{A})$ and $w$ such that $a^*(u, v)=\langle w, v\rangle_X$ for all $v \in V$, we define $\tilde{A}u=-w$.  
On the other hand, by the definition of $A^*$ and an integration by parts argument, one has $a^*(u,v)=-\langle A^*u, v\rangle_X$ for all $u\in D(A^*)$ and all $v\in V$.  
Comparing these two representations of $a^*(u,v)$, we deduce that $A^*u=\tilde{A}u$ for all $u\in D(A^*)$, and hence $D(A^*)\subset D(\tilde{A})$.
We now prove the reverse inclusion $D(\tilde{A}) \subset D(A^*)$.  
Let $\sbm{y\\ z} \in D(\tilde{A})$. Then there exists $\sbm{f\\g} \in X$ such that
$$
a^*\!\left(\bbm{y\\ z}, \bbm{p\\ 0}\right)
=
\left\langle -\bbm{f\\ g}, \bbm{p\\ 0} \right\rangle_X
\qquad \forall p \in H^1(0,1).
$$
Expanding this identity using \eqref{eq:adjoint-form} and rearranging terms, we obtain
\begin{equation}\label{eq:weakform}
\left\langle(\eta y-\gamma z)_\xi, p_\xi\right\rangle_{L^2(0,1)}=\left\langle\nu y_\xi+\gamma z-f, p\right\rangle_{L^2(0,1)}
\FORALL p\in H^1(0,1).
\end{equation}
Since the above identity holds for all $p\in H^1(0,1)$, it follows by definition of weak derivatives that
$\eta y-\gamma z\in H^2(0,1)$. Moreover, by choosing test functions $p\in H^1(0,1)$ with arbitrary boundary traces
at $\xi=0$ and $\xi=1$ and the absence of boundary terms in \eqref{eq:weakform} forces
$(\eta y-\gamma z)_\xi(0)=0$ and $(\eta y-\gamma z)_\xi(1)=0$. Next, taking test functions of the form
$\sbm{0\\ q}$ with arbitrary $q\in H^1(0,1)$ identifies the second component via the relation
$\langle g,q\rangle_{H^1}=-\kappa\langle z,q\rangle_{L^2}-\langle y_\xi+\kappa z_\xi,q_\xi\rangle_{L^2}$, hence
$g=\psi_{y,z}$ by uniqueness in $H^1(0,1)$. Hence $\sbm{y\\ z}\in D(A^*)$, and therefore
$D(\tilde{A})\subset D(A^*)$. Combined with the previously established inclusion
$D(A^*)\subset D(\tilde{A})$, we conclude that $D(A^*)=D(\tilde{A})$ and $A^*=\tilde{A}$. Finally, since
$a^*(u,v)=\overline{a(v,u)}$ for all $u,v\in V$, it follows that $\langle Au,v\rangle_X=\langle u,A^*v\rangle_X$
for all $u\in D(A)$ and $v\in D(A^*)$, which completes the proof that $A^*$ is the adjoint of $A$.
\end{proof}

It is easy to check that the adjoint of $B$ defined in \eqref{eq:operator-B} is the operator
\begin{equation}\label{eq:adjoint-B}
    B^{*}\bbm{p\\q}=\langle p,b \rangle_{L^2(0,1)} \FORALL \bbm{p\\q}\in X.
\end{equation}
We now present conditions for the $\omega$-stabilizability of the pair $(A, B)$.
\begin{framed}
\begin{theorem}\label{thm:omega-stab}
Let $\omega_0=\kappa+\gamma/\eta$ and fix $0<\omega<\omega_0$. Assume that
$$
\int_{0}^{1} \overline{b(\xi)}\,d\xi\neq 0,
$$
and set
$$
\mathcal{J}
=\Bigl\{(k,j)\in\{1,2,3\}\times\mathbb{N}\ \Bigm|\ \Re(\lambda_{k,j})\in[-\omega,0]\Bigr\}.
$$
For each $(k,j)\in\mathcal{J}$, define
$$
\alpha_{k,j}=-\frac{\nu(\lambda_{k,j}+\kappa)}{\eta(\lambda_{k,j}+\kappa)+\gamma}.
$$
Suppose that, for every $(k,j)\in\mathcal{J}$, the following condition also holds
$$
\int_{0}^{1} \overline{b(\xi)}\,e^{\frac{\alpha_{k,j}}{2}\xi}
\left(\cos(\pi j \xi)-\frac{\alpha_{k,j}}{2\pi j}\sin(\pi j \xi)\right)\,d\xi\neq 0.
$$
Then the pair $(A,B)$ is $\omega$-stabilizable.
\end{theorem}
\end{framed}

\begin{proof}
Let $\omega_0$ and $\omega$ be as in the theorem and define
$$
A_\omega = A+\omega I.
$$
Since $A$ generates an analytic semigroup on $X$ by Theorem~\ref{thm:analytic-semigroup}, the shifted operator $A_\omega$ also generates an analytic semigroup on $X$. Moreover, by Theorems~\ref{thm:point-spectrum}, \ref{thm: complete spectrum}, and \ref{thm:finite-multi}, the set
$$
\sigma_\omega^+
:=
\{\lambda+\omega \mid \lambda\in\sigma(A),\ \Re\lambda\ge -\omega\}
$$
is a finite set of eigenvalues of $A_\omega$, each of finite algebraic multiplicity. Denoting the remainder of the spectrum of $A_\omega$ by $\sigma_\omega^-$, we have $\Re\mu<0$ for every $\mu\in\sigma_\omega^-$. Since the only accumulation point of $\sigma(A_\omega)$ is $\omega-\omega_0<0$ and the set $\sigma_\omega^+$ is finite, it follows that there exists $\delta>0$ such that $\Re\mu<-\delta$ for all $\mu\in\sigma_\omega^-$.
In particular,
\begin{equation}\label{eq:specdet}
\lambda\in \sigma^-_\omega \implies \Re\lambda<-\delta. \vspace{-1mm}
\end{equation}
Let $\tline_\omega$ denote the analytic semigroup generated by $A_\omega$. Let $\Pi^-_\omega$ be the spectral projection on $\sigma^-_\omega$. Define $X^-_\omega=\Pi^-_\omega X$ and let $A^-_\omega$ and $\tline^-_\omega$ be the restrictions of $A_\omega$ and $\tline_\omega$, respectively, to $X^-_\omega$. From \cite[Lemma 2.5.7]{CuZw:1995} we get that the spectrum of $A^-_\omega$ is $\sigma^-_\omega$ and that $A^-_\omega$ is the generator of the strongly continuous semigroup $\tline^-_\omega$. Furthermore, $\tline^-_\omega$ is an analytic semigroup since it is the restriction of an analytic semigroup. Hence it satisfies the spectrum determined growth condition \cite[Part II, Chapter 1, Corollary 2.5]{BePrDeMi:2007} and so using \eqref{eq:specdet} we get \vspace{-1mm}
$$\|\tline^-_\omega(t)\|\leq C e^{-\delta t} \FORALL t\geq0 \vspace{-1mm}$$
and some $C>0$. From the above discussion it follows that the pair $(A_\omega, B)$ satisfies all the assumptions required for applying \cite[Part V, Chapter 1, Proposition 3.3]{BePrDeMi:2007}. 

We now verify the second condition in that result, namely
\begin{equation}\label{eq:conditons-verified}
\ker(\mu I-A_\omega^*)\cap \ker(B^*)=\{0\}
\FORALL \mu\in\sigma_\omega^{+}.
\end{equation}
Fix $\mu\in\sigma_\omega^{+}$ and consider 
\begin{equation}\label{eq:condition-2}
(\mu I-A_\omega^*)\bbm{f\\g}=0.
\end{equation}
Since $\mu\in\sigma_\omega^{+}$, we can write $\mu=\lambda+\omega$ for some eigenvalue $\lambda$ of $A$ with its real part in $[-\omega,0]$. Substituting $A_\omega^*=A^*+\omega I$ in \eqref{eq:condition-2} and using the definition of $A^*$ in \eqref{eq:Domain-Astar}--\eqref{eq:A-star-def}, we obtain
\begin{align}
&(\eta f-\gamma g)_{\xi\xi}(\xi)+\nu f_\xi(\xi)+\gamma g(\xi)=\lambda f(\xi) \qquad \text{for a.e. } \xi\in(0,1), \label{eq:adjoint one}\\
&\langle \lambda g,\phi\rangle_{H^1(0,1)}
= -\kappa\langle g,\phi\rangle_{L^2(0,1)}
-\left\langle f_\xi+\kappa g_\xi,\phi_\xi\right\rangle_{L^2(0,1)}
\FORALL \phi\in H^1(0,1), \label{eq:adjoint two}
\end{align}
with the boundary conditions
\begin{equation}\label{eq:adjoint-bound}
(\eta f-\gamma g)_\xi(0)=0,
\qquad
(\eta f-\gamma g)_\xi(1)=0.
\end{equation}
Because \eqref{eq:adjoint two} holds for all test functions $\phi\in C_c^\infty(0,1)$ and $g\in H^1(0,1)\subset L^2(0,1)$, it follows that $f_\xi+(\lambda+\kappa)g_\xi\in H^1(0,1)$. Furthermore,
\begin{equation}\label{eq:stab-ODE}
\bigl(f_\xi(\xi)+(\lambda+\kappa)g_\xi(\xi)\bigr)_\xi=(\lambda+\kappa)g(\xi) \qquad \text{for a.e. } \xi\in(0,1).
\end{equation}
Integrating \eqref{eq:adjoint two} by parts and substituting \eqref{eq:stab-ODE} we get that a boundary term
$\bigl(f_\xi(1)+(\lambda+\kappa)g_\xi(1)\bigr)\phi(1)-\bigl(f_\xi(0)+(\lambda+\kappa)g_\xi(0)\bigr)\phi(0)$, which must vanish for all
$\phi\in H^1(0,1)$. Since $\phi(0)$ and $\phi(1)$ can be prescribed independently, we obtain the boundary conditions for \eqref{eq:stab-ODE}
\begin{equation}\label{eq:stab-boud}
\bigl(f_\xi(0)+(\lambda+\kappa)g_\xi(0)\bigr)=0,
\qquad
\bigl(f_\xi(1)+(\lambda+\kappa)g_\xi(1)\bigr)=0.
\end{equation}
Since $\Re\lambda\in[-\omega,0]$, we have $\eta(\lambda+\kappa)+\gamma\neq 0$. Using this fact together with $f_\xi+(\lambda+\kappa)g_\xi\in H^1(0,1)$ and $\eta f-\gamma g\in H^2(0,1)$, a simple calculation shows that $f,g\in H^2(0,1)$. 

Next, we eliminate $g$ to obtain an ODE in terms of $f$. Using \eqref{eq:stab-ODE}, we express $g_{\xi\xi}$ in terms of $g$ and $f_{\xi\xi}$. Substituting this identity into \eqref{eq:adjoint one} cancels the $\gamma g$ terms,
and yields the second-order ODE in $f$.
\begin{equation}\label{eq:f-ODE}
f_{\xi \xi}(\xi)-\alpha_\lambda f_\xi(\xi)+\beta_\lambda f(\xi)=0 \qquad \text{for a.e. } \xi\in(0,1).
\end{equation}
Here,
\begin{equation}\label{eq:coefficieents-stabilization}
\alpha_\lambda=-\frac{\nu(\lambda+\kappa)}{\eta(\lambda+\kappa)+\gamma},
\qquad
\beta_\lambda=-\frac{\lambda(\lambda+\kappa)}{\eta(\lambda+\kappa)+\gamma}.
\end{equation}
The boundary conditions \eqref{eq:adjoint-bound} and \eqref{eq:stab-boud} couple $f$ and $g$. Using
$\eta(\lambda+\kappa)+\gamma\neq 0$, we eliminate $g_\xi(0)$ and $g_\xi(1)$ from these relations and obtain the
boundary conditions for \eqref{eq:f-ODE} to be
\begin{equation}\label{eq:f-bc}
f_\xi(0)=0, \qquad f_\xi(1)=0.
\end{equation}
We now consider the following cases.
\begin{itemize}
\item[Case 1.] If $\lambda=0$, then for $\mu=\omega\in[0,\omega]$. Solving \eqref{eq:f-ODE}--\eqref{eq:f-bc} gives the constant solution
$f= 1$ for all $\xi\in(0,1)$. Moreover,
$$
B^{*}\bbm{f\\g}=\langle f,b\rangle_{L^2(0,1)}=\int_{0}^{1} \overline{b(\xi)}\,d\xi\neq 0.
$$
Therefore, $\sbm{f\\g}\notin\ker(B^{*})$. Hence, the condition \eqref{eq:conditons-verified} holds for $\mu=\omega$.
\item[Case 2.] If $\lambda=-\kappa$ and $\mu=\omega-\kappa\in\sigma_\omega^+$, then \eqref{eq:f-ODE}--\eqref{eq:f-bc} again admits a constant solution, which we normalize by taking $f(\xi)=1$ for all $\xi\in(0,1)$. The same argument as in Case 1 then
shows that $B^{*}\sbm{f\\g}\neq 0$. Hence, \eqref{eq:conditons-verified} holds for $\mu=\omega-\kappa$.
\item[Case 3.]Suppose that $\mu=\lambda_{k,j}+\omega$ with $\Re\lambda_{k,j}\in[-\omega,0]$. The coefficients in \eqref{eq:coefficieents-stabilization} associated with
$\lambda_{k,j}$ satisfy
$\sqrt{\alpha_{\lambda_{k,j}}^{\,2}-4\beta_{\lambda_{k,j}}}=2\pi i\,j$, see Case 3 in the proof of Theorem~\ref{thm:point-spectrum}. Solving \eqref{eq:f-ODE} together with \eqref{eq:f-bc} yields the normalized solution $f=f_{k,j}$ with 
$$
f_{k,j}(\xi)=e^{\frac{\alpha_{\lambda_{k,j}}}{2}\xi}
\left(\cos(\pi j \xi)-\frac{\alpha_{\lambda_{k,j}}}{2\pi j}\sin(\pi j \xi)\right)
\FORALL \xi\in(0,1).
$$
By the standing hypothesis of the theorem, we have $\langle f_{k,j},b\rangle_{L^2}\neq 0$ for all $(k,j)\in \mathcal{J}$. Therefore, \eqref{eq:conditons-verified} holds for $\mu=\lambda_{k,j}+\omega$.
\end{itemize}
We have shown that $(A_\omega,B)$ satisfies the hypotheses of \cite[Part~V, Chapter~1, Proposition 3.3]{BePrDeMi:2007} and that
\eqref{eq:conditons-verified} holds. Therefore the stabilizability criterion in that proposition applies, and $(A_\omega,B)$ is stabilizable.
Consequently, there exists $K\in\mathcal{L}(X,\mathbb{C})$ such that the semigroup generated by $A_\omega+BK=A+\omega I+BK$ is exponentially stable. It follows that the closed-loop semigroup $\mathbb{T}^{cl}$ generated by $A+BK$ satisfies \eqref{eq:problem} for some $M,\epsilon>0$. This completes the proof.
\end{proof}

\section{LQR controller design \vspace{-1mm}} \label{sec4} \setcounter{equation}{0} 

\ \ \ Recall the operators $A$ and $B$ from Section \ref{sec:Math-model-problem}. Consider the abstract evolution equation
\begin{equation} \label{eq:HM_absomega}
 \bbm{\dot{y}(t)\\ \dot{z}(t)}=A_\omega\bbm{y(t)\\z(t)}+Bu(t) \FORALL t>0
\end{equation}
on the state space $X=L^2(0,1)\times H^1(0,1)$, where $A_\omega=A+\omega I$ for some constant $\omega>0$. Equation \eqref{eq:HM_absomega} is obtained from \eqref{eq:abstract-evolution} by replacing $A$ with $A_\omega$. Since $A$ generates an analytic semigroup $\tline$ on $X$, the operator $A_\omega$ generates the analytic semigroup $\tline_\omega$ on $X$, given by $\tline_\omega(t)=e^{\omega t}\tline(t)$ for all $t\geq0$. For every input $u\in L^2((0,\infty);\cline)$ and initial state $\sbm{y^0 \\ z^0}\in X$, the corresponding mild solution $\sbm{y \\ z}\in C([0,\infty);X)$ of \eqref{eq:HM_absomega} is
$$
 \bbm{y(t) \\ z(t)} =\tline_\omega(t) \bbm{y^0\\ z^0} + \int_0^t \tline_\omega(t-s) B u(s)\dd s \FORALL t\geq0. \vspace{-1mm}
$$
For \eqref{eq:HM_absomega} we consider the quadratic cost functional $J:L^2((0,\infty);\cline)\times X\to \rline$ defined by
\begin{equation}\label{eq:Cost}
  J\left(u,\bbm{y^0\\z^0}\right)=\int_{0}^{\infty} \left( \left\langle \bbm{y(t)\\z(t)}, Q\bbm{y(t)\\z(t)} \right\rangle_{X}+R|u(t)|^2 \right) \dd t.
\end{equation}
Here $Q\in \Lscr(X)$ is self-adjoint and coercive, $R\in \rline$ is a positive scalar, and $\sbm{y \\ z}$ denotes the mild solution of \eqref{eq:HM_absomega} corresponding to $u$ and $\sbm{y^0 \\ z^0}$. In particular, $Q$ and $R$ are boundedly invertible and admit coercive square roots which are also boundedly invertible. The algebraic Riccati equation associated with minimizing $J$ is
\begin{equation} \label{eq:Riccati}
 A_\omega^{*} \Pi+\Pi A_\omega - \Pi B R^{-1} B^{*}\Pi + Q = 0.
\end{equation}
Recall that $\omega_0=\kappa+\frac{\gamma}{\eta}$. Let $0<\omega<\omega_0$ and assume that $B$ satisfies the hypothesis in Theorem \ref{thm:omega-stab}. Then, by Theorem \ref{thm:omega-stab}, there exists $K\in\Lscr(X,\cline)$ such that the semigroup generated by $A_\omega+BK$ is exponentially stable. Fix any $\varepsilon>0$. Moreover, the semigroup generated by $A_\omega+LQ^{\frac{1}{2}}$ is exponentially stable, where $Q^{\frac{1}{2}}$ is the coercive square root of $Q$ and $L=-(\omega+\varepsilon)Q^{-\frac{1}{2}} \in\Lscr(X)$, because
$$
A_\omega+LQ^{\frac{1}{2}}=A-\varepsilon I.
$$
Since $\sigma(A)\subset\{\lambda\in\cline\mid \Re\lambda\le 0\}$, it follows that $\sigma(A-\varepsilon I)\subset\{\lambda\in\cline\mid \Re\lambda\le -\varepsilon\}$. As $A-\varepsilon I$ generates an analytic semigroup, it is exponentially stable. Hence \eqref{eq:Riccati} admits a unique nonnegative solution $\Pi\in\Lscr(X)$ and $A_\omega-BR^{-1} B^*\Pi$ generates an exponentially stable semigroup, see \cite[Theorem 6.2.7]{CuZw:1995}. Consequently, $K_\infty=-R^{-1}B^*\Pi$ stabilizes the pair $(A+\omega I,B)$. Thus a feedback $K_\omega$ solving Problem \ref{prob:w-stab} can be obtained by computing the nonnegative solution $\Pi$ of \eqref{eq:Riccati}.

We now outline a practical approach for approximating $\Pi$. In Section~\ref{sec4.1} we introduce, for each $n$, finite-dimensional approximations of the dynamics in \eqref{eq:HM_absomega} and of the cost functional $J$ in \eqref{eq:Cost}, together with the corresponding finite-dimensional algebraic Riccati equation. In Section~\ref{sec4.2} we show that these approximations are uniformly (in $n$) stabilizable, and, by \cite[Theorem~2.2]{BaKu:1984}, their nonnegative solutions $\Pi_n$ converge strongly to the nonnegative solution $\Pi$ of \eqref{eq:Riccati} as $n\to\infty$, see Theorem~\ref{thm:main-result}. The uniform stabilizability is the crucial ingredient here. In the compact-resolvent setting of \cite{BaKu:1984} it follows directly, but the operator $A$ in our problem does not have a compact resolvent, and we therefore establish it by adapting the arguments from \cite[Propositions~4.2 and~4.3]{SiAkMiNa:2025}. The resulting gain $K_{\omega}=-R^{-1}B^{*}\Pi_n$ solves Problem~\ref{prob:w-stab} once $n$ is sufficiently large, and since $\Pi_n$ is straightforward to compute, so is $K_{\omega}$, see the example in Section~\ref{sec:numerical-example}.

\subsection{Finite-dimensional approximations \vspace{0mm}} \label{sec4.1} 

\ \ \ Let $(V_n)_{n\in\nline}$ be a sequence of finite-dimensional subspaces of $V=H^1(0,1)\times H^1(0,1)$ with $V_n=H_n \times H_n$ for each $n\in\nline$, where $H_n$ is a finite-dimensional subspace of $H^1(0,1)$. We regard $V_n$ as a subspace of $X$ and equip it with the inner product and norm inherited from $X$. Recall the sesquilinear form $a:V\times V\to\cline$ from \eqref{eq:sesquilinear form} associated with the state operator $A$ in \eqref{eq:Domain-A}--\eqref{eq:operator-A}. The approximation $A_n$ of $A$ is defined by restricting $a$ to $V_n\times V_n$, i.e., $A_n\in\Lscr(V_n)$ is determined by
\begin{equation}\label{eq:A-approx}
 \langle -A_n v_1,v_2 \rangle_X = a(v_1,v_2) \FORALL v_1,v_2\in V_n.
\end{equation}
Let $P_n$ be the orthogonal projection from $X$ onto $V_n$. Recall $B$ from \eqref{eq:operator-B}. We define $B_n\in\Lscr(\cline ,V_n)$ by
$$
 B_n \alpha= \alpha P_n\bbm{ b \\ 0} \FORALL \alpha\in \cline.
$$
Using $A_n$ and $B_n$, we obtain the finite-dimensional approximation of \eqref{eq:HM_absomega} as the linear ordinary differential equation
\begin{equation} \label{eq:approx system}
 \bbm{\dot{y_n}(t) \\ \dot{z_n}(t)} = A_{\omega,n} \bbm{y_n(t) \\ z_n(t)} + B_n u(t) \FORALL t>0
\end{equation}
on the state space $V_n$, where $A_{\omega,n}= A_n+\omega I$. Clearly $y_n(t)$ and $z_n(t)$ belong to $H_n$.

Recall $Q$ and $R$ from \eqref{eq:Cost} and set $Q_n=P_n Q P_n|_{V_n}$. Then $Q_n\in\Lscr(V_n)$ is self-adjoint and coercive. Define the associated quadratic cost functional $J_n:L^2((0,\infty);\cline)\times V_n\to \rline$ by
\begin{equation}\label{eq:Cost_approx}
  J_n\left(u,\bbm{y^0_n\\z^0_n}\right)=\int_{0}^{\infty} \left( \left\langle \bbm{y_n(t)\\z_n(t)}, Q_n\bbm{y_n(t)\\z_n(t)} \right\rangle_{X}+R|u(t)|^2 \right) \dd t.
\end{equation}
Here $\sbm{y_n \\ z_n}$ denotes the solution of \eqref{eq:approx system} corresponding to $u$ and $\sbm{y_n^0 \\ z_n^0}\in V_n$. Thus $J_n$ is the finite-dimensional approximation of $J$ associated with the dynamics in \eqref{eq:approx system}. The algebraic Riccati equation associated with minimizing $J_n$ is
\begin{equation} \label{eq:Riccati_approx}
 A_{\omega,n}^{*} \Pi_n + \Pi_n A_{\omega,n} - \Pi_n B_n R^{-1} B_n^{*}\Pi_n + Q_n = 0.
\end{equation}
In Section \ref{sec4.2} we show that, under suitable conditions, \eqref{eq:Riccati_approx} admits a unique nonnegative solution $\Pi_n\in\Lscr(V_n)$ for large enough $n$ and that $\Pi_nP_n$ converges strongly to the unique nonnegative solution $\Pi\in \Lscr(X)$ of \eqref{eq:Riccati} as $n\to\infty$.
`
\subsection{Convergence of $\Pi_n$ to $\Pi$} \label{sec4.2} 

\ \ \ Recall $\omega_0=\kappa+\gamma/\eta$ and let $0<\omega<\omega_0$. Suppose that the control operator $B$ satisfies the hypothesis of Theorem~\ref{thm:omega-stab}. Then, by Theorem~\ref{thm:omega-stab}, there exists a feedback operator $K\in\Lscr(X,\cline)$ such that $A_\omega+BK$ generates an exponentially stable semigroup on $X$. Fix one such operator $K$. Define
\begin{equation} \label{eq:lenovo}
A_\omega^K = A_\omega+BK, \qquad A_{\omega,n}^K = A_{\omega,n}+B_n K P_n.
\end{equation}
In Proposition \ref{pr:approxeig} we show that if $\lambda_n$ is an eigenvalue of $A_{\omega,n}^K$ for each $n$, then every accumulation point of the sequence $(\lambda_n)_{n\in\nline}$ is an eigenvalue of $A_{\omega}^K$. Using this result we establish in Proposition \ref{pr:unifstab} that the pair $(A_{\omega,n},B_n)$ is uniformly (in $n$) stabilizable. Finally, appealing to \cite[Theorem 2.2]{BaKu:1984} we conclude that the nonnegative solution $\Pi_n$ of \eqref{eq:Riccati_approx} converges to the nonnegative solution $\Pi$ of \eqref{eq:Riccati} as $n\to\infty$. We also show that $K_\omega=-R^{-1}B_n^*\Pi_n P_n$ solves Problem \ref{prob:w-stab} if $n$ is sufficiently large.

We require the approximating subspaces $V_n$ to satisfy the following natural assumption:
\begin{framed}
\begin{assumption}\label{as:subspace}
 For every $v\in V$, there exists a sequence $(v_n)_{n\in\nline}$ in $V$ with $v_n\in V_n$ for all $n\in\nline$ such that
 $$\lim_{n\to \infty}\|v_n-v\|_V=0.$$
\end{assumption}
\end{framed}
The above assumption implies that 
\begin{equation}\label{eq:assumconv}
 \lim_{n\to \infty}\|P_n x - x\|_X=0 \FORALL x\in X. 
\end{equation}
Recall $a^{*}$ defined in \eqref{eq:adjoint-form}. Define the sesquilinear form ${a_\omega^K}^*:V\times V\to \cline$ as follows:
\begin{equation}\label{eq:sesq-AstarK}
 {a_\omega^K}^*(v_1,v_2) = a^{*}(v_1,v_2)-\langle (BK)^{*}v_1,v_2 \rangle_X - \omega\langle v_1,v_2 \rangle_{X} \FORALL v_1,v_2\in V.
\end{equation}
We have already shown that $A^*$ is the operator associated with the sesquilinear form $a^*$, see Lemma~\ref{lem:adjoint-A}. Since $BK\in\Lscr(X)$, the operator associated with the sesquilinear form $({a_\omega^K})^*$ coincides with the adjoint operator $(A+BK+\omega I)^*$. Using \eqref{eq:adjoint-continuity} and \eqref{eq:gardings-adjoint}, we get
\begin{equation}\label{eq:a_Kstar1}
\big|{a_\omega^K}^*(v_1, v_2)\big|
\leq C_1\, \|v_1\|_V \, \|v_2\|_V
\FORALL v_1,v_2\in V,
\end{equation}
and
\begin{equation}\label{eq:a_Kstar2}
\Re\,{a_\omega^K}^*(v,v)+C_2\,\|v\|^2_{X}
\geq \min\Big\{\frac{\eta}{4},\frac{\kappa}{2}\Big\}\,\|v\|^2_V
\FORALL  v\in V.
\end{equation}
Here
\begin{align*}
C_1=&1+\eta+\gamma+\nu+\kappa+\|BK\|_{\Lscr(X)}+\omega,\\
C_2=&\|BK\|_{\Lscr(X)}+\omega+\frac{\nu^2}{\eta}+\frac{\gamma^2}{2 \kappa}+\frac{(1-\gamma)^2}{\eta}+\frac{\eta+\kappa}{2}.
\end{align*}
Now observe from \eqref{eq:sesquilinear form} and \eqref{eq:adjoint-form} that
$$
a^*(v_1,v_2)=\overline{a(v_2,v_1)} \FORALL v_1,v_2\in V.
$$
Hence, for every $v_1,v_2\in V_n$, using \eqref{eq:A-approx} we get
$$
a^*(v_1,v_2)
=
\overline{a(v_2,v_1)}
=
-\overline{\langle A_n v_2,v_1\rangle_X}
=
-\langle v_1,A_n v_2\rangle_X
=
-\langle A_n^*v_1,v_2\rangle_X.
$$
Using this, \eqref{eq:sesq-AstarK} and the expression $A_{\omega,n}^K = A_n + B_n K P_n + \omega I$ we get
\begin{equation}
 {a_\omega^K}^*(v_1,v_2)=-\left\langle \left(A^K_{\omega,n}\right)^{*}v_1, v_2 \right\rangle_{X} \FORALL v_1,v_2\in V_n.
\end{equation}
Recall Assumption \ref{as:subspace}. Let $\zeta= C_2$. Applying \cite[Theorem 2.2]{BaIt:1988} to the sesquilinear form ${a_\omega^K}^*$ it follows that $\zeta$ is in the resolvent set of $(A^{K}_{\omega})^{*}$ and $(A^{K}_{\omega,n})^{*}$ for all $n\in\nline$. Since $\zeta=C_2\in\rline$ and
$\sigma(T^*)=\overline{\sigma(T)}$ for every densely defined closed operator $T$,
it follows that the same $\zeta$ also belongs to the resolvent sets of
$A^K_{\omega}$ and $A^K_{\omega,n}$ for all $n\in\nline$. Moreover,
\begin{equation} \label{eq:resolest}
  \lim_{n\to \infty}\big\| \left(\zeta I-\left(A^{K}_{\omega,n}\right)^{*}\right)^{-1}P_n x - \left(\zeta I-\left(A^{K}_{\omega}\right)^{*}\right)^{-1} x\big\|_{X}=0 \FORALL x\in X. \vspace{-1mm}
\end{equation}

\begin{framed} 
\begin{proposition}\label{pr:approxeig}
Let $(n_k)_{k\in\nline}$ be an increasing sequence of positive integers, $(\lambda_k)_{k\in\nline}$ a sequence in $\cline$, and $(v_k)_{k\in\nline}$, $(w_k)_{k\in\nline}$ two sequences in $X$ such that $\|v_k\|_X=1$ and $v_k,w_k\in V_{n_k}$ for every $k\in\nline$. Recall the constants $\omega$, $\omega_0$ and the operators $A^K_{\omega}$ and $A^K_{\omega,n}$ from \eqref{eq:lenovo} and the discussion preceding it. Assume that
\begin{align}
 &A_{\omega,{n_k}}^K v_k = \lambda_k v_k + w_k \FORALL k\in\nline, \label{eq:assump1}\\
 &\lim_{k\to\infty}\|w_k\|_X=0, \qquad \lim_{k\to\infty} \lambda_k=\lambda \label{eq:assump2}
\end{align}
for some $\lambda\in\cline$ with $\lambda\neq \omega-\omega_0$. Then there exists a non-zero $v\in D(A)$ and a subsequence $(v_{k_r})_{r\in \nline}$ of $(v_k)_{k\in \nline}$ such that \vspace{-1mm}
\begin{equation}\label{eq:weakconv}
  \lim_{r\to\infty}\langle v_{k_r},x \rangle_{X}=\langle v,x \rangle_{X} \FORALL x\in X \vspace{-1mm}
\end{equation}
and $\lambda$ is an eigenvalue of $A_{\omega}^K$ with eigenvector $v$, i.e. $A_{\omega}^K v =\lambda v$. \vspace{-1mm}
\end{proposition}
\end{framed}

\begin{proof}
The proof follows by a direct adaptation of the arguments in
\cite[Proposition~4.2]{SiAkMiNa:2025}. Suppose that
$$
v_k=\bbm{\phi_k \\ \psi_k}, \qquad w_k=\bbm{p_k \\ q_k},
$$
so that $\phi_k,\psi_k,p_k,q_k\in H_{n_k} \subset H^1(0,1)$. Using $v_k\in V_{n_k}$ and
$A^K_{\omega,n_k}=A_{n_k}+P_{n_k}BKP_{n_k}+\omega I$, we can rewrite
\eqref{eq:assump1} as
\begin{equation}\label{eq:mastereq}
 -A_{n_k}\bbm{\phi_k \\ \psi_k}
 =P_{n_k}BK\bbm{\phi_k \\ \psi_k}-(\lambda_k-\omega)\bbm{\phi_k \\ \psi_k}-\bbm{p_k \\ q_k}.
\end{equation}
Taking the inner product in $X$ of \eqref{eq:mastereq} with
$\sbm{\phi_k \\ \psi_k}\in V_{n_k}$, and using \eqref{eq:A-approx} together with
$\big\|\sbm{\phi_k \\ \psi_k}\big\|_{X}=1$, we obtain
\begin{equation} \label{eq:masterinp}
 a\!\left(\bbm{\phi_k \\ \psi_k},\bbm{\phi_k \\ \psi_k}\right)
 =\left\langle BK\bbm{\phi_k \\ \psi_k},\bbm{\phi_k \\ \psi_k}\right\rangle_{X}
 +\left(\omega-\lambda_k\right)
 -\left\langle \bbm{p_k \\ q_k},\bbm{\phi_k \\ \psi_k} \right\rangle_{X}.
\end{equation}
Taking real parts in \eqref{eq:masterinp}, using \eqref{eq:gardings-equality} on the left-hand side, and estimating the terms on the right-hand side by Cauchy--Schwarz, we obtain
\begin{align*}
\min\Big\{\frac{\eta}{4},\frac{\kappa}{2}\Big\}
\left\|\bbm{\phi_k \\ \psi_k}\right\|_{V}^{2}
&\leq
\Re\,a\!\left(\bbm{\phi_k \\ \psi_k},\bbm{\phi_k \\ \psi_k}\right)+C \\
&=
\Re\!\left(\left\langle BK\bbm{\phi_k \\ \psi_k},\bbm{\phi_k \\ \psi_k}\right\rangle_{X}
+\omega-\lambda_k
-\left\langle \bbm{p_k \\ q_k},\bbm{\phi_k \\ \psi_k} \right\rangle_{X}\right)
+C \\
&\le
\|BK\|_{\Lscr(X)}+\omega+|\lambda_k|
+\left\|\bbm{p_k \\ q_k}\right\|_{X}
+C.
\end{align*}
Here
$$
C=\frac{\nu^2}{\eta}+\frac{\gamma^2}{2\kappa}+\frac{(1-\gamma)^2}{\eta}+\frac{\eta+\kappa}{2}.
$$
Hence
\begin{equation}\label{eq:Vbound}
\left\|\bbm{\phi_k \\ \psi_k}\right\|_{V}^{2}
\leq
\frac{1}{\min\{\frac{\eta}{4},\frac{\kappa}{2}\}}
\left(
\|BK\|_{\Lscr(X)}+\omega+|\lambda_k|
+\left\|\bbm{p_k \\ q_k}\right\|_{X}
+C
\right).
\end{equation}
By \eqref{eq:assump2}, the right-hand side of \eqref{eq:Vbound} is bounded
uniformly in $k$. Hence $\big(\phi_k\big)_{k\in\nline}$ and
$\big(\psi_k\big)_{k\in\nline}$ are uniformly bounded in $H^1(0,1)$. Since
$H^1(0,1)\hookrightarrow L^2(0,1)$ compactly, there exist
$\phi,\psi\in H^1(0,1)$ and an increasing sequence $(k_r)_{r\in\nline}$ such that
\begin{equation}\label{eq:weakphipsi}
\lim_{r\to\infty}\langle \phi_{k_r}-\phi, q\rangle_{H^1(0,1)}=0,
\qquad
\lim_{r\to\infty}\langle \psi_{k_r}-\psi, q\rangle_{H^1(0,1)}=0
\end{equation}
for all $q\in H^1(0,1)$ and
\begin{equation}\label{eq:strongphipsi}
\lim_{r\to\infty}\|\phi_{k_r}-\phi\|_{L^2(0,1)}=0,
\qquad
\lim_{r\to\infty}\|\psi_{k_r}-\psi\|_{L^2(0,1)}=0 .
\end{equation}
Note that, weak convergence in $H^1(0,1)$ implies weak
convergence in $L^2(0,1)$ with the same limit, and any weak and strong limits in
$L^2(0,1)$ must coincide. Define $v=\sbm{\phi \\ \psi}$ and recall that $v_{k_r}=\sbm{\phi_{k_r}\\ \psi_{k_r}}$.
It follows from \eqref{eq:weakphipsi}-\eqref{eq:strongphipsi} that
\eqref{eq:weakconv} holds. We next show that
\begin{equation}\label{eq:vprop}
 v\in D(A), \qquad A^K_{\omega}v=\lambda v .
\end{equation}
Let $\zeta=C_2$, where $C_2$ is the constant appearing in \eqref{eq:a_Kstar2}. From \eqref{eq:assump1} we have
$$
\big(\zeta I-A^{K}_{\omega,n_{k_r}}\big)v_{k_r}
= (\zeta-\lambda_{k_r})v_{k_r}-w_{k_r}.
$$
Since $\zeta\in\rho(A^{K}_{\omega,n_{k_r}})$ (see the discussion above
\eqref{eq:resolest}), this can be written as
$$
v_{k_r}=\big(\zeta I -A^{K}_{\omega,n_{k_r}}\big)^{-1}
\big[(\zeta -\lambda_{k_r}) v_{k_r}-w_{k_r}\big].
$$
Taking the inner product in $X$ with $P_{n_{k_r}}x\in V_{n_{k_r}}$ and using
$\big((\zeta I -A^K_{\omega,n_{k_r}})^{-1}\big)^{*}=(\zeta I - (A^K_{\omega,n_{k_r}})^{*})^{-1},$
we obtain
\begin{equation}\label{eq:resolvlim}
 \langle v_{k_r},P_{n_{k_r}} x \rangle_{X}
 =
 \big \langle (\zeta -\lambda_{k_r})v_{k_r}-w_{k_r},
 \big(\zeta I - \big(A^K_{\omega,n_{k_r}}\big)^{*}\big)^{-1} P_{n_{k_r}} x
 \big \rangle_{X}
\end{equation}
for $x\in X$. From \eqref{eq:assumconv} and \eqref{eq:weakconv} we get $\lim_{r\to \infty}\langle v_{k_r}, P_{n_{k_r}} x \rangle_{X} = \langle v, x \rangle_{X}$.
Moreover, from \eqref{eq:resolest}, \eqref{eq:assump2} and \eqref{eq:weakconv} we get,
$$
 \lim_{r\to\infty}
 \left\|\big(\zeta I- \big(A^K_{\omega,n_{k_r}}\big)^{*}\big)^{-1} P_{n_{k_r}}x
 - \big(\zeta I - \big(A^K_{\omega}\big)^{*}\big)^{-1} x\right\|_X = 0 \FORALL x\in X,
$$
and
$$
 \lim_{r\to\infty}\langle(\zeta-\lambda_{k_r})v_{k_r}-w_{k_r}, x\rangle_X
 = \langle (\zeta-\lambda)v,x\rangle_X \FORALL x\in X.
$$
Using the above limits and letting $r\to\infty$ in \eqref{eq:resolvlim} yields
$$
\langle v, x\rangle_X
=
\left\langle (\zeta-\lambda)v,\ (\zeta I-(A_\omega^K)^*)^{-1}x\right\rangle_X
\qquad \forall x\in X.
$$
Using $\big((\zeta I-(A_\omega^K)^*)^{-1}\big)^*=(\zeta I-A_\omega^K)^{-1}$, we get
$$
\langle v,x\rangle_X
=
\left\langle (\zeta-\lambda)(\zeta I-A_\omega^K)^{-1}v,\ x\right\rangle_X
\qquad \forall x\in X,
$$
and hence
$$
v=(\zeta-\lambda)(\zeta I-A_\omega^K)^{-1}v.
$$
In particular, $v\in D(A)$ and by a straightforward rearrangement $A_\omega^K v=\lambda v$, proving \eqref{eq:vprop}. 

We now complete the proof by showing that $v\neq 0$, which ensures that $\lambda$ is an eigenvalue of $A_\omega^K$ with corresponding eigenvector $v$. To this end, assume for contradiction that
$$
\phi=0,\qquad \psi=0.
$$
We show below that this assumption leads to a contradiction, and hence $v\neq 0$.

Since $v_{k_r}=\sbm{\phi_{k_r}\\ \psi_{k_r}}\in V_{n_{k_r}}=H_{n_{k_r}}\times H_{n_{k_r}}$, we have $\sbm{\psi_{k_r}\\ (\eta/\gamma)\psi_{k_r}}\in V_{n_{k_r}}$. Letting $k=k_r$ in \eqref{eq:mastereq} and taking the inner product in $X$ with $\sbm{\psi_{k_r}\\ (\eta/\gamma)\psi_{k_r}}$, we obtain
\begin{align}
 a\left(\bbm{\phi_{k_r} \\ \psi_{k_r}},\bbm{\psi_{k_r} \\ \frac{\eta}{\gamma}\psi_{k_r}}\right)
 =&\left\langle BK \bbm{\phi_{k_r} \\ \psi_{k_r}},\bbm{\psi_{k_r} \\ \frac{\eta}{\gamma}\psi_{k_r}} \right\rangle_{X}
 +\left(\omega-\lambda_{k_r}\right)
 \left\langle \bbm{\phi_{k_r}\\ \psi_{k_r}}, \bbm{\psi_{k_r} \\ \frac{\eta}{\gamma}\psi_{k_r}} \right\rangle_{X}
 \nonumber\\[0.5ex]
 &
 -\left\langle \bbm{p_{k_r} \\ q_{k_r}},\bbm{\psi_{k_r} \\ \frac{\eta}{\gamma}\psi_{k_r}} \right\rangle_{X}.
 \label{eq:testpsi}
\end{align}
Expanding the left-hand side of \eqref{eq:testpsi} using the definition of $a$
and rearranging terms, yields
\begin{align}
\mu_r\|\psi_{k_r}\|_{H^1(0,1)}^2
=&
\left\langle BK\begin{bmatrix}\phi_{k_r}\\ \psi_{k_r}\end{bmatrix},
\begin{bmatrix}\psi_{k_r}\\ 0\end{bmatrix}\right\rangle_X
-\left\langle \begin{bmatrix}p_{k_r}\\ q_{k_r}\end{bmatrix},
\begin{bmatrix}\psi_{k_r}\\ \frac{\eta}{\gamma}\psi_{k_r}\end{bmatrix}\right\rangle_X +\nu\langle \phi_{k_r},\psi_{k_r,\xi}\rangle_{L^2(0,1)}
\nonumber\\
&+\eta\langle \phi_{k_r},\psi_{k_r}\rangle_{L^2(0,1)}+(\omega-\lambda_{k_r})\langle \phi_{k_r},\psi_{k_r}\rangle_{L^2(0,1)}
+\|\psi_{k_r}\|_{L^2(0,1)}^2.
\label{eq:H1_identity}
\end{align}
Here $\mu_r=1+\frac{\eta}{\gamma}\big(\kappa-\omega+\lambda_{k_r}\big)$. Since $\lim_{r\to\infty}\lambda_{k_r}=\lambda$ and $\lambda\neq \omega-\omega_0$, we have
$\lim_{r\to\infty}\mu_r=\frac{\eta}{\gamma}\big(\omega_0-\omega+\lambda\big)\neq 0.$
Assume $\phi=\psi=0$. Then \eqref{eq:strongphipsi} yields $\lim_{r\to\infty}\|\phi_{k_r}\|_{L^2(0,1)}=0$ and $\lim_{r\to\infty}\|\psi_{k_r}\|_{L^2(0,1)}=0$.

Consider \eqref{eq:H1_identity}. We claim that every term on its right-hand side converges to $0$. First, by Cauchy--Schwarz, boundedness of $BK\in\Lscr(X)$, and $\lim_{r\to\infty}\|\psi_{k_r}\|_{L^2(0,1)}=0$, we get
\begin{equation}\label{eq:conv1}
\left|\left\langle BK\,v_{k_r},\bbm{\psi_{k_r}\\0}\right\rangle_X\right|
\le \|BK\|_{\Lscr(X)}\,\|\psi_{k_r}\|_{L^2(0,1)}
\xrightarrow[r\to\infty]{}0.
\end{equation}
Next, since $(\psi_{k_r})_{r\in \nline}$ is uniformly bounded in $H^1(0,1)$, the sequence
$\sbm{\psi_{k_r}\\ \frac{\eta}{\gamma}\psi_{k_r}}$ is uniformly bounded in $X$, together with $\lim_{r\to\infty}\|w_{k_r}\|_X=0$ yields
\begin{equation}\label{eq:conv2}
\left|\left\langle w_{k_r},\bbm{\psi_{k_r}\\ \frac{\eta}{\gamma}\psi_{k_r}}\right\rangle_X\right|
\le \|w_{k_r}\|_X\,\left\|\bbm{\psi_{k_r}\\ \frac{\eta}{\gamma}\psi_{k_r}}\right\|_X
\xrightarrow[r\to\infty]{}0.
\end{equation}
Moreover, since $(\lambda_{k_r})_{r\in\nline}$ converges, the sequence $(\omega-\lambda_{k_r})_{r\in\nline}$ is bounded. Therefore, using $\|\phi_{k_r}\|_{L^2(0,1)}\to 0$ and $\|\psi_{k_r}\|_{L^2(0,1)}\to 0$ as $r\to\infty$, we get
\begin{equation}\label{eq:conv3}
\lim_{r\to\infty}(\omega-\lambda_{k_r})\langle\phi_{k_r},\psi_{k_r}\rangle_{L^2(0,1)}= 0
\quad\text{and}\quad
\lim_{r\to\infty}\eta\langle\phi_{k_r},\psi_{k_r}\rangle_{L^2(0,1)}= 0.
\end{equation}
Also, the mixed term $\nu\langle\phi_{k_r},\psi_{k_r,\xi}\rangle_{L^2(0,1)}$ converges to $0$. Indeed, $\|\phi_{k_r}\|_{L^2(0,1)}\to 0$, while $(\psi_{k_r})_{r\in\nline}$ is uniformly bounded in $H^1(0,1)$, so $(\psi_{k_r,\xi})_{r\in\nline}$ is uniformly bounded in $L^2(0,1)$. Therefore, by Cauchy--Schwarz,
\begin{equation}\label{eq:conv4}
\big|\nu\langle\phi_{k_r},\psi_{k_r,\xi}\rangle_{L^2(0,1)}\big|
\le \nu\|\phi_{k_r}\|_{L^2(0,1)}\|\psi_{k_r,\xi}\|_{L^2(0,1)}
\xrightarrow[r\to\infty]{}0.
\end{equation}
Finally, $\|\psi_{k_r}\|_{L^2(0,1)}^2\to 0$. Combining this with \eqref{eq:conv1}--\eqref{eq:conv4}, we conclude that every term on the right-hand side of \eqref{eq:H1_identity} converges to $0$, and hence
$$
\lim_{r\to\infty}\mu_r\|\psi_{k_r}\|_{H^1(0,1)}^2=0.
$$
Since $\lim_{r\to\infty}\mu_r\neq 0$, it follows that $\lim_{r\to\infty}\|\psi_{k_r}\|_{H^1(0,1)}=0$. Combining this with
$\lim_{r\to\infty}\|\phi_{k_r}\|_{L^2(0,1)}=0$ gives $\lim_{r\to\infty}\|v_{k_r}\|_X=0$, contradicting $\|v_{k_r}\|_X=1$.
Hence the assumption $\phi=\psi=0$ is false, and therefore $v\neq 0$.
\end{proof} 
In the following proposition, we prove uniform (with respect to $n$) stabilizability of the pair $(A_{\omega,n},B_n)$. More precisely, we show that for each $n\in\nline$ sufficiently large there exists a $K_n\in\Lscr(V_n,\cline)$ such that $\|e^{(A_{\omega,n}+B_n K_n)t} x\|_X \leq M e^{-\delta t} \|x\|_X$ for all $x\in V_n$, each $t\geq0$ and some $M,\delta >0$ independent of $n$.
\begin{framed}
\begin{proposition}\label{pr:unifstab}
Let $A^K_{\omega,n}=A_{\omega,n}+B_nKP_n$ be defined as in \eqref{eq:lenovo}. Then there exist $n_0\in\nline$ and constants $C,\delta>0$, independent of $n$, such that for all $n\ge n_0$,
\begin{equation}\label{eq:exponential-decay}
\|e^{A^K_{\omega,n} t}x\|_X \le Ce^{-\delta t}\|x\|_X \FORALL x\in V_n, \ \forall t>0 .
\end{equation}
\end{proposition}
\end{framed}

\begin{proof}
The proof proceeds exactly as in \cite[Proposition~4.3]{SiAkMiNa:2025}. It is therefore sufficient to verify that the relevant hypotheses of that result are met in the present setting. Once this is done, the proposition follows by the same argument. Define
$a^K_{\omega,n}\colon V_n\times V_n\to\cline$ by
$$
a^K_{\omega,n}(x_1,x_2)
=
a(x_1,x_2)-\langle B_nKP_nx_1,x_2\rangle_X-\omega\langle x_1,x_2\rangle_X\FORALL x_1,x_2\in V_n .
$$
By \eqref{eq:A-approx} and \eqref{eq:lenovo},
$$
a^K_{\omega,n}(x_1,x_2)=-\langle A^K_{\omega,n}x_1,x_2\rangle_X
\qquad \FORALL x_1,x_2\in V_n .
$$
Also, since $B_n=P_nB$ and $\|P_n\|_{\Lscr(X)}\le 1$, we have
$$
\|B_nKP_n\|_{\Lscr(X)}\le \|BK\|_{\Lscr(X)} \FORALL n\in\nline .
$$
Hence, arguing exactly as in the derivation of \eqref{eq:a_Kstar1} and \eqref{eq:a_Kstar2}, there exist constants $C_1,C_2>0$, independent of $n$, such that
$$
|a^K_{\omega,n}(x_1,x_2)|
\le C_1\|x_1\|_V\|x_2\|_V \FORALL x_1,x_2\in V_n ,
$$
and
$$
\Re\ a^K_{\omega,n}(x,x)+C_2\|x\|_X^2
\ge \min\Big\{\frac{\eta}{4},\frac{\kappa}{2}\Big\}\|x\|_V^2 \FORALL x\in V_n .
$$
Applying \cite[Chapter~IV, Theorem~6.1]{Sh:2010} exactly as in Step~1 of \cite[Proposition~4.3]{SiAkMiNa:2025}, we obtain a sector $\Sigma_0$, an angle $\theta_0\in(\pi/2,3\pi/4)$, and a constant $M_0>0$, all independent of $n$, such that
$$
\sigma(A^K_{\omega,n})\subset \Sigma_0\FORALL n\in\nline ,
$$
together with the uniform resolvent estimate outside $\Sigma_0$ and the growth bound
$$
\|e^{A^K_{\omega,n}t}\|_{\Lscr(X)}\le e^{C_2t}\FORALL t>0,\ \forall n\in\nline .
$$
Now choose $\beta\in(0,\omega_0-\omega)$ such that $\sigma(A_\omega^K)\subset \cline^-_{-\beta}$. The proof of \cite[Proposition~4.3]{SiAkMiNa:2025} uses, in addition to the sectorial form bounds established above, the eigenvalue approximation result \cite[Proposition~4.2]{SiAkMiNa:2025}. In our paper, the corresponding statement is precisely Proposition~\ref{pr:approxeig}. Hence the contradiction arguments in Steps~1 and~2 of the proof of \cite[Proposition~4.3]{SiAkMiNa:2025} apply verbatim to the operators $A_{\omega,n}^K$ considered here. It follows that there exist $n_0\in\nline$ and constants $C,\delta>0$, independent of $n$, such that
$$
\|e^{A^K_{\omega,n} t}x\|_X \le Ce^{-\delta t}\|x\|_X \FORALL x\in V_n,\ \forall t>0,\ \forall n\ge n_0 .
$$
This proves the proposition.
\end{proof}

Next we present the main result of this paper. Recall the state operator $A$ defined in \eqref{eq:Domain-A}-\eqref{eq:operator-A} with parameters $\eta,\gamma,\kappa,\nu>0$, the control operator $B\in\Lscr(\mathbb{C},X)$ in \eqref{eq:operator-B}, and its adjoint $B^*$ in \eqref{eq:adjoint-B}. From Section~\ref{sec4.1}, recall the finite-dimensional subspaces $V_n\subset V$, the orthogonal projections $P_n:X\to V_n$, and the corresponding approximations $A_n$ and $B_n$ of $A$ and $B$, respectively. Finally, recall the self-adjoint coercive operator $Q$ and the positive scalar $R$ defining the cost functional in \eqref{eq:Cost}, and set $Q_n=P_n Q P_n$. We now state the theorem.

\begin{framed} \vspace{-2mm}
\begin{theorem} \label{thm:main-result}
Let $\omega_0=\kappa+\frac{\gamma}{\eta}$. Fix $0<\omega<\omega_0$. Define $A_\omega = A + \omega I$ and $A_{\omega,n} = A_n + \omega I$. Let Assumption \ref{as:subspace} hold. Suppose that $B$ satisfies the hypothesis in Theorem \ref{thm:omega-stab}. Then there exists a unique nonnegative solution $\Pi\in\Lscr(X)$ to the operator Riccati equation \eqref{eq:Riccati} and a unique nonnegative solution $\Pi_n\in \Lscr(V_n)$ to the finite-dimensional Riccati equation \eqref{eq:Riccati_approx} for each $n>n_0$ (here $n_0>0$ is the integer in Proposition \ref{pr:unifstab}) such that \vspace{-2.5mm}
\begin{equation}\label{eq:Ricc-sol-conv}
  \lim_{\substack{n\to \infty\\ n>n_0}}\|\Pi_nP_nv-\Pi v\|_{X}=0 \FORALL v\in X. \vspace{-3mm}
\end{equation}
The feedback gain $K_\infty=-R^{-1}B^{*}\Pi$ stabilizes \eqref{eq:HM_absomega} and the feedback gain $K_n=-R^{-1}B^{*}_n\Pi_n$ stabilizes \eqref{eq:approx system} for each $n>n_0$ and \vspace{-2.5mm}
\begin{equation}\label{eq:Con-gain-conv}
  \lim_{\substack{n\to \infty\\ n>n_0}} \|K_n P_n-K_\infty\|_{\Lscr(X,\cline)}=0. \vspace{-3mm}
\end{equation}
In particular, the controller gain $K_\omega=K_n P_n$ with $n$ sufficiently large solves the $\omega$-stabilization problem, Problem \ref{prob:w-stab}.
\end{theorem} \vspace{-2mm}
\end{framed}
\vspace{-4mm}
\begin{proof}
Under the assumptions of this theorem, the operator Riccati equation \eqref{eq:Riccati} admits a unique nonnegative solution $\Pi\in\Lscr(X)$, and $A_\omega-BR^{-1}B^*\Pi$ generates an exponentially stable semigroup, see the discussion below \eqref{eq:Riccati}. By Proposition~\ref{pr:unifstab}, $A_{\omega,n}+B_n K P_n$ generates an exponentially stable semigroup for all $n>n_0$. Hence $(A_{\omega,n},B_n)$ is stabilizable for $n>n_0$. Also, $Q_n$ is coercive. Therefore, by \cite[Theorem~6.2.7]{CuZw:1995}, for each $n>n_0$ the finite-dimensional Riccati equation \eqref{eq:Riccati_approx} has a unique nonnegative solution $\Pi_n\in\Lscr(V_n)$. By \cite[Theorem~2.2]{BaKu:1984}, the limit in \eqref{eq:Ricc-sol-conv} follows once the following conditions hold for all $n>n_0$:
\begin{enumerate}
  \item[(C1)] For each initial state of \eqref{eq:approx system}, there exists an input $u\in L^2((0,\infty);\cline)$ such that the cost \eqref{eq:Cost_approx} is finite. Moreover, whenever the cost \eqref{eq:Cost_approx} is finite, the corresponding trajectory of \eqref{eq:approx system} converges to zero as $t\to\infty$.
  \item[(C2)] The semigroup $\tline_\omega$ generated by $A+\omega I$ and the adjoint semigroup $\tline_\omega^*$ satisfy
      $$
      \lim_{n\to\infty}\big\|\tline_\omega(t)v-e^{A_{\omega,n}t}P_nv\big\|_X=0
      \quad \forall\,v\in X,
      $$
      $$
      \lim_{n\to\infty}\big\|\tline_\omega^*(t)v-e^{A_{\omega,n}^*t}P_nv\big\|_X=0
      \quad \forall\,v\in X,
      $$
      and both limits are uniform in $t$ on bounded subsets of $[0,\infty)$.
\item[(C3)] The operators $B,B_n$ and their adjoints satisfy
     $$
     \lim_{n\to\infty}\|B_n\alpha-B\alpha\|_X=0,
     \qquad
     \lim_{n\to\infty}\|B_n^*P_nv-B^*v\|_{\cline}=0
     $$
     for all $\alpha\in\cline$ and $v\in X$.
  \item[(C4)] The operators $Q$ and $Q_n$ satisfy
      $$
      \lim_{n\to\infty}\|Q_nP_nv-Qv\|_X=0
      \FORALL v\in X.
      $$
  \item[(C5)] There exist constants $M_1,M_2,\delta>0$, independent of $n$, such that $\Pi_n$ and $K_n=-R^{-1}B_n^*\Pi_n$ satisfy
       $$
       \|\Pi_n v\|_X\le M_1\|v\|_X
       \quad \forall\,v\in V_n,\ \forall\,n>n_0,
       $$
       $$
       \|e^{(A_{\omega,n}+B_nK_n)t}v\|_X\le M_2e^{-\delta t}\|v\|_X
       \quad \forall\,v\in V_n,\ \forall\,n>n_0.
       $$
\end{enumerate}

We verify (C1)--(C5) using the same arguments as in \cite{BaKu:1984}, and we only sketch the main steps. Condition (C1) holds since $(A_{\omega,n},B_n)$ is stabilizable for $n>n_0$ and $Q_n$ is coercive. Recall $\zeta=C_2$ from \eqref{eq:a_Kstar2}, which exceeds the constant in the
G\aa rding inequalities \eqref{eq:gardings-equality} and \eqref{eq:gardings-adjoint}
for $a$ and $a^*$. As in the proof of \eqref{eq:resolest}, we obtain the resolvent
convergences
$$
\lim_{n\to\infty}\big\|(\zeta I-A_n)^{-1}P_nx-(\zeta I-A)^{-1}x\big\|_X=0
\quad \forall\,x\in X,
$$
$$
\lim_{n\to\infty}\big\|(\zeta I-A_n^*)^{-1}P_nx-(\zeta I-A^*)^{-1}x\big\|_X=0
\quad \forall\,x\in X.
$$
Using these limits and the G\aa rding inequalities for $a$ and $a^*$, see \eqref{eq:gardings-equality} and \eqref{eq:gardings-adjoint}, the Trotter--Kato theorem \cite[Chapter~3, Theorem~4.4]{Paz:83} gives (C2). Condition (C3) follows from the definitions of $B$ and $B_n$ and the limit in \eqref{eq:assumconv}. For (C4), write
$Q_nP_nv-Qv=P_nQP_nv-P_nQv+P_nQv-Qv$, so
$$
\|Q_nP_nv-Qv\|_X
\le \|Q\|_{\Lscr(X)}\|P_nv-v\|_X+\|P_nQv-Qv\|_X.
$$
This and \eqref{eq:assumconv} imply (C4). Finally, (C5) follows by combining the uniform stabilizability estimate \eqref{eq:exponential-decay} from Proposition~\ref{pr:unifstab} with the arguments after the proof of Lemma~3.3 in \cite{BaKu:1984}. Hence (C1)--(C5) hold, and \eqref{eq:Ricc-sol-conv} follows from \cite[Theorem~2.2]{BaKu:1984}.

Next, since $P_n\Pi_n=\Pi_n$ and $P_nB=B_n$, taking $x=B\alpha$ in \eqref{eq:Ricc-sol-conv} yields
$$
\lim_{\substack{n\to\infty\\ n>n_0}}\|P_n\Pi_nB_n\alpha-\Pi B\alpha\|_X=0
\quad \forall\,\alpha\in\cline.
$$
Therefore,
$$
\lim_{\substack{n\to\infty\\ n>n_0}}\|P_n\Pi_nB_n-\Pi B\|_{\Lscr(\cline,X)}=0.
$$
Since $P_n=P_n^*$, $\Pi_n=\Pi_n^*$, and $\Pi=\Pi^*$, we have
$$
(P_n\Pi_nB_n)^*=B_n^*\Pi_nP_n,
\qquad
(\Pi B)^*=B^*\Pi.
$$
Taking adjoints and using $\|T^*\|=\|T\|$, we obtain
$$
\lim_{\substack{n\to\infty\\ n>n_0}}\|B_n^*\Pi_nP_n-B^*\Pi\|_{\Lscr(X,\cline)}=0.
$$
The limit in \eqref{eq:Con-gain-conv} follows from this limit and the definitions of $K_\infty$ and $K_n$. Finally, $A_\omega+BK_\infty$ generates an exponentially stable semigroup and \eqref{eq:Con-gain-conv} holds. Standard perturbation arguments then imply that $A_\omega+BK_nP_n$ also generates an exponentially stable semigroup for all sufficiently large $n$. Hence, for $n$ large enough, $K_nP_n$ solves Problem~\ref{prob:w-stab}. This completes the proof.
\end{proof}


\section{Numerical example}\label{sec:numerical-example}

\ \ \ Consider the heat equation with memory \eqref{eq:PDE-ODE-1}--\eqref{eq:boundary-PDE-ODE} with $\eta=0.01$, $\kappa=0.01$, $\nu=1$, $\gamma=1$, and the input shape function $b$ given by $b(\xi)=10$ for $\xi\in(0.1,0.7)$ and $b(\xi)=0$ otherwise. For these parameters, $\omega_0=\kappa+\frac{\gamma}{\eta}=100.01$. Since $A$ has an eigenvalue at $0$, the open-loop system \eqref{eq:PDE-ODE-1}--\eqref{eq:boundary-PDE-ODE} is not exponentially stable and its response does not decay. In this example, we design a state-feedback controller so that the closed-loop system attains an exponential decay rate equal to $1$. Thus, we solve Problem~\ref{prob:w-stab} with $\omega=1<\omega_0$.

The eigenvalues of $A$ with real parts in $[-1,0]$ consist of the real eigenvalues $0$ and $-\kappa=-0.01$, together with the complex-conjugate pairs $\lambda_{k,j}$ for $k\in\{2,3\}$, computed according to the convention in Theorem~\ref{thm:point-spectrum}. For $j=1,2,3,4$, these eigenvalues are listed below.
\begin{table}[ht]
\centering
\[
\begin{array}{|c|c|c|}
\hline
j & \lambda_{2,j} & \lambda_{3,j} \\[0.3ex]
\hline
1 & -0.05 - 2.81\,i & -0.05 + 2.81\,i \\
\hline
2 & -0.20 - 5.62\,i & -0.20 + 5.62\,i \\
\hline
3 & -0.43 - 8.42\,i & -0.43 + 8.42\,i \\
\hline
4 & -0.76 - 11.21\,i & -0.76 + 11.21\,i \\
\hline
\end{array}
\]
\vspace{-4mm}
\end{table}

Furthermore, for all $j\ge 5$, we have $\Re(\lambda_{2,j})<-1$ and $\Re(\lambda_{3,j})<-1$. Moreover, $\Re(\lambda_{1,j})<-100$ for every $j\in\mathbb{N}$. We next verify that the pair $(A,B)$ is $\omega$-stabilizable for $\omega=1$. For the chosen input shape function $b$, we have $\int_0^1 \overline{b(\xi)}\,d\xi=6$. For each $k\in\{1,2,3\}$ and $j\in\mathbb{N}$, define
$$
I_{k,j}=\int_0^1 \overline{b(\xi)}\, e^{\frac{\alpha_{k,j}}{2}\xi}\left(\cos(\pi j \xi)-\frac{\alpha_{k,j}}{2\pi j}\sin(\pi j \xi)\right)\,d\xi,
$$
where
$$
\alpha_{k,j}=-\frac{\nu(\lambda_{k,j}+\kappa)}{\eta(\lambda_{k,j}+\kappa)+\gamma}.
$$
For the eigenvalues with real parts in $[-1,0]$, the corresponding values of $I_{2,j}$ and $I_{3,j}$ are listed below.
\begin{table}[htp!]
\centering
\[
\begin{array}{|c|c|c|}
\hline
j & I_{2,j} & I_{3,j} \\[0.3ex]
\hline
1 & 2.90 - 1.47\,i & 2.90 + 1.47\,i \\
\hline
2 & 0.47 - 3.54\,i & 0.47 + 3.54\,i \\
\hline
3 & -1.41 - 2.10\,i & -1.41 + 2.10\,i \\
\hline
4 & -1.61 - 0.88\,i & -1.61 + 0.88\,i \\
\hline
\end{array}
\]
\vspace{-4mm}
\end{table}

Since the above quantities are nonzero, it follows from Theorem~\ref{thm:omega-stab} that the pair $(A,B)$ is $\omega$-stabilizable for $\omega=1$.

Let $H_n=\operatorname{span}\{\psi_0,\psi_1,\dots,\psi_n\}$, where $\psi_0=1$ and $\psi_j(\xi)=\sqrt{2}\cos(j\pi\xi)$ for $j\ge1$. If $v\in H^1(0,1)$ satisfies $\langle v,\psi_j\rangle_{H^1(0,1)}=0$ for all $j\ge0$, then for $j\ge1$, integration by parts gives
$$
\langle v,\psi_j\rangle_{H^1(0,1)}=(1+j^2\pi^2)\langle v,\psi_j\rangle_{L^2(0,1)}.
$$
For $j=0$, we have
$$
\langle v,\psi_0\rangle_{H^1(0,1)}=\langle v,\psi_0\rangle_{L^2(0,1)}.
$$
Hence $\langle v,\psi_j\rangle_{L^2(0,1)}=0$ for all $j\ge0$. Since $(\psi_j)_{j\ge0}$ is complete in $L^2(0,1)$, it follows that $v=0$. Thus $\operatorname{span}\{\psi_j:\,j\ge0\}$ is dense in $H^1(0,1)$. Therefore, for every $v\in H^1(0,1)$, there exists $v_n\in H_n$ such that $\|v-v_n\|_{H^1(0,1)}\to0$ as $n\to\infty$. Hence, with $V_n=H_n\times H_n$, Assumption~\ref{as:subspace} holds, and all the hypotheses of Theorem~\ref{thm:main-result} are satisfied.

Although the theoretical analysis in the previous sections is carried out on the complexified state space, the present numerical example is implemented on the corresponding real state space. This is justified because the parameters $\eta,\kappa,\nu,\gamma$, the actuator profile $b$, the basis functions $(\psi_j)_{j\ge0}$, the operators $Q$ and $R$, and the chosen initial data are all real-valued. Consequently, the finite-dimensional matrices arising in this example are real, the corresponding matrix Riccati equations are solved over $\mathbb{R}$, and the computed feedback gains are real-valued.

We take $Q=I$ and $R=1$. By Theorem~\ref{thm:main-result}, there exists a unique nonnegative solution $\Pi\in\Lscr(X)$ of \eqref{eq:Riccati} such that the gain $K_\infty=-R^{-1}B^*\Pi\in\Lscr(X,\rline)$ stabilizes \eqref{eq:HM_absomega}. Moreover, for each sufficiently large $n$, there exists a unique nonnegative solution $\Pi_n\in\Lscr(V_n)$ of \eqref{eq:Riccati_approx} such that $K_n=-R^{-1}B_n^*\Pi_n\in\Lscr(V_n,\rline)$ stabilizes \eqref{eq:approx system}. Since $K_\infty$ is a bounded real-valued linear functional on the real Hilbert space $X=L^2(0,1)\times H^1(0,1)$, the Riesz representation theorem yields unique real-valued functions $\alpha\in L^2(0,1)$ and $\beta\in H^1(0,1)$ such that
$$
K_\infty \bbm{p \\ q}= \langle\alpha, p \rangle_{L^2(0,1)} + \langle \beta, q \rangle_{H^1(0,1)} \FORALL \bbm{p \\ q} \in X.
$$
Likewise, for each sufficiently large $n$, there exist unique real-valued functions $\alpha_n\in H_n$ and $\beta_n\in H_n$ such that
$$
K_n P_n \bbm{p \\ q}= \langle\alpha_n, p \rangle_{L^2(0,1)} + \langle \beta_n, q \rangle_{H^1(0,1)} \FORALL \bbm{p \\ q} \in X.
$$
For several values of $n$, we solved \eqref{eq:Riccati_approx} numerically through its matrix representation to compute $\Pi_n$. We then formed $K_n$ and extracted $\alpha_n$ and $\beta_n$. Table~\ref{tab:convex1} shows that both $\|\alpha_{n+1}-\alpha_n\|_{L^2(0,1)}$ and $\|\beta_{n+1}-\beta_n\|_{H^1(0,1)}$ decrease as $n$ increases. This indicates that $\alpha_n$ converges in $L^2(0,1)$ and $\beta_n$ converges in $H^1(0,1)$ as $n\to\infty$. The same behavior is also seen in the plots of $\alpha_n$, $\beta_n$, and $\frac{d}{d\xi}\beta_n$ in Figures~1--2, which illustrate the gain convergence stated in \eqref{eq:Con-gain-conv}.

The plots further suggest that $\alpha_{25}$ and $\beta_{25}$ are already close to their limiting profiles, and we therefore choose $K_\omega=K_{25}P_{25}$. We then use the $n=500$ approximate model \eqref{eq:approx system} as a sufficiently accurate approximation of \eqref{eq:HM_absomega} and implement the feedback $K_{25}P_{25}$. Figure~3 compares the spectra of $A_{500}+I$ and $A_{500}+I+B_{500}K_{25}P_{25}$. As expected, the closed-loop eigenvalues lie in the open left half-plane. This shows that $K_\omega=K_{25}P_{25}$ achieves $\omega$-stabilization in this example and corroborates Theorem~\ref{thm:main-result}.
\vspace{-2mm}

\begin{table}[ht]
\centering
\[
\begin{array}{|c|c|c|}
\hline
n & \|\alpha_{n+1}-\alpha_n\|_{L^2(0,1)} & \|\beta_{n+1}-\beta_n\|_{H^1(0,1)} \\[0.5ex]
\hline
2  & 3.43 & 1.71 \\ \hline
5  & 1.46 & 1.37 \\ \hline
10 & 7.37\times 10^{-2} & 9.84\times 10^{-2} \\ \hline
20 & 6.61\times 10^{-3} & 5.70\times 10^{-3} \\ \hline
25 & 4.82\times 10^{-3} & 2.53\times 10^{-3} \\ \hline
\end{array}
\]
\vspace{-5mm}
\parbox{5.2in}{\caption{The quantities $\|\alpha_{n+1}-\alpha_n\|_{L^2(0,1)}$ and $\|\beta_{n+1}-\beta_n\|_{H^1(0,1)}$ decrease as $n$ increases, indicating that $\alpha_n$ and $\beta_n$ converge in $L^2(0,1)$ and $H^1(0,1)$, respectively, as $n\to\infty$.\vspace{-2mm}}}
\label{tab:convex1}
\end{table}

$$ \vspace{-4mm}
{\includegraphics[width=0.5\textwidth]{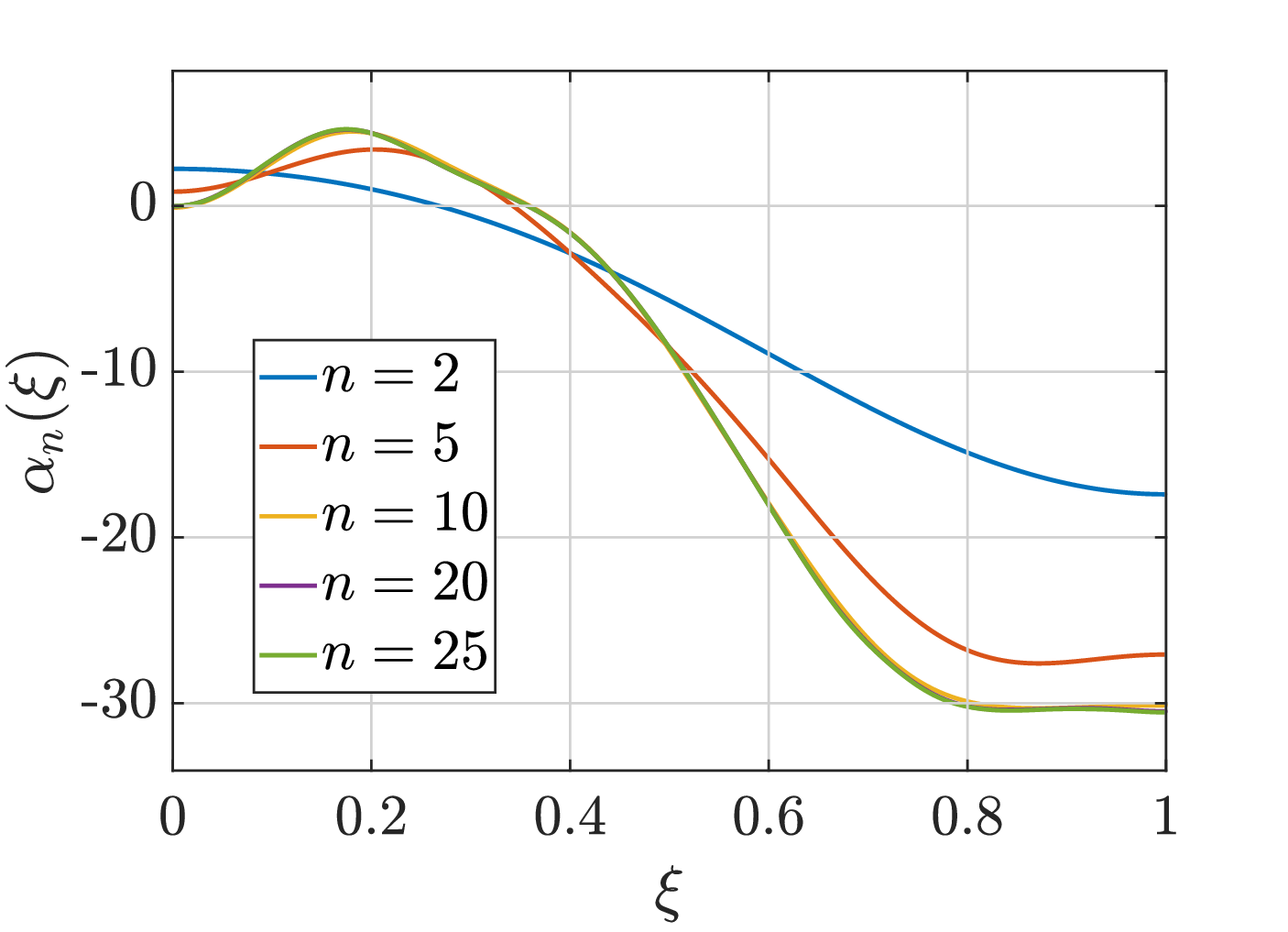}}
$$
\begin{center}
{\parbox{5.2in}{\small{Figure 1. Plots of $\alpha_n$ for different values of $n$ indicate that $\alpha_n$ converges to a limit in $L^2(0,1)$ as $n\to\infty$.}}} \vspace{-4mm}
\end{center}

$$ \vspace{-4mm}
{\includegraphics[width=0.5\textwidth]{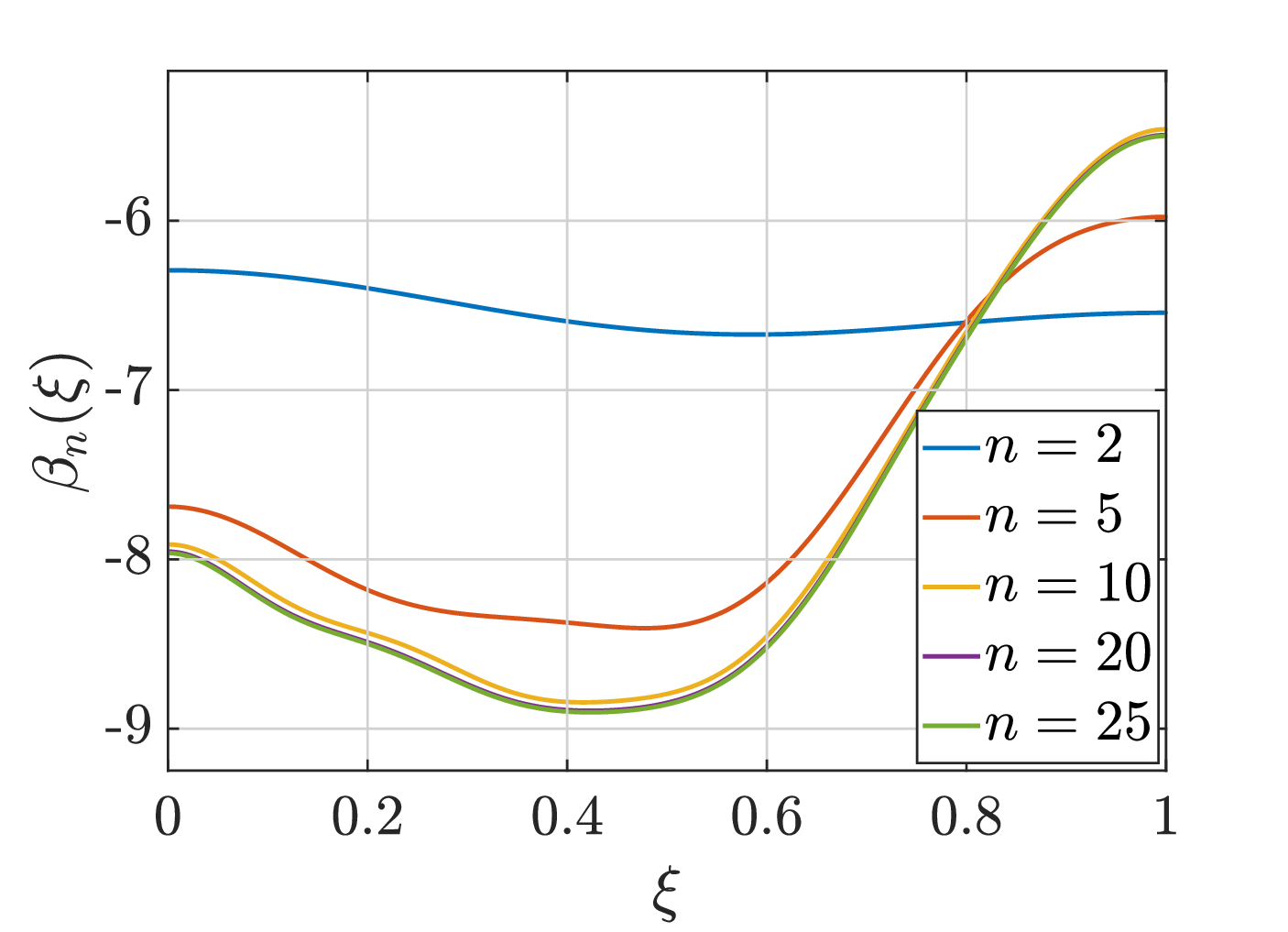}}
{\includegraphics[width=0.5\textwidth]{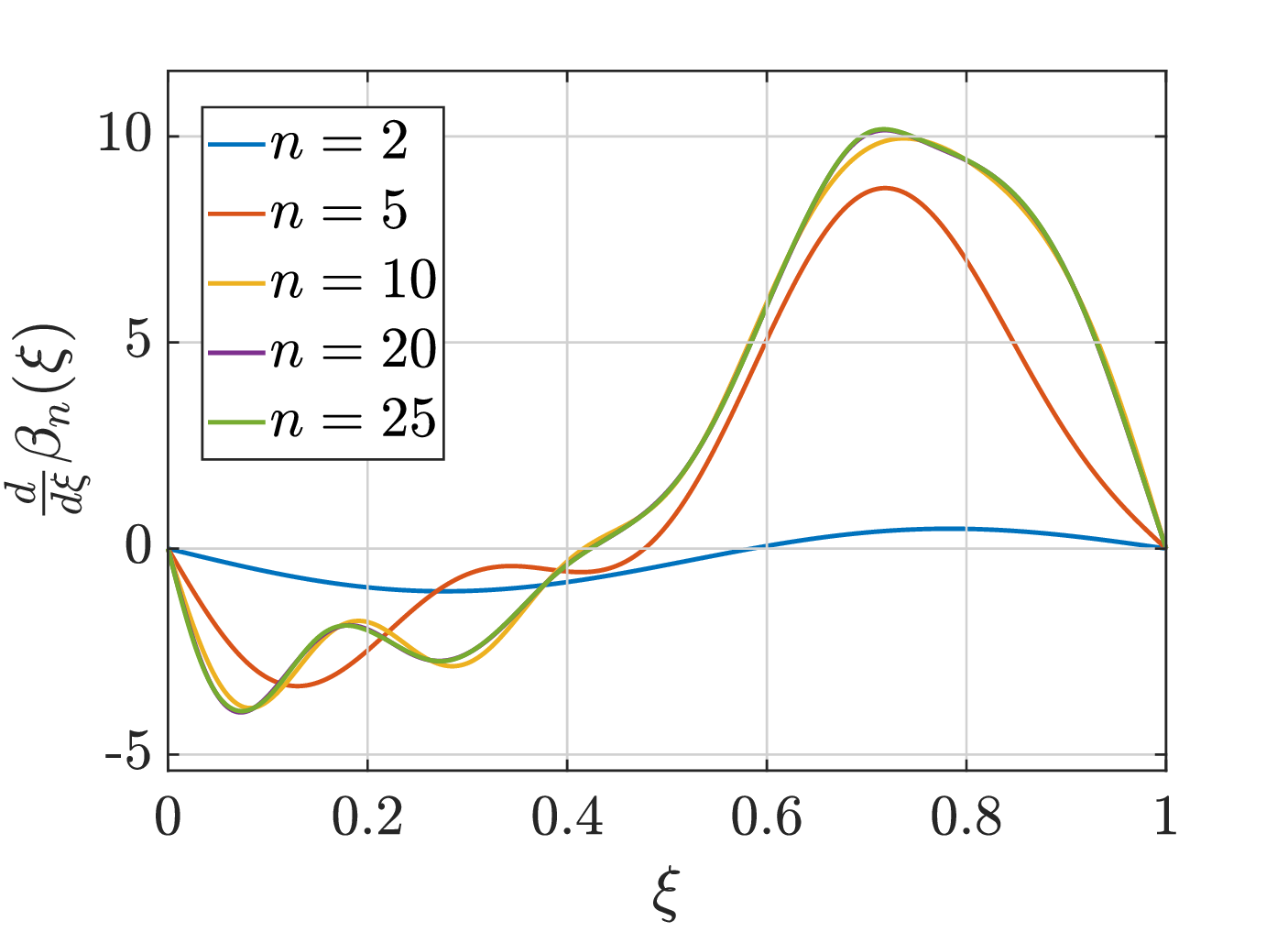}}
$$
\begin{center}
{\parbox{5.2in}{\small{Figure 2. Plots of $\beta_n$ and $\frac{d}{d\xi}\beta_n$ for different values of $n$ indicate that $\beta_n$ converges to a limit in $H^1(0,1)$ as $n\to\infty$.}}} \vspace{-4mm}
\end{center}

$$ \vspace{-2mm}
{\includegraphics[width=0.48\textwidth]{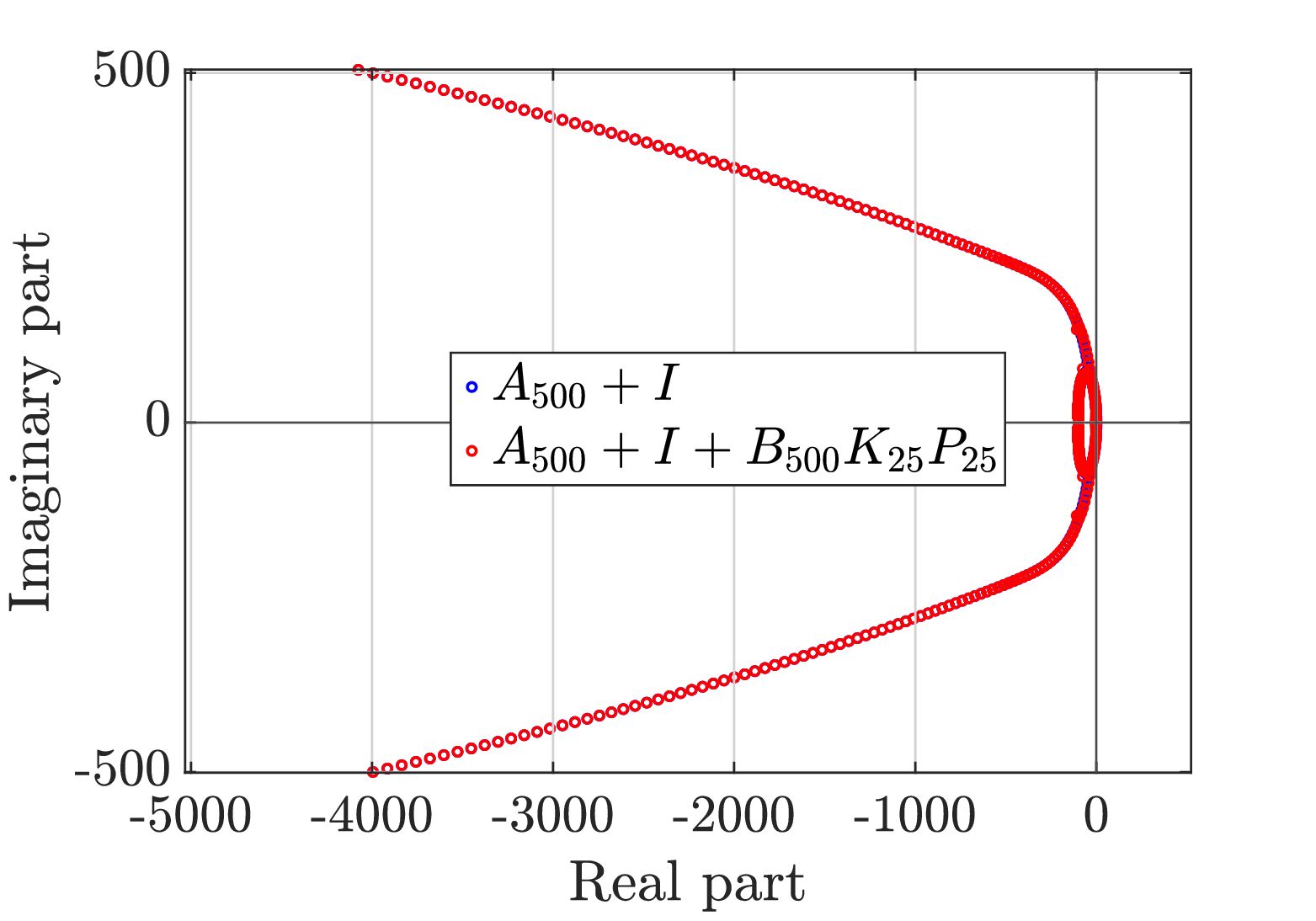}}
{\includegraphics[width=0.48\textwidth]{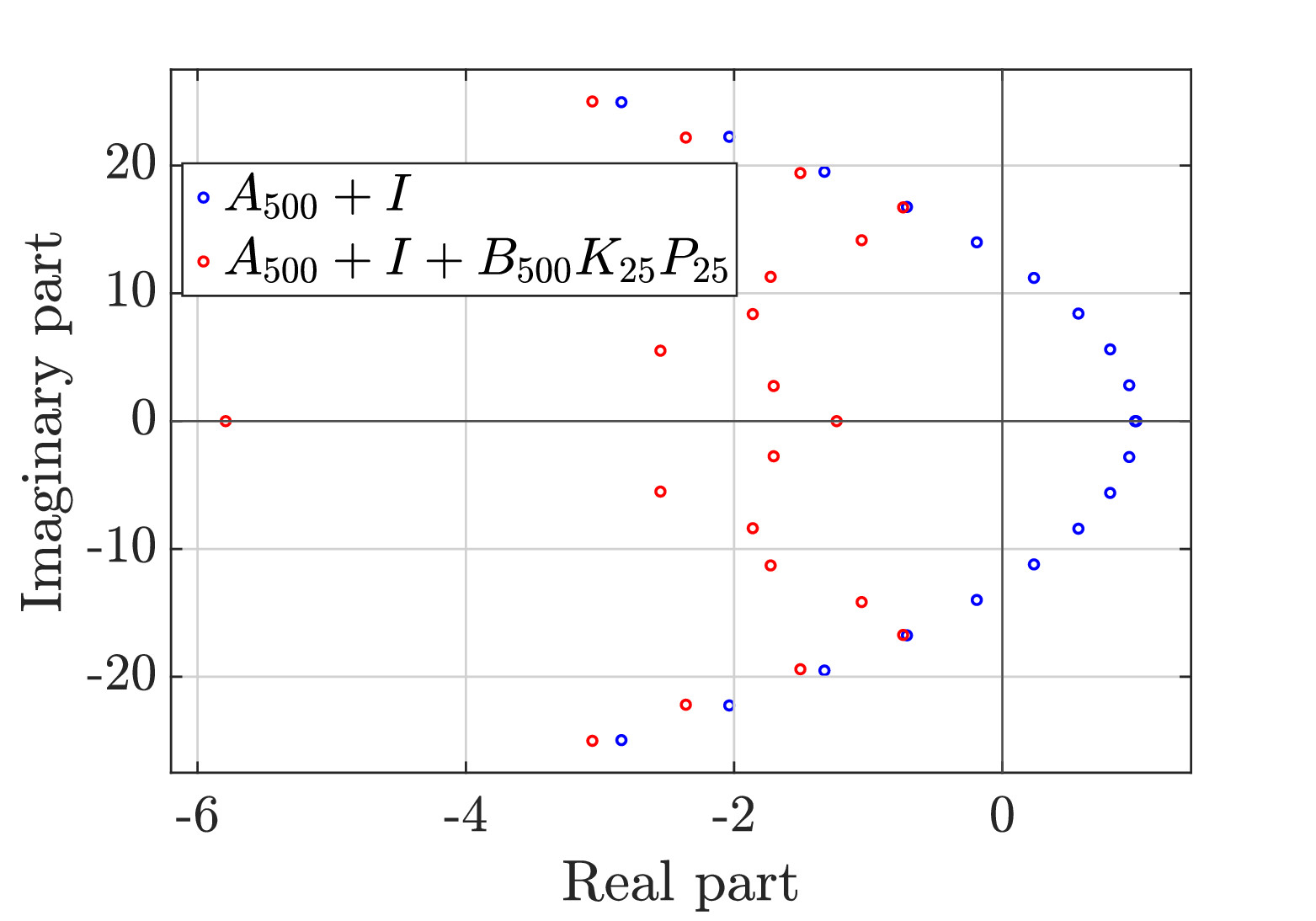}}
$$
\begin{center}
\parbox{5.2in}{{\small{Figure 3. The plots show the eigenvalues of $A_{500}+I$ and $A_{500}+I+B_{500}K_{25}P_{25}$, together with a zoomed-in view near the origin. The shifted open-loop model has eigenvalues in the open right half-plane, whereas all eigenvalues of $A_{500}+I+B_{500}K_{25}P_{25}$ have negative real parts. Hence, $K_\omega=K_{25}P_{25}$ solves the $\omega$-stabilization problem for $\omega=1$.
\vspace{-4mm}}}}
\end{center}

$$
{\includegraphics[width=0.48\textwidth]{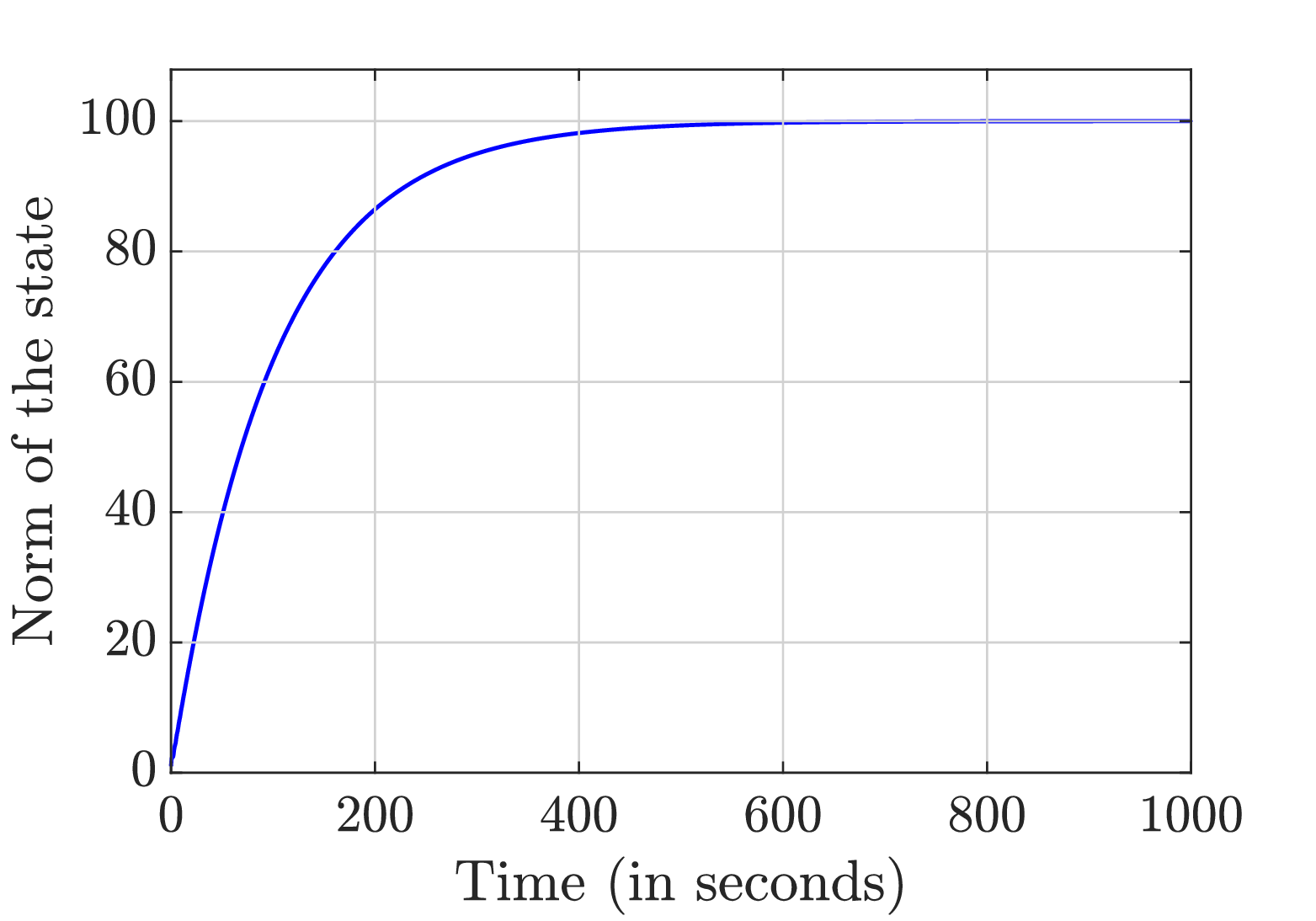}}
{\includegraphics[width=0.48\textwidth]{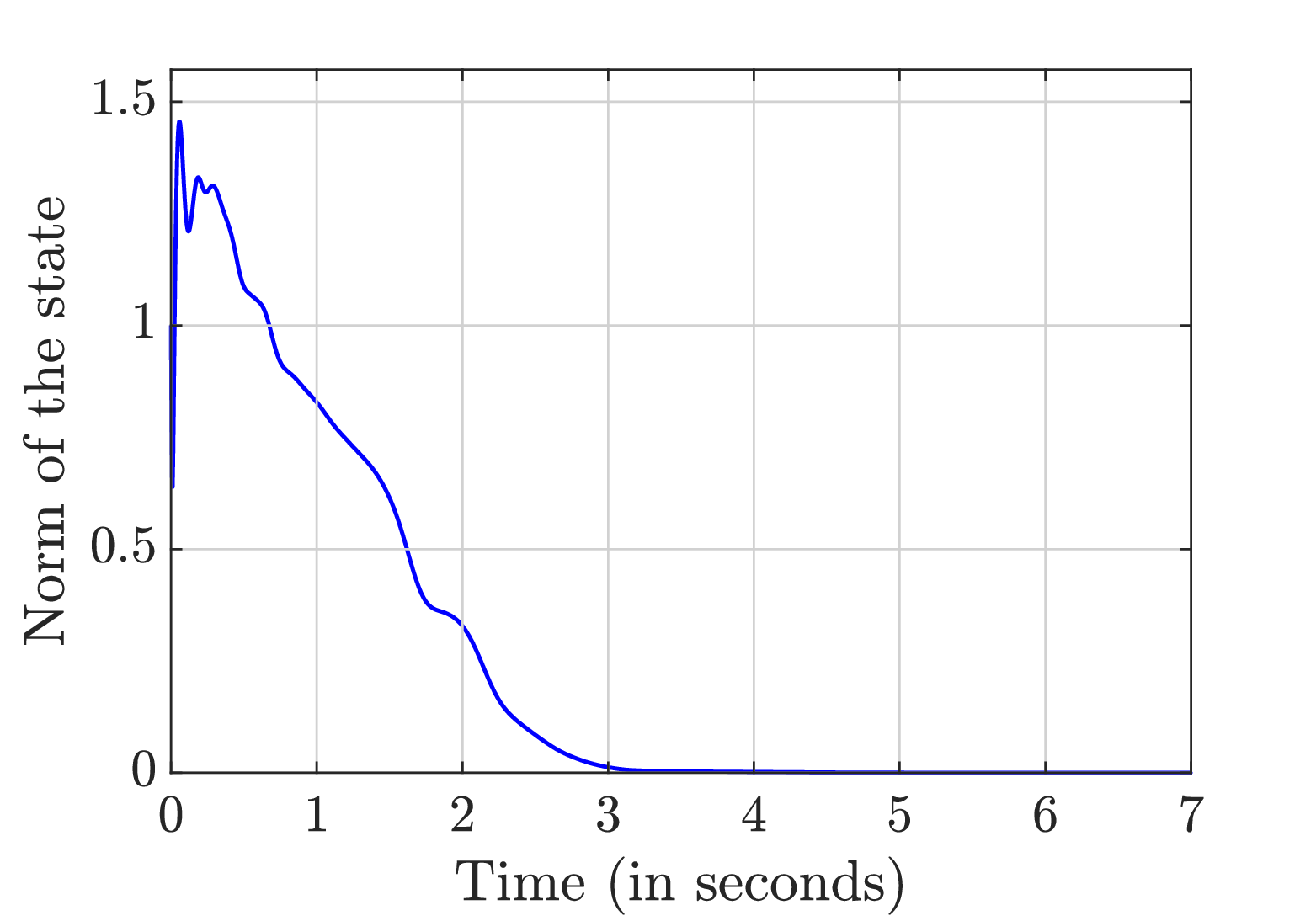}}
$$
\begin{center}
\parbox{5.2in}{{\small{Figure 4. The plot on the left shows the open-loop response of \eqref{eq:PDE-ODE-1}--\eqref{eq:boundary-PDE-ODE} for the chosen initial state, while the plot on the right shows the closed-loop response for the same initial state under the controller gain $K_{25}P_{25}$. The open-loop response saturates to a nonzero steady state because $A$ has an eigenvalue at $0$, whereas the closed-loop feedback stabilizes the system with the desired decay rate and drives the state to $0$.}}}
\end{center}

Figure~4 shows the response of \eqref{eq:PDE-ODE-1}--\eqref{eq:boundary-PDE-ODE} from the initial state
$\sbm{y^0\\ z^0}=\sbm{1\\0}$ under open-loop control $u=0$ and under the state feedback
$u=K_{25}P_{25}\sbm{y\\ z}$. In open loop, the constant mode is not damped because $A$ has an eigenvalue at $0$. Hence, $y(\xi,t)$ remains close to $1$ for all $\xi\in(0,1)$ and $t>0$. Solving $\dot z(t)=-\kappa z(t)+\gamma y(t)$ then gives the steady state $z_\infty=(\gamma/\kappa)y=100$, and hence the $X$-norm of
$\sbm{y(t)\\ z(t)}$ converges to the norm of $\sbm{1\\100}$, which is approximately $100$. In closed loop, the feedback moves this mode into the left half-plane, so both $y$ and $z$ decay exponentially to $0$.

\section{Conclusion}\label{sec:conclu}

\ \ \ In this work, for a prescribed $\omega>0$, we have addressed the $\omega$-stabilization problem for a one-dimensional heat-transfer model with advection and exponentially fading memory under thermally insulated boundary conditions. The dynamics are formulated as an abstract control system on a Hilbert space with state operator $A$ and control operator $B$. We first carried out a detailed spectral analysis of $A$ by explicitly characterizing its point spectrum, which consists of two special real eigenvalues together with three mode-dependent eigenvalue branches given by the roots of a cubic polynomial. This characterization reveals a finite negative accumulation point $-\omega_0$ and shows that one branch diverges to $-\infty$ while the remaining branches accumulate at $-\omega_0$, yielding an intrinsic upper bound on decay rates achievable under bounded state feedback and restricting stabilization to $0<\omega<\omega_0$. Since $A$ does not have compact resolvent, identifying the full spectrum as $\sigma(A)=\overline{\sigma_p(A)}$ and establishing finite algebraic multiplicity of eigenvalues require additional arguments.

Under a verifiable non-orthogonality condition on the input profile defining $B$, we then established $\omega$-stabilizability of the shifted pair $(A+\omega I,B)$ for every $\omega\in(0,\omega_0)$. Building on this, we presented a numerical scheme for computing stabilizing feedback gains via an LQR design. The scheme constructs finite-dimensional approximations $(A_n)_{n\in\mathbb{N}}$ and $(B_n)_{n\in\mathbb{N}}$ of $A$ and $B$, solves the associated finite-dimensional algebraic Riccati equations, and produces gains $K_n$ such that the approximate closed-loop generators $A_n+\omega I+B_nK_n$ are exponentially stable. The validity of this approximation-based design is justified by proving the crucial uniform stabilizability of the pairs $(A_n+\omega I,B_n)$, which allows one to rigorously bridge finite-dimensional Riccati computations to $\omega$-stabilization of the full coupled PDE--ODE system for all sufficiently large $n$.

\appendix
\renewcommand{\thesection}{\Alph{section}}
\numberwithin{equation}{section}
\section{Appendix}\label{sec:appendix}
\ \ \ The following two auxiliary lemmas are used in the proof of
Theorem~\ref{thm:finite-multi}. Throughout these lemmas, the constants
$\eta,\gamma,\kappa$ and $\nu$ are the fixed positive parameters introduced below
\eqref{eq:total-flux}, $\omega_0=\kappa+\frac{\gamma}{\eta}$, and we recall that
$X=L^2(0,1)\times H^1(0,1)$ and $V=H^1(0,1)\times H^1(0,1)$. The operator
$A_{re}\colon D(A_{re})\subset V\to V$ is defined in
\eqref{eq:operator-AV} and its domain given by \eqref{eq:Domain-AV}.
\begin{framed}
\begin{lemma}\label{lem:SL-solvability}
There exists $\mu_1>-\omega_0$ such that, for every $\mu_0\in[\mu_1,\infty)$
and every $h\in L^2(0,1)$, the boundary value problem
\begin{equation}\label{eq:u-ode}
-\eta u_{\xi\xi}
+
\left(
\mu_0-\frac{\gamma}{\eta}+\frac{\nu^2}{4\eta}
\right)u
=
h
\qquad \text{in } L^2(0,1),
\end{equation}
with the separated boundary conditions
\begin{equation}\label{eq:u-bc}
u_\xi(0)-\frac{\nu}{2\eta}u(0)=0,
\qquad
u_\xi(1)-\frac{\nu}{2\eta}u(1)=0,
\end{equation}
has a unique solution $u\in H^2(0,1)$. Moreover, for each fixed
$\mu_0\in[\mu_1,\infty)$, there exists a constant $C_{\mu_0}>0$, independent of
$h$, such that
\[
\|u\|_{H^2(0,1)}
\le C_{\mu_0}\|h\|_{L^2(0,1)}.
\]
\end{lemma}
\end{framed}

\begin{proof}
Consider the eigenvalue problem
\[
-\eta u_{\xi\xi}
+
\left(
-\frac{\gamma}{\eta}+\frac{\nu^2}{4\eta}
\right)u
=
\lambda u
\qquad \text{in } L^2(0,1),
\]
with the boundary conditions \eqref{eq:u-bc}. By
\cite[Theorem~4.3.1]{Zettl:2005}, its eigenvalues are bounded below. Hence we
can choose $\mu_1>-\omega_0$ such that, for every $\mu_0\in[\mu_1,\infty)$,
the number $-\mu_0$ is not an eigenvalue. Equivalently, the homogeneous problem
corresponding to \eqref{eq:u-ode}--\eqref{eq:u-bc} has only the trivial solution.

Fix $\mu_0\in[\mu_1,\infty)$ and $h\in L^2(0,1)$. Since $0$ is not an eigenvalue
of \eqref{eq:u-ode}--\eqref{eq:u-bc} with $h=0$, the Green's function
$G_{\mu_0}$ exists by \cite[Theorem~4.11.1]{Zettl:2005}. The solution $u$ is
given by
\[
u(\xi)
=
\int_0^1 G_{\mu_0}(\xi,s)h(s)\,ds.
\]
The function $G_{\mu_0}$ is continuous on $[0,1]\times[0,1]$, see
\cite[Remark~3.8.1]{Zettl:2005}. Hence
\[
\|u\|_{L^2(0,1)}
\le C_{\mu_0}\|h\|_{L^2(0,1)}.
\]
From \eqref{eq:u-ode}, we have
$\eta u_{\xi\xi}
=
\left(\mu_0-\frac{\gamma}{\eta}+\frac{\nu^2}{4\eta}\right)u-h$.
Therefore $u_{\xi\xi}\in L^2(0,1)$ and
$\|u_{\xi\xi}\|_{L^2(0,1)}
\le C_{\mu_0}\|h\|_{L^2(0,1)}$. Consequently $u\in H^2(0,1)$, and
the standard estimate
$\|u\|_{H^2(0,1)}
\le C\left(\|u\|_{L^2(0,1)}
+\|u_{\xi\xi}\|_{L^2(0,1)}\right)$ gives
\[
\|u\|_{H^2(0,1)}
\le C_{\mu_0}\|h\|_{L^2(0,1)}.
\]
This completes the proof of the lemma.
\end{proof}

\begin{framed}
\begin{lemma}\label{lem:density}
The domain $D(A_{re})$ defined in \eqref{eq:Domain-AV} is dense in $V$.
\end{lemma}
\end{framed}
\begin{proof}
Fix $\sbm{\widetilde w\\ \widetilde z}\in V$. Put
$a_j=\frac{\nu}{\eta}\bigl(\widetilde w(j)-\widetilde z(j)\bigr)$, $j=0,1$, and
define
\[
q(\xi)=(1-\xi)\,\widetilde z(0)+\xi\,\widetilde z(1),
\]
\[
p(\xi)=\widetilde w(0)\,(1-3\xi^2+2\xi^3)+\widetilde w(1)\,(3\xi^2-2\xi^3)
+a_0\,(\xi-2\xi^2+\xi^3)+a_1\,(-\xi^2+\xi^3).
\]
Then $p,q\in C^\infty([0,1])$, and a direct computation gives
$q(j)=\widetilde z(j)$, $p(j)=\widetilde w(j)$, and $p_\xi(j)=a_j$ for
$j=0,1$. Hence $p_\xi(j)=\frac{\nu}{\eta}\bigl(p(j)-q(j)\bigr)$ for $j=0,1$.
The differences $\widetilde w-p$ and $\widetilde z-q$ lie in $H^1(0,1)$ and
vanish at $j=0,1$, so they belong to $H_0^1(0,1)$, in which $C_c^\infty(0,1)$ is
dense \cite[Theorem~8.12 and Chapter~8, Remark~14]{Brezis:2011}. Hence there exist
sequences $(\phi_n)_{n\in\nline}$ and $(\psi_n)_{n\in\nline}$ in $C_c^\infty(0,1)$
such that
\[
\lim_{n\to\infty}\|\phi_n-(\widetilde w-p)\|_{H^1(0,1)}=0,
\qquad
\lim_{n\to\infty}\|\psi_n-(\widetilde z-q)\|_{H^1(0,1)}=0.
\]
Set $w_n=p+\phi_n$ and $z_n=q+\psi_n$.
Since $\phi_n,\psi_n$ vanish together with all derivatives near $0$ and $1$, we
have $w_n(j)=p(j)$, $[w_n]_{\xi}(j)=p_\xi(j)$, and $z_n(j)=q(j)$, hence
$[w_n]_{\xi}(j)=\frac{\nu}{\eta}\bigl(w_n(j)-z_n(j)\bigr)$ for $j=0,1$. Moreover
$w_n,z_n\in C^\infty([0,1])$, so $w_n\in H^2(0,1)$ and
$\eta [w_n]_{\xi\xi}+\nu [z_n]_{\xi}\in C^\infty([0,1])\subset H^1(0,1)$. By
\eqref{eq:Domain-AV}, $\sbm{w_n\\ z_n}\in D(A_{re})$ for every $n\in\nline$.
Finally, $w_n-\widetilde w=\phi_n-(\widetilde w-p)$ and
$z_n-\widetilde z=\psi_n-(\widetilde z-q)$, so
\[
\lim_{n\to\infty}\left\|\bbm{w_n\\ z_n}-\bbm{\widetilde w\\ \widetilde z}\right\|_V=0.
\]
Therefore $D(A_{re})$ is dense in $V$.
\end{proof}


\vspace{-2mm}


\begin{thebibliography}{99} \vspace{-1mm}{\small

\bibitem{AkMi:2022}
W. Akram and D. Mitra,
``Local stabilization of viscous Burgers equation with memory,''
{\it Evol. Equ. Control Theory}, vol. 11, pp. 939--973, 2022.
\vspace{-2mm}

\bibitem{AkMiNaRa:2024}
W. Akram, D. Mitra, N. Nataraj and M. Ramaswamy,
``Feedback stabilization of a parabolic coupled system and its numerical study,''
{\it Math. Control Relat. Fields}, vol. 14, pp. 695--746, 2024.
\vspace{-2mm}

\bibitem{BaRa:2024}
M. Badra and J.-P. Raymond,
``Approximation of feedback gains for abstract parabolic systems,''
{\it ESAIM: Control, Optim. Calc. Var.}, 2024, doi: 10.1051/cocv/2024084.
\vspace{-2mm}

\bibitem{BaRa:2024b}
M. Badra and J.-P. Raymond,
``Approximation of feedback gains using spectral projections: Application to the Oseen system,''
{\it SIAM J. Control Optim.}, vol. 62, pp. 2910--2935, 2024.
\vspace{-2mm}

\bibitem{BaKr:2002}
A. Balogh and M. Krstic,
``Infinite-dimensional backstepping-style feedback transformations for a heat equation with an arbitrary level of instability,''
{\it Eur. J. Control}, vol. 8, pp. 165--175, 2002.
\vspace{-2mm}

\bibitem{BaIt:1988}
H. T. Banks and K. Ito,
``A unified framework for approximation in inverse problems for distributed parameter systems,''
{\it Control Theory Adv. Technol.}, vol. 4, pp. 73--90, 1988.
\vspace{-2mm}

\bibitem{BaKu:1984}
H. T. Banks and K. Kunisch,
``The linear regulator problem for parabolic systems,''
{\it SIAM J. Control Optim.}, vol. 22, pp. 684--698, 1984.
\vspace{-2mm}

\bibitem{BaIa:2000}
V. Barbu and M. Iannelli,
``Controllability of the heat equation with memory,''
{\it Differ. Integral Equ.}, vol. 13, pp. 1393--1412, 2000.
\vspace{-2mm}

\bibitem{Baxley:1972}
J. V. Baxley,
``On the Weyl spectrum of a Hilbert space operator,''
{\it Proc. Amer. Math. Soc.}, vol. 34, pp. 447--452, 1972.
\vspace{-2mm}

\bibitem{BePrDeMi:2007}
A. Bensoussan, G. Da Prato, M. C. Delfour and S. K. Mitter,
{\em Representation and Control of Infinite Dimensional Systems},
Birkh\"{a}user, Boston, 2007.
\vspace{-2mm}

\bibitem{BoBaKr:2003}
D. M. Bo\v{s}kovi\'c, A. Balogh and M. Krsti\'c,
``Backstepping in infinite dimension for a class of parabolic distributed parameter systems,''
{\it Math. Control Signals Syst.}, vol. 16, pp. 44--75, 2003.
\vspace{-2mm}

\bibitem{BrKu:14}
T. Breiten and K. Kunisch,
``Riccati-based feedback control of the monodomain equations with the FitzHugh--Nagumo model,''
{\it SIAM J. Control Optim.}, vol. 52, pp. 4057--4081, 2014.
\vspace{-2mm}

\bibitem{BrKu:17}
T. Breiten and K. Kunisch,
``Compensator design for the monodomain equations with the FitzHugh--Nagumo model,''
{\it ESAIM: Control, Optim. Calc. Var.}, vol. 23, pp. 241--262, 2017.
\vspace{-2mm}

\bibitem{Brezis:2011}
H. Brezis,
{\em Functional Analysis, Sobolev Spaces and Partial Differential Equations},
Springer, New York, 2011.
\vspace{-2mm}

\bibitem{BuCh:2022}
J. Burns and J. Cheung,
``Optimal convergence rates for Galerkin approximation of operator Riccati equations,''
{\it Numer. Methods Partial Differ. Equ.}, vol. 38, pp. 2045--2083, 2022.
\vspace{-2mm}

\bibitem{ChNa:2025}
S. Chatterjee and V. Natarajan,
``Motion planning for parabolic equations using flatness and finite-difference approximations,''
{\it IEEE Trans. Autom. Control}, vol. 70, pp. 4439--4454, 2025.
\vspace{-2mm}

\bibitem{ChSuNa:2025}
S. Chattopadhyay, S. Sukumar and V. Natarajan,
``Adaptive identification of linear infinite-dimensional systems,''
{\it Int. J. Control}, vol. 98, pp. 593--608, 2025.
\vspace{-2mm}

\bibitem{SiZhZu:17}
F. W. Chaves-Silva, X. Zhang and E. Zuazua,
``Controllability of evolution equations with memory,''
{\it SIAM J. Control Optim.}, vol. 55, pp. 2437--2459, 2017.
\vspace{-2mm}

\bibitem{Conway:1978}
J. B. Conway,
{\em Functions of One Complex Variable},
2nd ed.,
Springer-Verlag, New York, 1978.
\vspace{-2mm}

\bibitem{CuZw:1995}
R. F. Curtain and H. J. Zwart,
{\em An Introduction to Infinite-Dimensional Linear Systems Theory},
Springer-Verlag, New York, 1995.
\vspace{-2mm}

\bibitem{GiRoTa:1992}
J. S. Gibson, I. G. Rosen and G. Tao,
``Approximation in control of thermoelastic systems,''
{\it SIAM J. Control Optim.}, vol. 30, pp. 1163--1189, 1992.
\vspace{-2mm}

\bibitem{GrMo:1996}
J. R. Grad and K. A. Morris,
``Solving the linear quadratic optimal control problem for infinite-dimensional systems,''
{\it Comput. Math. Appl.}, vol. 32, pp. 99--119, 1996.
\vspace{-2mm}

\bibitem{GuIm:13}
S. Guerrero and O. Yu. Imanuvilov,
``Remarks on non-controllability of the heat equation with memory,''
{\it ESAIM: Control, Optim. Calc. Var.}, vol. 19, pp. 288--300, 2013.
\vspace{-5mm}

\bibitem{HaPa:12}
A. Halanay and L. Pandolfi,
``Lack of controllability of the heat equation with memory,''
{\it Syst. Control Lett.}, vol. 61, pp. 999--1002, 2012.
\vspace{-2mm}

\bibitem{IvPa:09}
S. Ivanov and L. Pandolfi,
``Heat equation with memory: Lack of controllability to rest,''
{\it J. Math. Anal. Appl.}, vol. 355, pp. 1--11, 2009.
\vspace{-2mm}

\bibitem{KaMo:2013}
D. Kasinathan and K. Morris,
``$\mathcal{H}_{\infty}$-optimal actuator location,''
{\it IEEE Trans. Autom. Control}, vol. 58, pp. 2522--2535, 2013.
\vspace{-2mm}

\bibitem{Kato}
T. Kato,
{\em Perturbation Theory for Linear Operators},
Classics in Mathematics,
Springer-Verlag, Berlin, 1995.
\vspace{-2mm}

\bibitem{LiZhGu:18}
L. Li, X. Zhou and H. Gao,
``The stability and exponential stabilization of the heat equation with memory,''
{\it J. Math. Anal. Appl.}, vol. 466, pp. 199--214, 2018.
\vspace{-2mm}

\bibitem{Nu:1971}
J. W. Nunziato,
``On heat conduction in materials with memory,''
{\it Quart. Appl. Math.}, vol. 29, pp. 187--204, 1971.
\vspace{-2mm}

\bibitem{Paz:83}
A. Pazy,
{\em Semigroups of Linear Operators and Applications to Partial Differential Equations},
Springer-Verlag, New York, 1983.
\vspace{-2mm}

\bibitem{RaTaTu:2007}
K. Ramdani, T. Takahashi and M. Tucsnak,
``Uniformly exponentially stable approximations for a class of second-order evolution equations: Application to LQR problems,''
{\it ESAIM: Control, Optim. Calc. Var.}, vol. 13, pp. 503--527, 2007.
\vspace{-2mm}

\bibitem{Sh:2010}
R. E. Showalter,
{\em Hilbert Space Methods in Partial Differential Equations},
Dover Publications, New York, 2010.
\vspace{-2mm}

\bibitem{SiAkMiNa:2025}
B. P. K. Sistla, W. Akram, D. Mitra and V. Natarajan,
``LQR-based $\omega$-stabilization of a heat equation with memory,''
{\it IMA J. Math. Control Inform.}, vol. 42, pp. 1--35, 2025.
\vspace{-2mm}

\bibitem{obs_book}
M. Tucsnak and G. Weiss,
{\em Observation and Control for Operator Semigroups},
Birkh\"{a}user, Basel, 2009.
\vspace{-2mm}

\bibitem{Zettl:2005}
A. Zettl,
{\em Sturm--Liouville Theory},
Mathematical Surveys and Monographs, vol. 121,
American Mathematical Society, Providence, RI, 2005.
\vspace{-2mm}

\bibitem{ZhGa:14}
X. Zhou and H. Gao,
``Interior approximate and null controllability of the heat equation with memory,''
{\it Comput. Math. Appl.}, vol. 67, pp. 602--613, 2014.
\vspace{-2mm}

}
\end{thebibliography}
\end{document}